\documentclass[11pt,preprint]{imsart}

\RequirePackage[OT1]{fontenc}
\RequirePackage{amsthm,amsmath}
\RequirePackage{natbib}
\RequirePackage[colorlinks,citecolor=blue,urlcolor=blue]{hyperref}

\usepackage{graphicx}
\usepackage{float}
\usepackage{amssymb}
\usepackage{booktabs, longtable}
\usepackage{threeparttable}

\usepackage{natbib}
\usepackage[inline]{enumitem}

\usepackage{marquis_general}
\usepackage{HD_competing_risks}

\usepackage[dvipsnames]{xcolor}

\usepackage{titlesec}
\titleformat{\section}[block]{\Large\bfseries\filcenter}{\thesection}{1em}{}

\startlocaldefs
\numberwithin{equation}{section}
\theoremstyle{plain}

\endlocaldefs

\begin{document}

\begin{frontmatter}
\title{ \sc Inference under Fine-Gray  competing risks  model  with high-dimensional covariates}
\runtitle{Inference for Competing Risks}

\begin{aug}
\author{\fnms{Jue} \snm{Hou}\ead[label=e1]{j7hou@ucsd.edu}} \ \
\author{\fnms{Jelena} \snm{Bradic} \thanksref{t2}\ \
\ead[label=e2]{jbradic@ucsd.edu}}
\author{\fnms{Ronghui} \snm{Xu}
\ead[label=e3]{rxu@ucsd.edu}}

\address{Department of Mathematics\\
University of California, San Diego\\
\printead{e1,e2}}


\address{Department of Family Medicine and Public Health\\
Department of Mathematics\\
University of California, San Diego\\
\printead{e3}}

\thankstext{t2}{Bradic is supported by the NSF DMS grant number \#1712481 }
\runauthor{J.Hou et al.}

\affiliation{Some University and Another University}

\end{aug}

\begin{abstract}
The purpose of this paper is to construct confidence intervals for the regression coefficients in the Fine-Gray model for competing risks data with random censoring, where the number of
covariates can be larger than the sample size.
Despite strong motivation from biomedical applications,
a high-dimensional Fine-Gray model has attracted relatively little attention among the methodological or theoretical literature.
We fill in this gap by developing confidence intervals based on a one-step bias-correction for a regularized estimation. We develop a theoretical framework for the partial likelihood, which does not have independent and identically distributed entries and therefore presents many technical challenges.
We also study the approximation error from the weighting scheme under random censoring for competing risks and establish new concentration results for time-dependent processes.
In addition to the theoretical results and algorithms, we present extensive numerical experiments and an application to a study of non-cancer mortality among prostate cancer patients using the linked Medicare-SEER data.\end{abstract}


\end{frontmatter}

\section{Introduction}


High-dimensional regression has attracted increasing interest in statistical analysis and has provided a useful tool in modern biomedical, ecological, astrophysical or economics data pertaining to the setting  where the number of parameters is greater than the number of samples (see \cite{buhlmann2011statistics} for an overview).  Regularized methods \citep{tibshirani1996regression, fan2001variable} provide straightforward interpretation of resulting estimators while allowing the number of covariates to be exponentially larger than the sample size. While they can be consistent for prediction (i.e.~estimating the underlying regression function),
 confidence intervals cannot be consistently formulated,
 as firm guarantees of correct variable selection can only be established under a restrictive set of assumptions,
  including but not limited to the assumption of the minimal signal strength of the true parameter  \citep{wasserman2009high,fan2010selective,meinshausen2009lasso},
    which cannot be verified in practice.
For practical purposes, it is of interest to develop inferential tools, most commonly confidence intervals and p-values,  that do not depend on such assumptions and are yet able to provide theoretical guarantees of the quality of estimation and/or testing; and this is the goal of our work here.

For the purposes of constructing confidence intervals or testing significance of the effect from certain covariates, relying on a naive regularized estimation alone is not appropriate;
notably, construction of confidence intervals for those coefficients that have been shrunk to zero is impossible.
\cite{ZhangZhang14} and \cite{vdGeerEtal14}
 proposed the one-step bias-correction estimator, which can be subsequently used to carry out proper statistical inference.
Our work here was motivated by an illustration project of how information contained in patients' electronic medical records can be harvested for precision medicine. The  data set linking the
Surveillance, Epidemiology and End Results  (SEER) Program database of the National Cancer Institute with the federal health insurance program Medicare database contained  prostate cancer patients of age 65 or older.  A total of 57,011 patients diagnosed between 2004 and 2009 had
 information available on 7 relevant clinical variables (age, PSA, Gleason score, AJCC stage, and AJCC stage T, N, M, respectively), 5 demographical variables (race, marital status, metro, registry and year of diagnosis), plus 9321 binary insurance claim codes. Among these patients
 1,247  died due to cancer, and 5,221 had deaths unrelated to cancer by December 2013.
An important goal of the project was to evaluate the impact of risk factors (clinical, demographical, and claim codes) on the non-cancer versus cancer mortality, with appropriate statistical inference.
Cancer and non-cancer versus cancer mortality are known as competing risks in survival analysis, and cannot be handled using linear or generalized linear regression models as considered in \cite{ZhangZhang14} and \cite{vdGeerEtal14}.
Instead, we consider the proportional subdistribution hazards regression model, often known as
the Fine-Gray model \cite{FineGray99}.
Under classical low-dimensional setting, Fine and Gray derived the estimation and inference for the model coefficients via the partial likelihood principle,
and handled right censoring by inverse probability censoring weighting (IPCW).

Considerable research effort has been devoted to developing regularized  methods to handle various regression settings  \citep{ravikumar2010high,belloni2011,obozinski2011support,meinshausen2006high,basu2015regularized,cho2015multiple}, including those for right-censored time-to-event data \cite[among others]{sun2014network,bradic2011regularization,gaiffas2012high,johnson2008variable,lemler2013oracle,bradic2015structured,huang2006regularized}.
However,   regression has  been little studied for the competing risks setting, with random  censoring and high-dimensional covariates.
The purpose of this paper has two folds: 1) to study estimators under the Fine-Gray regression model for competing risks data with many more covariates than the number of events; 2) to develop statistical inference procedures in this setting. To our best knowledge, no published work exists on statistical inference for competing risks data that allows high-dimensional models;    univariate testing was studied in Cox proportional hazards model  -- however, our construction allows for the testing of general linear hypothesis.

 There are at least three significant challenges for addressing high-dimensional competing risks regression under the Fine-Gray model.
  The structure of the score function related to the partial likelihood causes a somewhat subtle issue with many of the unobserved factors  preventing a simple martingale representation. Additionally,  the structure, as well as,  size of the sample information matrix renders both methodology and theoretical analysis based on the Hessian matrix problematic. Thirdly, random censoring presents non-trivial challenges in the presence of competing risks and weighting is needed which further complicates the theoretical analysis.  Also, although bootstrap has been considered for inference under the Fine-Gray regression model, this approach is no longer applicable given the known problems of the bootstrap in high-dimensional settings. Development of high-dimensional inferential methods for competing risks data and under the Fine-Gray model, in particular, may have been hampered by these considerations.

In this paper, we propose a natural and sensible formulation of inferential procedure for this high-dimensional competing risks regression.
 We first study a  regularized estimator of the high-dimensional parameter of interest and derive its finite-sample properties where the interplay between the sparsity, ambient dimension and the sample size can be directly seen.    We then propose a bias-correction procedure by formulating a new pragmatic estimator of the inverse of a large covariance matrix that allows broad dependence structures within the Fine-Gray model.
  We find that the bias-corrected estimator   is effective at capturing strong
as well as weak signals, and can be used for statistical inference.
  This  combination leads to
an efficient and simple-to-implement procedure under the Fine-Gray model
with many covariates.

\subsection{Model and notation}\label{section:model}


For  subject  $i=1, ..., n$ in a study, let $T_i$ be the event time, with the event type or cause $\epsilon_i$; we use the two words interchangeably in the following.
Under the Fine-Gray model that we consider below, we assume without loss of generality that the event type of interest is `1', and we code all the other event types as `2'
without further specification.
In the presence of a potential right-censoring time $C_i$,
the observed time is $X_i = T_i \wedge C_i$.
We denote the event indicator as $\delta_i = I(T_i \le C_i)$.
The  type of the event  $\epsilon_i$ is observed,  if the event occurs before the censoring time, i.e., when $\delta_i=1$.
Let $\bZ_i(\cdot)$ be the vector of covariates that are possibly time-dependent.
We focus on the situation that the dimension of $\bZ_i(\cdot)$, $p$,
is larger than the sample size $n$.
Assume that the observed data $\{(X_i, \delta_i, \delta_i\epsilon_i, \bZ_i(\cdot))\}$ are independent and identically distributed (i.i.d.) for
$i = 1,\dots, n$.

Since the cumulative incidence function (CIF) is often the quantity of interest,
\cite{FineGray99} proposed a proportional subdistribution hazards model where the CIF
\begin{equation}\label{def:CIF}
  F_1(t | \bZ_i(\cdot)) = \Pr(T_i \le t, \epsilon_i = 1 |\bZ_i(\cdot) )
  = 1-\exp\left( -\intt \RR{i}{u}{o} \lamT(u) du\right),
\end{equation}
 the $p$-dimensional coefficient $\betao$
 is the unknown parameter of interest, and $\lamT(t)$ is the baseline subdistribution hazard.
 Under the model (\ref{def:CIF})  corresponding subdistribution hazard $h_1(t | \bZ_i(\cdot)) = \lamT(t) \RR{i}{t}{o} $.
Throughout the paper, we assume
 that there exists  a sparsity factor $\sbeta = |\mbox{supp}(\betao)|$ for some $\sbeta \leq n$.
Note that if we define an improper random variable
$
T^1_i = T_i I(\epsilon_i = 1) + \infty I(\epsilon_i >1)
$,
then the subdistribution hazard can be seen as the conditional hazard of $T^1_i$ given $ \bZ_i(\cdot)$.

We denote the counting process for  type 1 event as
$\NI_i(t) = I(T^1_i \le t)$ and its observed counterpart as
$
\Nobs_i(t) = I(\delta_i\epsilon_i = 1)I(X_i\le t)
$.
We also denote the counting process for the censoring time
as $\Nc_i(t) = I(C_i \le t)$.
Let $Y_i(t)=1-\NI_i(t-)$ (note that this is not the `at risk' indicator like under the classic Cox model), and $r_i(t) = I(C_i \ge T_i \wedge t)$. Note that $r_i(t)Y_i(t) = I(t \le X_i)+I(t>X_i)I(\delta_i\epsilon_i>1)$ is always observable, even though $Y_i(t)$ or $r_i(t)$ may not.
Let $G(t) = \Pr(C_i \ge t)$ and let $\hat{G}(\cdot)$ be the Kaplan-Meier estimator for $G(\cdot)$.
Here we assume that $C$ is independent of $T$, $\epsilon$ and $\bZ$.
Following the notation of Fine and Gray we call  the IPW at-risk process:
\begin{equation}\label{def:ipw}
  \omega_i(t)Y_i(t) = r_i(t) Y_i(t) \frac{\hat{G}(t)}{\hat{G}(t\wedge X_i)};
\end{equation}
in other words, the weight for subject $i$ is one if $t<X_i$, zero after being censored or failure due to cause 1, and $ \hat{G}(t) / \hat{G}( X_i)$ after failure due to other causes.
The log pseudo likelihood (as recently named in \cite[]{BellachEtal18JASA}) that gives rise to the weighted score function in \cite{FineGray99}
 for $\bbeta$  is
\begin{equation}\label{def:pl}
  m(\bbeta) = n^{-1} \sumin \inttao \left\{ \bbeta^\top \bZ_i(t) - \log\left(\sumjn \omega_j(t)Y_j(t)\RR{j}{t}{}\right) \right\} d\Nobs_i(t).
\end{equation}
where $t^*<\infty$ is the study end time.

In the following, for a vector $\bfv$, let
$\bfv^{\otimes 0} = 1$, $\bfv^{\otimes 1} = \bfv$ and $\bfv^{\otimes 2} = \bfv \bfv^\top$.
We define for $l=0,1,2$
\begin{gather}
\bs{l}(t, \bbeta) =  \E\left\{ G(t)/{G}(t\wedge X_i) r_i(t) Y_i(t)\RR{i}{t}{}\bZ_i(t)^{\otimes l}\right\}, \quad
\bmu(t) =  \bs{1}(t, \betao)/\bso(t, \betao), \notag \\
\bS{l}(t, \bbeta) =  n^{-1}\sumin \omega_i(t)Y_i(t)\RR{i}{t}{}\bZ_i(t)^{\otimes l}, \quad
\bZbar(t, \bbeta) =  \bS{1}(t, \bbeta)/\bSo(t, \bbeta). \label{def:Sk}
\end{gather}
We then have 
the score function, i.e.~derivative of the log pseudo likelihood \eqref{def:pl}, 
$$
\dm(\bbeta) = n^{-1}\sumin \inttao \{\bZ_i(t)-\bZbar(t,\bbeta)\} d\Nobs_i(t).
$$

Regarding notation, let us mention that all constants are assumed to be independent of $n$, $p$ and $\sbeta$.
We use $K$ and $\rho$ with corresponding enumerated subscripts to denote ``big" and ``small" constants.
We use $Q$ to denote intermediate terms used in the statements or the proofs of various results.
We label the subscripts by the corresponding order of their appearance.

\subsection{Organization of the paper}

This paper is organized as follows. In Section \ref{section:method}, we provide the bias corrected estimator,  Section \ref{section:method one-step}, as well as the confidence interval construction,   Section \ref{section:method CI},   for the high-dimensional Fine-Gray model. Construction of a new Hessian estimator, the cornerstone for our bias correction and variance estimation, is presented in Section \ref{section:method nodewise}. Section \ref{section:theory} presents properties of the developed estimator. Additional notations for theoretical considerations are presented in Section \ref{section:notation}. Bounds for the prediction error  are presented in Section \ref{section:theory-init}; Theorem \ref{thm:initial oracle} is the main result on estimation. Section \ref{section:theory-inf} studies the sampling distribution of a newly develop test statistics while not requiring model selection consistency or minimal signal strength. Theorem \ref{thm:normality} is the main result regarding asymptotic distribution. There we present a sequence of  intermediate  results as well. We examine our regularized estimator and the one-step bias-corrected estimator through simulation experiments in Section \ref{section:simulation}, and apply them to the  SEER-Medicare data in Section \ref{section:data analysis}.

\section{Estimation and inference for competing risks with more regressors than events}\label{section:method}

\subsection{One-step corrected estimator}\label{section:method one-step}


 A natural initial estimator to consider under the high dimensional setting is a  $l_1$-regularized estimator, where the particular loss function of interest would be  the log pseudo likelihood as defined in \eqref{def:pl}.
 We note that our results are easily generalizable to any sparsity-inducing and convex penalty functions, but due to  the simplicity of presentation we present details only for the $l_1$ regularization.
 That is, we consider
 \begin{equation}\label{def:bini}
\bini(\lambda)=\argmin_{\bbeta \in \R^p} \Bigl\{-m(\bbeta) + \lambda\|\bbeta\|_1\Bigl\}
\end{equation}
 for a suitable choice of the tuning parameter $\lambda>0$.
 Whenever possible, we suppress $\lambda$ in the notation above and use $\bini$ to denote the  $l_1$-regularized estimator. In the Section \ref{section:theory-init}, we quantify the non-asymptotic oracle risk bound and show that the estimator above, as a typical regularized estimator with $p \gg n$, converges at a rate slower than root-$n$. This implies that for inferential purposes the bias of the estimator cannot be ignored.

Inspired by the work of \cite{ZhangZhang14} and \cite{vdGeerEtal14},
we propose the one-step bias-correction estimator
\begin{equation}\label{def:bCI}
\bCI := \bini + \hTh \dm(\bini),
\end{equation}
where $\bini$ is  defined in \eqref{def:bini}, $\hTh$ is an estimator of the ``asymptotic" precision matrix $\boldsymbol\Theta$ to be defined later.
The above construction of the one-step estimator is inspired by the first order Taylor expansion of   $\dm(\cdot)$,
\begin{align}\label{eq:approx idea}
\dm(\betao) &\approx \dm(\bini) - \ddm(\betao)(\bini-\betao) \nonumber
\\
&\approx  \ddm(\betao) \left[ \betao - \{\bini  + \hTh\dm(\bini)\} \right]
= \ddm(\betao)\{ \betao - \bCI\}.
\end{align}
The notation ``$\approx$" in the above indicates that the equivalence is approximate with the higher order error terms omitted. We aim to find a good candidate matrix $\hTh$, such that $-\ddm(\betao)\hTh \approx \mathbb{I}_p$, with $\mathbb{I}_p$ denoting the $p \times p$ identity matrix.
Note that when $p \leq n$ an inverse of the Hessian matrix above would naturally be a good candidate for $\hTh$, 
but when $p \geq n$ such an inverse does not necessarily exist.
We will elucidate the  construction of   $\hTh$ towards the end of this section.

\subsection{Confidence Intervals}\label{section:method CI}

To construct the confidence intervals for components of $\betao$, we need the asymptotic distribution of $\bCI$.
We will first establish the asymptotic distribution of the score $\dm(\betao)$.
With $p>n$, we have to restrict the space in which we want to establish the asymptotic distribution.
The asymptotic distribution for  $\dm(\betao)$ is established in the following sense
 --- for any $\bfc \in \R^p$ such that $\|\bfc\|_1 = 1$ we have
$$
\sqrt{n} \bfc^\top \dm(\betao) \stackrel{d}{\to} N(0, \bfc^\top \Scov \bfc),
$$
where $\Scov$ is the variance-covariance matrix for $\sqrt{n} \dm(\betao)$.
We construct the following estimator for $\Scov$:
\begin{equation}\label{def:hScov}
  \hScov =n^{-1} \sumin (\heta_i + \hpsi_i)^{\otimes 2},
\end{equation}
where $\hat \eta_i$ and $\hat \psi_i$ are defined as follows:
\begin{align}
 &\heta_i  = \inttao \{\bZ_i(t) - \bZbar(t, \bini)\} \omega_i(t)d\MIh_i(t), \label{def:heta}\\
 &\hpsi_i  = \inttao \frac{\hbfq(t)}{\hpi(t)} d\Mch_i(t), \label{def:hpsi} \\
 & \hbfq(t) =  n^{-1}\sumin I(t>X_i)\int_t^{t^*} \{\bZ_i(u)-\bZbar(u,\bini)\}\omega_i(u)d\MIh_i(u), \label{def:hbfq} \\
&   \hpi(t) =  n^{-1}\sumin I(X_i \ge t), \label{def:hpi} \\
& d\MIh_i(t) =  d\Nobs_i(t) - \frac{\omega_i(t)Y_i(t)\RRh{i}{t}}{\bSo(t,\bini)} n^{-1} \sumjn d\Nobs_j(t), \label{def:MIh} \\
& d\Mch_i(t) = I(X_i \ge t)d\Nc_i(t) - \frac{I(X_i \ge t)}{\hpi(t)} n^{-1} \sumjn I(X_j \ge t)d\Nc_j(t). \label{def:Mch}
\end{align}

As illustrated in \eqref{eq:approx idea}, we have
$\sqrt{n}\bfc^\top(\bCI-\betao)$ to be asymptotically equivalent to
$$
\sqrt{n} \bfc^\top \bTh \dm(\betao)
 \stackrel{d}{\to}  N(0, \bfc^\top \bTh \Scov \bTh^\top \bfc).
$$
We may now estimate the variance of $\sqrt{n}\bfc^\top(\bCI-\betao)$
 using a ``sandwich" estimator
$
\bfc^\top \hTh \hScov \hTh^\top \bfc
$.
Therefore a $(1-\alpha)100\%$ confidence interval for $\bfc^\top \betao$ is
\begin{equation}\label{eq:CI}
\left[\bfc^\top \bCI -\Zcal_{1-\alpha/2}\sqrt{\bfc^\top \hTh \hScov \hTh^\top \bfc/n}, \;
\bfc^\top \bCI +\Zcal_{1-\alpha/2}\sqrt{\bfc^\top \hTh \hScov \hTh^\top \bfc/n}\right]
\end{equation}
with standard normal quantile $\Zcal_{1- \alpha/2}$.

Our proposed approach addresses various practical questions as special cases.
First, we can construct confidence interval for a chosen coordinate $\beta^o_j$ in $\betao$. To that end, one needs to consider  $\bfc = \bfe_j$, the $j$-th natural basis for $\R^p$ and apply the result \eqref{eq:CI}.
Generally, we can construct a confidence interval for any linear contrasts $\bfc^\top \betao$, potentially of any dimension.
For example, we can have confidence intervals for the linear predictors $\bZ^\top \betao$ if the non-time-dependent covariate $\bZ$ is also sparse so that we may assume $\|\bZ\|_1$ to be bounded.
As the dual problem, we may use the Wald test statistic
\begin{equation}\label{def:test}
  Z = \sqrt{n}(\bfc^\top\bCI-\theta_0)/ \sqrt{\bfc^\top \hTh \hScov \hTh^\top \bfc}
\end{equation}
to test the hypothesis with $H_0: \bfc^\top \betao = \theta_0$.


\subsection{Construction of  the inverse Hessian matrix}\label{section:method nodewise}

Although the early works under the linear model inspire the construction here,  the specifics, as well as the theoretical analysis, the latter remains a challenge.
We start by writing   the negative Hessian of the log pseudolikelihood \eqref{def:pl}:
\begin{equation}\label{def:Hess}
  -\ddm(\bbeta) = n^{-1}\sumin \inttao
  \left\{\frac{\bS{2}(t,\bbeta)}{\bSo(t,\bbeta)} - \bZbar(t,\bbeta)^{\otimes 2}\right\} d\Nobs_i(t).
\end{equation}
We define 
\begin{equation}\label{def:bSig}
  \bSig =  \E \left[ \inttao \left\{\bZ_i(t) - \bmu(t)\right\}^{\otimes 2}
  d\Nobs_i(t) \right] = \E \left[ \inttao \left\{\bZ_i(t) - \bmu(t)\right\}
  d\Nobs_i(t)\right]^{\otimes 2}.
\end{equation}
Under the regularity conditions, to be specified later, we have $\bSig$ as the ``asymptotic negative Hessian" in the sense that the element-wise maximal norm
$
\|\bSig + \ddm(\betao)\|_{\max}
$
converges to zero in probability.
Our goal is to estimate its inverse $\bTh = \bSig^{-1} = (\bth_1,\dots, \bth_p)^\top$, where $\bth_j$'s are the rows of $\bTh$.

By definition \eqref{def:bSig},
the positive semi-definite matrix $\bSig$ is also the second moment of the random vector
\begin{equation}\label{def:bfXi}
\bfXi_i = \inttao \left\{\bZ_i(t) - \bmu(t)\right\}
  d\Nobs_i(t)
\end{equation}
with $\bmu(t)$ defined in \eqref{def:Sk}.
The expectation of $\bfXi_i$ is zero,
$$
\E(\bfXi_i) = \E\left[ \inttao \left\{\bZ_i(t) - \bmu(t)\right\}
  Y_i(t)I(C_i \ge t) \RR{i}{o}{t}\lamT(t)dt\right] = \bzero.
$$
Hence, to estimate $\bTh$, we may draw inspiration from the early work on inverting the high-dimensional variance-covariance matrix \cite[]{ZhouEtal11}.
Consider the minimizers of the expected loss functions 
\begin{equation}\label{def:grotau}
\gro_j = \argmin_{\bfgr_j \in \R^p}\E (\Xi_j - \bfXi_{-j}^\top \bfgr_j)^2, \quad \tau_j^2 = \E (\Xi_j - \bfXi_{-j}^\top \gro_j)^2,
\end{equation}
where $\Xi_j$ is the $j$th element of $\bfXi$, and $\bfXi_{-j}$ is a $p-1$ dimensional vector created by dropping the $j$th element from $\bfXi$.
We show that the quantities $ \gro_j $ and $ \tau_j $ defined in \eqref{def:grotau} can be used to construct the inverse of $\Hess$.
This is because 
$\tau_j^2$ can also be alternatively written as
\begin{equation}\label{eq:tau-alt}
\E \{(\Xi_j - \bfXi_{-j}^\top \gro_j)\Xi_j\} - \grot_j\E\{ (\Xi_j - \bfXi_{-j}^\top \gro_j)\bfXi_{-j}\}.
\end{equation}
By the convexity of the target function $\E (\Xi_j - \bfXi_{-j}^\top \bfgr_j)^2$,  $\gro_j$ must satisfy the first order Karush-Kuhn-Tucker  conditions (KKT)
\begin{equation}\label{eq:tau-KKT}
  - \grot_j\E \{(\Xi_j - \bfXi_{-j}^\top \gro_j)\bfXi_{-j}\} = 0.
\end{equation}
Applying \eqref{eq:tau-KKT} to \eqref{eq:tau-alt}, we have
$$
\tau_j^2 = \E \{(\Xi_j - \bfXi_{-j}^\top \gro_j)\Xi_j\}.
$$
We can then define a vector $\bth_1 = (1, -\grot_1)^\top / \tau_1^2$ that satisfies
$$
\bth_1^\top \Hess = \E \{(\Xi_1 - \bfXi_{-1}^\top \grot_1)\bfXi \}/\E \{(\Xi_1 - \bfXi_{-1}^\top \gro_1)\Xi_1\} = (1, \bzero_{p-1}) = \mathbf{e}_1.
$$
Without loss of generality, we may define $\bth_j$ accordingly for $j=2,\dots,p$, satisfying
$
\bth_j^\top \Hess = \mathbf{e}_j
$.
The matrix $\bTh = (\bth_1,\dots, \bth_p)^\top$ satisfies
$$
\bTh \Hess = (\mathbf{e}_1,\dots, \mathbf{e}_p) = \mathbb{I}_p,
$$
therefore $\bTh$ is the inverse of $\Hess$.
We now utilize the sample form of  $\boldsymbol \Sigma$,
  \eqref{def:bSig},
\begin{equation}\label{def:hHess}
\hHess = n^{-1}\sumin \inttao \{\bZ_i(t)-\bZbar(t,\bini)\}^{\otimes 2} d\Nobs_i(t).
\end{equation}
In particular we observe that
 $\hHess$ is that it can be written  as the sample second moment
$
\hHess  = n^{-1} \sumin \hXi_i ^{\otimes 2}
$
where
\begin{equation}\label{def:hXi}
  \hXi_i = \inttao \{\bZ_i(t)-\bZbar(t,\bini)\} d\Nobs_i(t).
\end{equation}
This form allows us to define the inverse of $\boldsymbol\Sigma$ as a regression between the vectors $  \hXi_i $. For that purpose we define the least squares loss function as

\begin{equation}\label{def:Gr}
  \Gr_j(\bfgr_j, \bini) = n^{-1}\sumin \left(\hat{\Xi}_{i,j} - \hXi_{i,-j}^\top \bfgr_j\right)^2, \ j=1,\dots,p,
\end{equation}
where $\hat{\Xi}_{i,j}$ is the $j$th element of $\hXi_i$, and $\hXi_{i,-j}$ is a $p-1$ dimensional vector obtained by dropping the $j$th element from $\hXi_i$.
We then define the nodewise LASSO in our context to be
\begin{equation}\label{def:hgrtau}
\hgr_j = \argmin_ {\bfgr_j \in \R^{p-1}} \left\{\Gr_j(\bfgr_j, \bini) + 2\lambda_j \|\bfgr_j\|_1 \right\}, \qquad
\hat{\tau}_j^2 =  \Gr_j(\hgr_j, \bini) + \lambda_j \|\hgr_j\|_1.
\end{equation}
Accordingly, we use $\hgr_j$ and $\hat{\tau}_j^2$ to construct
\begin{equation}\label{def:hTh}
  \hat{\boldsymbol\Theta}_{jk} = \left\{
  \begin{array}{cc}
    -\hat{\gamma}_{j,k}/(\hat{\tau}^2_j), &  k < j; \\
    1/(\hat{\tau}^2_j), &  k = j; \\
    \hat{\gamma}_{j,k-1}/(\hat{\tau}^2_j), &  k > j.
  \end{array}
  \right.
\end{equation}
By the first order KKT condition, we have $(\hat{\boldsymbol\Theta} \hHess)_{j,j}=1$ and  $|(\hat{\boldsymbol\Theta} \hHess)_{j,k}|\le \lambda_j$ for
$j\neq k$.
Choosing $\lambda_{\max} = \maxj \lambda_j =o_p(1)$, we achieve that $\|\hTh\hHess - \mathbb{I}_p\|_{\max}$ goes to zero. 
The one-step estimator proposed in \eqref{def:bCI} with such $\hTh$ hence converges to the true coefficient $\betao$ approximately at the rate equivalent to $\dm(\betao)$, as
illustrated in \eqref{eq:approx idea}.

Our proposed nodewise LASSO estimator is innovative in several aspects.
Given the difficulty imposed by the model,
we cannot make high-dimensional inference by simply inverting the $XX^\top$ for a design matrix $X$ like in a linear or generalized linear model.
The log pseudo likelihood \eqref{def:pl} has dependent entries.
The covariates $\bZ_i(t)$ for $i=1,\dots,n$ are allowed to be time-dependent.
Nevertheless, we identify for our model that the key element for the high-dimensional inference is each observation's contribution to the score, the $\bfXi_i$'s.
Our solution generalizes high-dimensional matrix inversion in a non-trivial way to complex models with censoring, non-standard likelihoods and weighting.

\section{Theoretical considerations}\label{section:theory}

 In this section,  we present the theory for the estimators $\bini$, $\bCI$ and the confidence intervals described in the previous section.
 We will quantify the non-asymptotic oracle risk bound for the estimator above while allowing $p \gg n$ with a minimal set of assumptions. Theoretical study of this kind is  novel, since in the context of competing risks, the  martingale structures typically utilized  are unavailable and new techniques need to be developed. In particular, we show that the  inverse probability weighting
has a finite-sample effect that separates this model from  the classical Cox model (see comments after Theorem \ref{thm:initial oracle}).  We will also establish that a certain tighter bound can be established whenever the hazard rate is bounded (Theorem \ref{thm*:initial rate}).

Throughout our work we assume that   $\{(T_i, C_i, \epsilon_i, \bZ_i(t)): t \in [0, \infty) \}$ are i.i.d. with
      $C_i$ independent of $(T_i, \epsilon_i, \bZ_i(\cdot))$. Moreover, for any $t \in [0,t^*]$, $G(t) = I(C_i \ge t)$ is differentiable,
       and its hazard function $\lamC(t) = -G'(t)/G(t) \le \Lc$.
We also assume that the baseline CIF $F_1(t; \bzero)$ is differentiable.
  The baseline subdistribution hazard $\lamT(t) = -d\log\{F_1(t; \bzero)\}/dt \in [\lamlow, \Llam]$ for all $t \in (0,t^*)$  and some  $\lamlow>0$ and $\Llam<\infty$.


\subsection{Additional notation}\label{section:notation}
%

In the following, we  introduce some additional notations. The counting process martingales
\begin{equation}\label{def:MI}
  \MI_i(t)  = \NI_i(t) - \intt Y_i(u) \RR{i}{u}{o} \lamT(u) du
\end{equation}
are essentially helpful tools in high-dimensions for establishing theory with dependent partial likelihoods.
Unfortunately, the uncensored counting processes $\{\NI_i(t), i=1,\dots,n\}$ are not always observable.
The observable counterpart $\Nobs_i(t)$ has
no known martingale related to it under the observed filtration $\Ft=\sigma\{\Nobs_i(u), I(X_i \ge u), r_i(u): u \le t, i= 1,\dots, n\}$.
The Doob-Meyer compensator for the submartingale $\Nobs_i(t)$ under the observed filtration involves the nuisance distribution of $T_i|\epsilon_i>1$.
To utilize the martingale structure for our theory, we have to define the ``censoring complete" filtration
\begin{equation}\label{def:ccFt}
  \ccFt = \sigma \{\Nobs_i(u), I(C_i \ge u), \bZ_i(\cdot): u \le t, i = 1,\dots, n\},
\end{equation}
on which we have a martingale related to $ N_i^o(t)$,
\begin{equation}\label{def:ccFtM}
  \intt I(C_i \ge t)d \MI_i(u) = \Nobs_i(t)  - \intt I(C_i \ge u)Y_i(u) \RR{i}{u}{o} \lamT(u) du.
\end{equation}
To relate the martingale \eqref{def:ccFtM} with our log pseudo likelihood $m(\bbeta)$,
we define its proxy with $\ccFt$ measurable integrand
\begin{equation}\label{def:tm}
  \tilde{m}(\bbeta) =  n^{-1} \sumin \inttao \bbeta^\top \bZ_i(t) - \log\left(\sumjn I(C_j \ge t)Y_j(t)\RR{j}{t}{}\right)d\Nobs_i(t).
\end{equation}
We define processes related to $\tilde{m}(\bbeta)$ and its derivatives as
\begin{equation}\label{def:Stk}
\btS{l}(t, \bbeta) =  n^{-1}\sumin I(C_i \ge t)Y_i(t)\RR{i}{t}{}\bZ_i(t)^{\otimes l}, \quad
\bZtil(t, \bbeta) =  \btS{1}(t, \bbeta)/\btSo(t, \bbeta).
\end{equation}
They can also be seen as proxies to the processes in \eqref{def:Sk}.
To see that, observe that by conditioning,
\begin{align*}
\E\left\{\btS{l}(t,\bbeta)\right\}& =
\E\left[\E\{ I(C_i \ge t)Y_i(t)|\Ft\}\RR{i}{t}{}\bZ_i(t)^{\otimes l}\right]
\\
&= \E\left\{\tilde{\omega}_i(t)Y_i(t)\RR{i}{t}{} \bZ_i(t)^{\otimes 2}\right\},
\end{align*}
where
\begin{equation}\label{def:wtil}
\tilde{\omega}_i(t) = r_i(t)G(t)/G(t\wedge X_i)
\end{equation}
 is the weight with the true censoring distribution $G(\cdot)$.
We denote their expectations as
\begin{equation}\label{eq:ES}
  \bs{l}(t,\bbeta) = \E\left\{\btS{l}(t,\bbeta)\right\} = \E\left\{\tilde{\omega}_i(t)Y_i(t)\RR{i}{t}{} \bZ_i(t)^{\otimes 2}\right\}.
\end{equation}
Our proxies precisely target those weighted samples, as $\btS{l}(t, \bbeta)$ differs from $\bS{l}(t, \bbeta)$ only at those summands with observed type-2 events.

Note that the  Kaplan-Meier estimator for $G(t)$ can be written as
$$
\hat{G}(t) = \prod_{u\le t} \left(1 - \frac{d\Nc_i(u)}{I(X_i \ge u)}\right).
$$
To study the convergence of $\hat{G}(t)$ to $G(t)$,
we denote a martingale related to $\Nc_i(t)$, the counting process of observed censoring, $  \Mc_i(t) $.
Let the censoring hazard be defined as $\lamC(t) = -d\log(G(t))/dt$.
Under the ``censoring" filtration
\begin{equation}\label{def:cFt}
  \cFt = \sigma\{\Nc_i(u), T_i, \epsilon_i, \bZ_i(\cdot): u\le t, i = 1,\dots, n\},
\end{equation}
we have a martingale
\begin{equation}\label{def:Mc}
  \Mc_i(t) = \Nc_i(t)  - \intt I(C_i \ge u)\lamC(u) du.
\end{equation}

We use the integration-by-parts arguments \cite[the Helly-Bray argument on page 727]{Murphy94} with random martingale measures, e.g. $d\MI_i(t)$, in our proof.
The total variation of $\MI_i(t; w)$ is defined as
\begin{equation}\label{def:totVar}
 \bigvee_0^{t^*}\MI_i(t; w) = \sup_{k=1,2,\dots}\sup_{0\le t_1<\dots<t_k \le t^*} \sum_{j=2}^n |\MI_i(t_j; w)-\MI_i(t_{j-1}; w)|.
\end{equation}
Since $\MI_i(t; w)$ can be decomposed into a nondecreasing counting process $\NI_i(t)$ minus another nondecreasing compensator $\intt Y_i(u)\RR{i}{u}{o}\lamT(u)du$,
we have a bound for its total variation
\begin{equation}\label{eq:totVarMI}
  \bigvee_0^{t^*}\MI_i(t; w) \le \NI_i(t^*) + \inttao Y_i(u)\RR{i}{u}{o}\lamT(u)du.
\end{equation}
Similar conclusion also applies to $\Mc_i(t)$, i.e.~we have a bound for its total variation
\begin{equation}\label{eq:totVarMc}
  \bigvee_0^{t^*}\Mc_i(t; w) \le \Nc_i(t^*) + \inttao I(C_i \ge t)\lamC(u)du.
\end{equation}
As a convention, from hereon we suppress the $w$ in the notation to keep it simple.

\subsection{Oracle inequality }\label{section:theory-init}

We first establish  oracle inequality for the initial estimation error $\|\bini-\betao\|_1$ based on the following  set of
conditions that are weaker than those in the next subsection.


\begin{enumerate}[label=(C\arabic*)]
  \item \label{Assum:design}({\bf Design})
      With probability equal to one, the covariates satisfy
      \begin{equation}\label{aseq:Zij}
      \supi\supt\|\bZ_i(t)\|_{\infty} \leq \KZ/2.
      \end{equation}
       The expected at-risk process is bounded away from zero, i.e.,
    for positive  $\M$ and $\rstar$
    \begin{equation}\label{aseq:denom}
    \inft \E \left[I(C_i \ge t^*)I(t^* < T^1_i < \infty)\min\{\M,\RR{i}{t}{o}\}\right] > \rstar.
    \end{equation}

    \item ({\bf Covariance}) \label{Assum:Hess} For    $\M$ in \eqref{aseq:denom}, the smallest eigenvalue of the matrix
  $$
  \Hess(\M) = \E \left\{ \inttao \left(\bZ(t) - \frac{\E \left[\bZ(t) \{1 - F_1(t;\bZ)\}\min\{\M,\RR{}{t}{o}\}\right]}{\E \left[\{1 - F_1(t;\bZ)\}\min\{\M,\RR{}{t}{o}\}\right]} \right)^{\otimes 2}\lamT(t)dt \right\}
  $$
  is at least $\rhosig > 0$.

  \item \label{Assum:CR}({\bf Continuity})
  $\bZ_i(t)$ may have $\Kzi$ jumps at
  $t_{i,1} < t_{i,2}< \dots < t_{i, \Kzi}$ with minimal gap between jumps bounded away from zero,
  $$\min_{i = 1,\dots,n}\min_{1<k\le \Kzi} t_{i,k}-t_{i, k+1} \ge \dt.$$
  Between two consecutive jumps, $\bZ_i(t)$ has at most $\rtz$ elements Lipschitz continuous with Lipschitz constant $\Lz$ while  the rest of the elements are considered to be constant.

\end{enumerate}

\underline{\it Remark 1}. Overall, the conditions above are minimal in the sense that they appear in results pertaining to the  Cox model \cite[see e.g, (3.9) on page 1149; (4.5) and Theorem 4.1 on page 1154]{HuangEtal13}.

\underline{\it Remark 2}. We consider a finite interval $[0,t^*]$. Due to missing censoring times among those with observed type-2 events, we have to make the additional assumptions to control the weighting errors.
Although the weighted at-risk processes $\omega_i(t)$'s are asymptotically unbiased, the approximation errors in the tail $t \to \infty$ are poor  for any finite $n$.
To avoid unnecessary complications, we set the $[0,t^*]$ such that we always have sufficient at-risk subjects; note that \eqref{aseq:denom} implies that $P(C>t^*)>0$.

\underline{\it Remark 3}. We assume a finite maximal norm of $\bZ(t)$. Condition \eqref{aseq:Zij} in \ref{Assum:design} is equivalent to the apparently weaker assumption  (see for example \cite{HuangEtal13} equation (3.9)):
\begin{equation}\label{aseq:Zij*}
\supij \supt\|\bZ_i(t) - \bZ_j(t)\|_\infty \le \KZ.
\end{equation}
This can be seen by noting that the Cox type partial likelihood for the proportional hazards model is invariant when subtracting $\bZ_i(t)$ by any deterministic  $\bzeta(t)$.

\underline{\it Remark 4}.  Condition \ref{Assum:design} \eqref{aseq:denom} carries two facts.
First, the at-risk rate for type 1 events
is bounded away from zero.
Second,  relative-risks arbitrarily close to zero
is  truncated at a finite $\M$; this is necessary in high-dimensions, in order to rule out
 the irregular cases where the non-zero expectation of the relative risk   is dominated by a diminishing proportion of the excessively large relative risks.
The same argument applies for \ref{Assum:Hess} in which a lower bound of the restricted eigenvalue of the negative Hessian \cite[]{BickelEtal09} is defined.

\underline{\it Remark 5}. We assume the smoothness of the time-dependent covariates $\bZ(t)$. Subjects with observed type 2 events, remain indefinitely in the risk sets for type 1 events.
         For time-dependent covariates, continuity is helpful in establishing a slow growing rate of the maximal relative risks among those subjects; something that is a fact for time independent covariates.
Note that the coordinate wise continuity in  $\bZ_i(t)$ is insufficient as $p$ grows to infinity.
We propose \ref{Assum:CR} taking into account  likely practical scenarios,
where the  covariates are either constant, or change  only at finitely many discrete time points.

Under the above assumptions, we are ready to present our estimation error result.
Since the result holds in finite samples, we  define a sequence of  important constants first. For a $\varepsilon>0$ and constants $K_1,\cdots, K_7$ as well as $\rho_1,\cdots, \rho_4$ (introduced in the conditions above)
\begin{equation}\label{def:Ke}
  \Ke  =  e^{\rtz \Lz \|\betao\|_\infty \dt} \log(n/\varepsilon)/\dt\lamlow, \\
\end{equation}
  \begin{equation}\label{def:Csk}
 \CSk{l}  =\frac{\Ke \KZ^l}{2^l}\left\{\frac{4\M^2(1+\Lc t^*)}{\rstar^2}\sqrt{\frac{4\log(2/\varepsilon)}{n}}+ \frac{4\M^2\Lc t^*}{\rstar^2n}
 + \sqrt{\frac{2\log(2np^l/\varepsilon)}{n}}+ \frac{1}{n}\right\},
 \end{equation}
 where $l=0,1$, and
\begin{equation}\label{def:dC}
 \dC = \left\{2\CSk{1}+\KZ\CSk{0}\right\}/\rstar+ \KZ\sqrt{2\log(2p/\varepsilon)/n}.
\end{equation}
In high-dimensional models  an additional constant, the so called compatibility factor, plays an important role.  For a positive constant $\xi>1$, the compatibility factor
 \begin{equation}\label{def:kappab}
  \compb = \sup_{0\neq \bfb \in \coneb} \frac{\sqrt{\sbeta \bfb^\top \{-\ddm(\betao)\bfb\}}}{\|\bfb_\Ocal\|_1}
\end{equation}
where $\coneb$ denotes the cone set
\begin{equation*}
  \coneb = \{\bfb \in \R^p : \|\bfb_{\Ocal_c}\|_1 \le \xi\|\bfb_\Ocal\|_1\},
\end{equation*}
with $\Ocal$ denoting  the indices   of non-zero elements  $\betao$
and $\Ocal_c$  denoting its  compliment.

\begin{theorem}\label{thm:initial oracle}
For $\xi>1$ and a $\varepsilon>0$,
let
$$\lambda=\dC (\xi-1)/(\xi+1)$$ with $\dC$ defined in \eqref{def:dC}.
When $n > -\log(\varepsilon/3)/(2\rstar^2)$ with $\rstar$ given in \ref{Assum:design},
we have under regularity conditions \ref{Assum:design} and \ref{Assum:CR}
that
$$
    \|\bini-\betao\|_1<\frac{e^\eta(\xi+1)\sbeta\lambda}{2\Ckap^2}
$$
occurs with probability no less than
    $$\Pr\big(\compb >\Ckap\big) - e^{-n \rstar^2/(2\M^2)} - n e^{-n(\rstar-2/n)^2/(8\M^2)}- 5\varepsilon,$$
where $\Ckap$ is a positive constant satisfying
$$2\KZ(\xi+1)s_o\lambda /(2\Ckap)^2 \le 1/e$$ and
$\eta$ is the smaller solution of $\eta e^{-\eta}=2\KZ(\xi+1)s_o\lambda /(2\Ckap)^2$.
\end{theorem}
Our proof of Theorem \ref{thm:initial oracle} applies to the result with $l_2$-norm and general $l_q$-norm for $q \ge 1$. Namely, under the same conditions we have that
$$
\|\bini-\betao\|_q<\frac{2e^\eta\xi\sbeta^{1/q}\lambda}{(\xi+1)\Ckap}
$$
occurs with probability no less than
    $$\Pr\big(F_q(\xi,\Ocal) >\Ckap\big) - e^{-n \rstar^2/(2\M^2)} - n e^{-n(\rstar-2/n)^2/(8\M^2)}- 5\varepsilon,$$
with the  weak cone invertibility condition defined as
$$
F_q(\xi,\Ocal) = \sup_{0\neq \bfb \in \coneb} \frac{-\sbeta^{1/q} \bfb^\top \ddm(\betao)\bfb}{\|\bfb_\Ocal\|_1\|\bfb\|_q}.
$$
\vskip 10pt

A few comments are in order. For a fixed $\varepsilon$, $\dC$ is of order $\log(n)\sqrt{\log(p)/n}$.
Thus,  Theorem \ref{thm:initial oracle},  together with  Lemma  \ref{lemma:kappa} (see below),
guarantee that
for  $\lambda $ chosen to be of the order $  \log(n)\sqrt{\log(p)/n}$
$$\|\bini-\betao\|_1 = O_p\left(\sbeta \log(n)\sqrt{\log(p)/n}\right).$$
The above estimation error rate  to the error rate $\sqrt{\log(p)/n}$ of the simple Cox   model  \citep[]{HuangEtal13, 2018arXiv180301150Y}, differing only by a factor of $\log(n)$. This factor  is
brought in by the error 
induced by the IPCW weights.
Therefore, under the rate condition
$\sbeta \log(n)\sqrt{\log(p)/n} = o(1)$, we obtain an asymptotically $l_1$-consistent regularized estimator $\bini$.

The quantity $\Ke$ describes the error from IPCW weights through the measurable approximation to processes $\bS{l}$, $\bS{l}(t,\betao) - \btS{l}(t,\betao)$.
A na\"{i}ve bound for the measurable approximation
is proportional to the magnitude of the relative risks in $\bS{l}$,
naturally of the order $e^{\|\betao\|_1\KZ} \asymp e^{\sbeta}$, potentially
growing in exponential rate of $n$ if $\sbeta \asymp n^a$ for some $a>0$.
Such bound grows way too rapidly to deliver any meaningful result.
Observing that the summands in $\bS{l}$ and $\btS{l}$ at a particular index $i$ differ from each other only when the $i$-th subject has type-2 event we are able to establish a significantly sharper  bound. For that purpose, we develop
$\varepsilon$-tail bound of the maximal relative risk among observed type 2 events
(see  Appendix Lemma \ref{lemma:RR}).
 The quantity $\CSk{l}$, involving $\Ke$ directly in the definition, gives the bound for the error from the measurable approximation to $\bS{l}$ 
  (See in Appendix Lemma \ref{lemma:Sk}).

For the rest of this section, we provide further details on the proof of Theorem \ref{thm:initial oracle}, as well as the technical challenges involved.
We highlight two results, Lemma \ref{lemma:score oracle} and \ref{lemma:kappa}. The first establishes properties of the score vector while the second one establishes the properties of the compatibility factor \eqref{def:kappab}.

\begin{lemma}\label{lemma:score oracle}
Let
$\dC$ be defined as in \eqref{def:dC}.
    Under Assumptions \ref{Assum:design} and \ref{Assum:CR},
    $$
    \Pr\Bigl(\| \dm(\betao)  \|_\infty < \dC \Bigl)
    \ge 1 - e^{-n \rstar^2/(2\M^2)} - n e^{-n(\rstar-2/n)^2/(8\M^2)} - 5\epsilon.
    $$
\end{lemma}

 Lemma \ref{lemma:score oracle}  establishes that such event $\{\|\dm(\betao)\|_\infty < \lambda (\xi-1)/(\xi+1)\}$ (of interest in Theorem \ref{thm:initial oracle})
 happens with high probability. This task
is not straightforward in the presence of both competing risks and censoring.
The greatest challenge
is the lack of the martingale property in $\dm(\betao)$.
Even if we use its martingale proxy (an approach useful in low-dimensions)
as the gradient of \eqref{def:tm}
\begin{equation}
\dtm(\betao) =  n^{-1}\sumin \inttao \{\bZ_i(t)-\bZtil(t,\bbeta)\} d\Nobs_i(t)
\end{equation}
with $\bZtil(t,\bbeta)$ defined in \eqref{def:Stk}, 
the approximation error between $\dm(\betao)$ and $\dtm(\betao)$
is difficult to control 
because the error is
determined by  $\{\omega_i(t)-I(C_i \ge t)\}e^{\betaot\bZ_i(t)}$ with $\omega_i(t)$ defined
in \eqref{def:ipw}, 
which can be significantly amplified when the  relative risks grow with the dimension.
 To prove Lemma \ref{lemma:score oracle},  we first show that the relative risks among subjects with observed type 2 events has sub-Gaussian tails. This is achieved through the argument that their CIF cannot be arbitrarily close to one; otherwise, these subjects would have probability close to one experiencing type 1 event. As the CIF is monotonically increasing with the relative risks, it is also unlikely to observe excessively large relative risks among the subjects with observed type 2 events.
We then use Lemma \ref{lemma:supij} in the Appendix to establish the concentration of $\bS{l}(t, \betao) - \btS{l}(t,\betao)$ around zero across all observed type 1 event times.

 Theorem \ref{thm:initial oracle}  assumes that
$\Pr\big(\compb >\Ckap\big)$ converges to zero for a sequence of $\Ckap$ bounded away from zero, as sample size $n$ goes to infinity.
 In Lemma \ref{lemma:kappa}, we show  that such event happens with high probability.
Using the connection between the compatibility factor and the restricted eigenvalue \cite[]{vdGeerBuhlmann09}, we show that $\compb$, the compatibility factor in the cone $\coneb$, is bounded away from zero with probability tending to one.


\begin{lemma}\label{lemma:kappa}
Let $\CSk{l}$ be defined as in \eqref{def:Csk}.
Denote
\begin{align*}
\ddC = & \left\{2\CSk{2}+4\KZ\CSk{1} + (5/2)\KZ^2\CSk{0}\right\}/\rstar\\
&+ \KZ^2\left\{(1+t^*\Llam)\sqrt{2\log\big(p(p+1)/\varepsilon\big)/n}+
    (2/\rstar)t^*\Llam \tnp^2\right\},
\end{align*}
where $\tnp$ is the solution of
$$
p(p+1)\exp\{-n\tnp^2/(2+2\tnp/3)\}=\varepsilon/2.221
.$$
If $\sbeta \sqrt{\log(p)/n} = o(1)$, we have under Assumptions \ref{Assum:design}- \ref{Assum:Hess}
for $n$ sufficiently large
$$
\Pr\left(\compb \ge \sqrt{\rhosig - \sbeta(\xi+1)\ddC}\right) \ge 1 - 6\varepsilon.
$$
\end{lemma}

\subsection{Asymptotic normality for one-step estimator and honest coverage of confidence intervals}
\label{section:theory-inf}

 Obtaining the asymptotic normality is technically challenging. The log-likelihood has dependent summands both through the initial lasso estimator as well as the Kaplan-Meier estimator. We establish the asymptotic normality for the one-step estimator $\bCI$ and coverage of the confidence intervals without requiring model-selection consistency of the initial estimator. To remove the small-sample bias of IPCW, we need  slightly stronger conditions than in the previous section. In this section alone, we use $K$ and $\rho$ without subscript to denote the constants independent of $n$, $p$ and $\sbeta$; 
 we have only one constant $\Kz$ that is allowed to grow with the dimension and is therefore denoted differently.

 \vskip 10pt

\begin{enumerate}[label=(D\arabic*)]
  \item \label{Assum:betaI} ({\bf Design})
  The true linear predictors are uniformly bounded with probability one
   \begin{equation}\label{aseq:bZ}
   \supi\supt \left|\betaot\bZ_i(t)\right| \le \Kb.
   \end{equation}
    \vskip 10pt

   \item ({\bf Hessian}) \label{Assum:HessI} The smallest eigenvalue  $ \lambda_{\min}( \Hess) \geq \rho >0$, where $\Hess$ is defined in \eqref{def:bSig}.

\vskip 12pt

  \item \label{Assum:CRI}({\bf Continuity})
  Each $\bZ_i(t)$  can be represented as
  $$
  \bZ_i(t) = \bZ_i(0)+\intt \dZ_i(u) du + \intt \jZ_i(u) d \Nz_i(u).
  $$
for random processes $\dZ_i(t)$, $\jZ_i(t)$ and the counting process $\Nz_i(t)$ such that ,
  $\betaot \dZ_i(t)$ is uniformly bounded between $\pm\LzI$ and uniformly Lipschitz-$\LzII$.
  Moreover,
   $\Nz_i(t)$'s  number of jumps $\Kz =o\left(\sqrt{n/(\log(p)\log(n))}\right)$
    and an  intensity function $\lamN(t)\le \LN$.
 \vskip 12pt

   \item ({\bf Dimension}) \label{Assum:norm rate}  The rows of the   matrix $\Hess^{-1}$ are  $\|\bth_j/\Theta_{j,j}\|_1 \le \Kgr$ and sparse with sparsities $\sgr_1,\dots, \sgr_p  \le \sgr_{\max}$. Lastly, $\sbeta (\sgr_{\max}+\sbeta) \log(p)/\sqrt{n} = o(1)$.
\end{enumerate}

\vskip 15pt
We next present  Theorem \ref{thm:normality} that justifies all the proposed inference  procedures in Section \ref{section:method CI}. For that purpose we denote the
  asymptotic variance of $\dm(\betao)$  with
\begin{equation}\label{def:Scov}
  \Scov = \E\{\bfeta_i + \bfpsi_i\}^{\otimes 2},
\end{equation}
where
\begin{align}
&    \bfeta_i =  \inttao \{\bZ_i(t)-\bmu(t)\} \tilde{\omega}_i(t) d\MI_i(t), \\
 &   \bfpsi_i = \inttao \inttao \frac{\bfq(t)}{\pi(t)}I(X_i \ge t) d\Mc_i(t),  \\
 &   \bfq(t) = \E \left[I(t>X_i) \int_t^{t^*} \{\bZ_i(u) - \bmu(u)\}\tilde{\omega}_i(u)d\MI_i(u)\right], \\
 &   \pi(t) = \Pr(X_i \ge u),
\end{align}
with $\MI_i(t)$, $\Mc_i(t)$ as defined in \eqref{def:MI} and \eqref{def:Mc}.

\begin{theorem}\label{thm:normality}
Let $\bTh$ be defined as in Section \ref{section:method nodewise}.
Let $\Scov$, $\bCI$, $\hTh$ and $\hScov$ be defined as in \eqref{def:Scov}, \eqref{def:bCI}, \eqref{def:hTh} and \eqref{def:hScov}, respectively.
Let $\bfc \in \R^p$ with $\|\bfc\|_1 =1$ and $\bfc^\top \bTh \Scov \bTh \bfc \to \nu^2 \in (0,\infty)$.
Then, whenever \ref{Assum:design} and \ref{Assum:betaI}-\ref{Assum:norm rate} hold,
  $$
  \frac{\sqrt{n}\bfc^\top(\bCI-\betao)}{\sqrt{\bfc^\top \hTh \hScov \hTh^\top \bfc}} \stackrel{d}{\to} N(0,1).
  $$
\end{theorem}

As a result of the stronger conditions required for Theorem \ref{thm:normality}, which we will explain in more details below, we are able to achieve an improved estimation error for the initial estimator
as stated in the next theorem.
\begin{theorem}\label{thm*:initial rate}
Under \ref{Assum:design} and \ref{Assum:betaI}-\ref{Assum:norm rate},
 we can choose $\lambda \asymp \sqrt{\log(p)/n}$ and $\Ckap = \sqrt{\rhosig/2}$,
 such that
$$
\|\bini-\betao\|_1 = O_p\left(\sbeta \sqrt{\log(p)/n}\right)= o_p(1).
$$
\end{theorem}

For the rest of this section, we explain  the   assumptions and theoretical results needed for Theorem \ref{thm:normality} summarized in Lemmas  \ref{lemma:bTh}-\ref{lemma:score cov}.
Condition \ref{Assum:betaI} is needed whenever the model departs significantly from the linear case \citep{vdGeerEtal14, fang2016testing}.
In our case, the asymptotic normality of  $\sqrt{n}\dm(\betao)$ depends fundamentally on the asymptotic tightness of  $\sqrt{n}\dtm(\betao)$.
As a necessary condition, the predictable quadratic variation under filtration $\ccFt$ of the martingale $\sqrt{n}\dtm(\betao)$
\begin{equation}\label{def:QV}
\langle\sqrt{n}\dtm(\betao)\rangle_{t^*} = \inttao n^{-1} \sumin I(C_i \ge t) Y_i(t) \RR{i}{t}{o} \{\bZ_i(t)-\bZtil(t,\betao)\}^{\otimes 2} \lamT(t) dt,
\end{equation}
must have a finite bound independent of the dimension of the covariates.
This requires that the magnitude of the summands in \eqref{def:QV} either be bounded or have light tails.
Hence, we cannot allow the relative risk $\RR{i}{t}{o}$ to grow arbitrarily large.
We next  observe that
\ref{Assum:HessI} is a standard assumption for the validity of the nodewise penalized regressions \eqref{def:hgrtau}.
Finally, note that Theorem \ref{thm:normality} utilizes
  Condition \ref{Assum:CRI}; a condition stronger than \ref{Assum:CR} needed for $\sqrt{n}$-  approximation error between $\dm(\betao)$ and $\dtm(\betao)$.

If we define the population versions of the nodewise   components defined in \eqref{def:hXi}-\eqref{def:hgrtau},
\begin{align}
&\bfXi = \inttao \{\bZ(t)-\bmu(t)\} d\Nobs(t), \;
  \bGr_j(\bfgr) = \E\{\Xi_j - \bfXi_{i,-j}^\top \bfgr \}^2, \notag\\
 & \gro_j = \argmin_{\bfgr \in \R^{p-1}} \bGr_j(\bfgr), \;
  \tau^2_j = \bGr_j(\gro_j),\label{def:asymp nodewise}
\end{align}
then the true   parameters $\{\gro_j,\tau^2_j: j=1, \dots, p\}$ uniquely define the inverse negative Hessian $\bTh$ as described in Section \ref{section:method nodewise}.
We prove this statement in the following Lemma.
\begin{lemma}\label{lemma:bTh}
Under \ref{Assum:HessI}, $\Theta_{j,j} = 1/\tau^2_j$ and $\bth_{j,-j}\tau^2_j = \gro_j$.
Moreover, $\|\gro_j\|_1 \le \Kgr$, $\tau^2_j \ge \rhosigI$ and
$\|\bTh\|_1 \le \Kgr/ \rhosigI$.
\end{lemma}

Next, we  discuss the properties of  estimands $\hat {\boldsymbol \gamma}_j$, $\hat \tau_j$ and $\hat {\boldsymbol \Theta}$ -- defining components of the variance estimate.

\begin{lemma}\label{lemma:gamma-l1}
Under  \ref{Assum:design} and \ref{Assum:betaI}-\ref{Assum:norm rate},
for  $\lambda_j \asymp \sbeta\sqrt{\log(p)/n}$, we obtain
  $$\sup_j\|\hgr_j - \gro_j\|_1 = O_p\left(\sbeta\sgr_j\sqrt{\log(p)/n}\right)$$ and
  $\sup_j|\htau_j^2 - \tau_j^2| = O_p(\sbeta\sgr_j\sqrt{\log(p)/n})$,
leading to $\|\hTh-\bTh\|_1 = O_p\left(\sbeta\sgr_{\max}\sqrt{\log(p)/n}\right)$.
\end{lemma}

The nodewise LASSO  in \eqref{def:hgrtau}, unlike \cite{vdGeerBuhlmann09} that has i.i.d. entries, has dependent $\hXi_i$'s through the common $\bZbar(t,\bini)$; see \eqref{def:hXi}.
Thus, our error rate takes the multiplicative form $\sbeta\sgr_{\max}$, instead of the summation $\sbeta+\sgr_{\max}$ that may be expected under the generalized linear models.
In general, we consider our rate to be optimal under our model.


Using Lemma \ref{lemma:gamma-l1}, we can establish the approximation condition for $\bCI$ proposed in \eqref{eq:approx idea}.
\begin{lemma}\label{lemma:approx cond}
Under  \ref{Assum:design} and \ref{Assum:betaI}-\ref{Assum:norm rate},
the one-step estimator $\bCI$ satisfies the approximation condition
$$
\sqrt{n} \bfc^\top \left\{\bTh\dm(\betao) + \betao - \bCI\right\} = O_p\left(\sbeta (\sgr_{\max}+\sbeta) \log(p)/\sqrt{n}\right) = o_p(1)
$$
for any $\bfc$ such that $\|\bfc\|_1=1$.
\end{lemma}

Next, we show the asymptotic normality of $\dm(\betao)$.
\begin{lemma}\label{lemma:weak conv}
Under  conditions \ref{Assum:design} and \ref{Assum:betaI}-\ref{Assum:norm rate}, for directional vector $\bfc \in \R^p$ with $\|\bfc\|_1 =1$ and $\bfc^\top \bTh \Scov \bTh^\top \bfc \to \nu^2 \in (0,\infty)$,
$$\sqrt{n} \bfc^\top \bTh \dm(\betao)/\sqrt{\bfc^\top \bTh \Scov \bTh^\top \bfc} \stackrel{d}{\to} N(0, 1).$$
\end{lemma}

The proof uses the same approach as the initial low-dimensional result in \cite{FineGray99}.
We approximate $\dm(\betao)$ by the sample average of i.i.d.~terms $\bfeta_i+\bfpsi_i$ plus an $o_p\left(n^{-1/2}\right)$ term.
We note that the same approach involves nontrivial techniques in order to  be valid in high-dimensions.
In particular, we discover and exploit the martingale property of the term $\{\omega_i(t)- I(C_i \ge t)\}/G(t)$.

The last piece of our proof for Theorem \ref{thm:normality} is the element-wise convergence of the ``meat" matrix \eqref{def:hScov} in the ``sandwich" variance estimator.
\begin{lemma}\label{lemma:score cov}
Under conditions  \ref{Assum:design} and \ref{Assum:betaI}-\ref{Assum:norm rate},
$$\supi \|\hat{\bfeta}_i(\hat{\bbeta}) + \hat{\bfpsi}_i(\hat{\bbeta}) - \bfeta_i -\bfpsi_i\|_\infty = O_p\left(\|\bini-\betao\|_1 + \sqrt{\log(p)/n}\right)=o_p(1).$$
Hence, $\|\hScov -\Scov \|_{\max} = o_p(1)$.
\end{lemma}

Putting Lemmas \ref{lemma:weak conv} and \ref{lemma:score cov} together, we obtain the main result stated in the Theorem \ref{thm:normality}.

The details of the proofs are presented in the Appendix. Throughout the proof, we rely heavily on our concentration results for time-dependent processes, which we state in Section \ref{section:concentration} and prove in Section \ref{section:aux proof}.

\section{Numerical Experiments}\label{section:simulation}

To assess the finite sample  properties of our proposed methods,
we conduct extensive simulation experiments with various dimensions and dependence structure among covariates.

\subsection{Setup 1}
Our first simulation setup follows closely the one of  \cite{FineGray99} but considers high-dimensional covariates. In particular, each $\bZ_i$ is a  vectors consisting of 
i.i.d.~standard normal random variables. For cause 1, only $\beta_{1,1}=\beta_{1,2}=0.5$ are non-zero.
The cumulative incidence function is:
$$
\Pr(T_i\le t,\epsilon_i=1|\mathbf{Z}_i)=1-[1-p\{1-\exp(-t)\}]^{\exp(\boldsymbol{\beta}_1^\top\mathbf{Z}_i)}.
$$
For  cause 2, $\beta_{2,1}=\beta_{2,3}=\dots=\beta_{2,p-1} =-0.5$ and $\beta_{2,2} =\beta_{2,4}=\dots=\beta_{2,p} =0.5$, with
$$
\Pr(T_i\le t | \varepsilon_i=2, \mathbf{Z}_i)=1-\exp \left(t e^{\boldsymbol{\beta}_2^\top\mathbf{Z}_i} \right).
$$
We consider four different combinations: $n=200$, $p=300$; $n=200, \, p=500$; $n=200, \, p=1000$; and $n=500, \, p=1000$.
Note that this setup considers sparsity for cause 1 but non-sparsity for cause 2 effects.
As the Fine-Gray model does not require modeling cause 2 to make inference on cause 1, we expect that the non-sparsity in cause 2 effects should not affect the inference on cause 1.

\begin{center}
\begin{longtable}{lccccccc}
  \caption{Simulation results with independent covariates }\label{table:simFG}\\
\toprule
& True & Mean Est & SD & SE & SE corrected & Coverage & Level/Power\\
\midrule
\hline
\multicolumn{8}{c}{n=200, p=300}\\
\hline

 $\beta_{1,1}$ & 0.5 & 0.51 & 0.16 & 0.13 & 0.25 & 0.94 & 0.92 \\
  $\beta_{1,2}$ & 0.5 & 0.47 & 0.15 & 0.14 & 0.22 & 0.94 & 0.93 \\
  $\beta_{1,10}$ & 0 & 0.03 & 0.12 & 0.15 & 0.18 & 0.98 & 0.04 \\

\hline
\multicolumn{8}{c}{n=200, p=500}\\
\hline

 $\beta_{1,1}$ & 0.5 & 0.51 & 0.16 & 0.14 & 0.19 & 0.93 & 0.95 \\
  $\beta_{1,2}$ & 0.5 & 0.48 & 0.15 & 0.13 & 0.19 & 0.93 & 0.88 \\
  $\beta_{1,10}$ & 0 & -0.01 & 0.10 & 0.14 & 0.16 & 1.00 & 0.01 \\

\hline
\multicolumn{8}{c}{n=200, p=1000}\\
\hline

 $\beta_{1,1}$ & 0.5 & 0.46 & 0.17 & 0.13 & 0.18 & 0.94 & 0.86 \\
  $\beta_{1,2}$ & 0.5 & 0.48 & 0.14 & 0.13 & 0.18 & 0.93 & 0.92 \\
  $\beta_{1,10}$ & 0 & -0.00 & 0.11 & 0.14 & 0.17 & 0.99 & 0.06 \\

\hline
\multicolumn{8}{c}{n=500, p=1000}\\
\hline

 $\beta_{1,1}$ & 0.5 & 0.51 & 0.10 & 0.08 & 0.14 & 0.99 & 1.00 \\
  $\beta_{1,2}$ & 0.5 & 0.50 & 0.10 & 0.08 & 0.15 & 0.99 & 0.99 \\
  $\beta_{1,10}$ & 0 & -0.00 & 0.07 & 0.08 & 0.14 & 1.00 & 0.03 \\

\bottomrule
\end{longtable}
\end{center}

The results are presented in Table \ref{table:simFG}.
We focus on inference for the two non-zero coefficients $\beta_{1,1}$ and $\beta_{1,2}$,
as well as one arbitrarily chosen zero coefficient $\beta_{1,10}$.
The mean estimates are the average of the one-step $\bCI$ over the 100 repetitions, reported together with
other quantities described below.
We can see from the average estimates column that the one-step $\bCI$ is bias-corrected and that the presence of many non-zero coefficients for causet 2 does not affect our inference on cause 1.

In practice the choice of the tuning parameters is particularly challenging; the optimal value is determined up to a constant. Moreover, the theoretical results are asymptotic.  These together with the finite sample effects of $n \ll p$, lead to suboptimal performance of many proposed one-step correction estimators \cite[]{vdGeerEtal14,fang2016testing}. Suboptimality is amplified for survival models, due to the nonlinearity of the loss function and the presence of censoring --  both require more significant sample size (to observe asymptotic statements in the finite samples).
In the following, we propose
 a finite-sample correction to the construction of confidence intervals and in particular the estimated standard error (SE).

Let $se(\hat{b}_j;\bini)$ denote the asymptotic standard error as given in
  Section \ref{section:method CI}.
As a finite-sample correction we propose to consider    $se(\hat{b}_j;\bCI)$ in place of $se(\hat{b}_j;\bini)$, where
  the variance estimation  based on  the initial LASSO estimate $\bini$ is  replaced by the one-step $\bCI$.
This  can be viewed as another iteration of the bias-correction formula.
The resulting SE is therefore a ``two-step" SE estimator. We report the coverage rate of the confidence intervals constructed with this finite-sample correction in Table \ref{table:simFG} and we observe good coverage close to the nominal level of  $95\%$. We note that with 100 simulation runs the margin of error for the simulated coverage probability is about 2.18\%, if the true coverage is 95\%; that is, the observed coverage can range between 95$+/-$4.36\%.
We note that the coverage is good for all three coefficients, where non-zero or zero. In contrast, results in the existing literature suffer under-coverage of the non-zero coefficients.


The last column `level/power' in Table \ref{table:simFG} refers to the empirical rejection rate of the null hypothesis that the coefficient is zero, by the two-sided
 Wald test
 $
Z = (\hat{b}_j - \beta_{1,j})/se(\hat{b}_j;\bini)
$ at a nominal $0.05$ significance level.
We see that although $ se(\hat{b}_j;\bini)$ is used, 
 the nominal level is
well preserved for the zero coefficient $\beta_{1,10}$, and the power is high for the non-zero coefficients $\beta_{1,1}$ and $\beta_{1,2}$ for the given sample sizes and signal strength.

We repeat the above simulations with different values  for $\beta_{1,1}$
to investigate the power of the Wald test.
The results are illustrated in Figure \ref{fig:power}, where we see that the power increases with $n$ and decreases with $p$ as expected.
\begin{figure}[h]
\centering
\caption{Power curve for testing $\beta_{1,1}=0$ at nominal level $0.05$}\label{fig:power}
\includegraphics[scale=0.63]{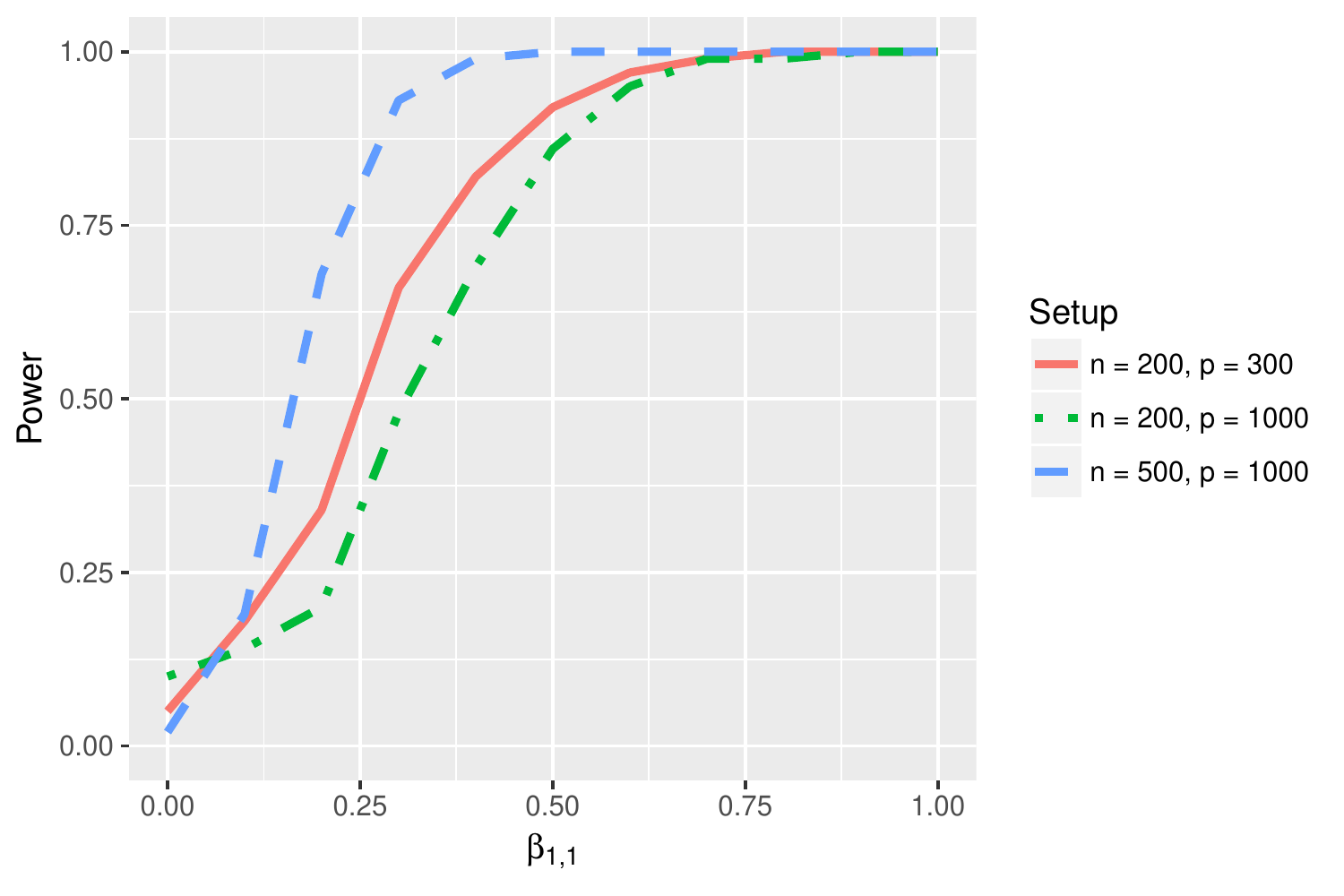}
\end{figure}

\subsection{Setup 2}

In the second setup we consider the case where the covariates are not all independent, which is more likely the case in practice for high dimensional data. We consider the block dependence structure also used  in \cite{BinderEtal09}.
We consider $n=500$, $p=1000$;
$\beta_{1,1\sim 8}=0.5$, $\beta_{1,9\sim 12}=-0.5$ and the rest are all zero.
$\beta_{2,1\sim 4}=\beta_{2,13\sim16}=0.5$, $\beta_{2,5\sim 8}=-0.5$ and the rest of $\bbeta_1$ are all zero.
The covariates are grouped into four blocks of size $4$, $4$, $8$ plus the rest, with
the within-block correlations equal to $0.5$, $0.35$, $0.05$ and $0$. The four blocks are separated by the horizontal lines in Table \ref{table:simBinder}.

\begin{center}
\begin{longtable}{lccccccc}

  \caption{Simulation results with block correlated covariates }\label{table:simBinder}\\
\toprule
& True & Mean Est & SD & SE & SE corrected & Coverage & Level/Power\\
\midrule
\hline
\multicolumn{8}{c}{n=500, p=1000}\\
\hline

 $\beta_{1,1}$ & 0.5 & 0.47 & 0.10 & 0.07 & 0.12 & 0.97 & 1.00 \\
  $\beta_{1,2}$ & 0.5 & 0.48 & 0.10 & 0.07 & 0.12 & 0.94 & 0.98 \\
  $\beta_{1,3}$ & 0.5 & 0.47 & 0.10 & 0.07 & 0.12 & 0.98 & 1.00 \\
  $\beta_{1,4}$ & 0.5 & 0.47 & 0.10 & 0.07 & 0.12 & 0.94 & 1.00 \\
   \hline$\beta_{1,5}$ & 0.5 & 0.48 & 0.10 & 0.06 & 0.11 & 0.93 & 1.00 \\
  $\beta_{1,6}$ & 0.5 & 0.46 & 0.10 & 0.06 & 0.11 & 0.94 & 1.00 \\
  $\beta_{1,7}$ & 0.5 & 0.47 & 0.09 & 0.06 & 0.11 & 0.95 & 1.00 \\
  $\beta_{1,8}$ & 0.5 & 0.47 & 0.08 & 0.06 & 0.11 & 0.98 & 1.00 \\
   \hline$\beta_{1,9}$ & -0.5 & -0.44 & 0.08 & 0.06 & 0.11 & 0.93 & 1.00 \\
  $\beta_{1,10}$ & -0.5 & -0.42 & 0.08 & 0.06 & 0.11 & 0.92 & 1.00 \\
  $\beta_{1,11}$ & -0.5 & -0.41 & 0.08 & 0.06 & 0.11 & 0.91 & 1.00 \\
  $\beta_{1,12}$ & -0.5 & -0.43 & 0.07 & 0.05 & 0.11 & 0.94 & 1.00 \\
  $\beta_{1,13}$ & 0 & -0.01 & 0.06 & 0.05 & 0.11 & 0.98 & 0.11 \\
  $\beta_{1,14}$ & 0 & -0.02 & 0.05 & 0.05 & 0.11 & 1.00 & 0.06 \\
  $\beta_{1,15}$ & 0 & -0.02 & 0.06 & 0.06 & 0.11 & 0.99 & 0.08 \\
  $\beta_{1,16}$ & 0 & -0.02 & 0.06 & 0.05 & 0.11 & 1.00 & 0.05 \\
   \hline$\beta_{1,30}$ & 0 & -0.00 & 0.05 & 0.06 & 0.11 & 1.00 & 0.01 \\

\bottomrule
\end{longtable}
\end{center}

Table \ref{table:simBinder} shows
the inferential results  for the non-zero coefficients $\beta_{1,1} \sim \beta_{1,12}$,
as well as the zero coefficients $\beta_{1,13} \sim \beta_{1,16}$ from the third correlated block that also contains some of the non-zero coefficients, and plus arbitrarily chosen zero coefficient $\beta_{1,30}$.
The initial LASSO estimator tended to select only one of every four non-zero coefficients of the correlated covariates (data not shown), as it is known that block dependence structure is particularly challenging for the Lasso type estimators. On the other hand,
 the one-step estimator performed remarkably well, capturing all of the non-zero coefficients.

Compared to the results in the last part of Table \ref{table:simFG} with the same $n$ and $p$,
the block correlated covariates led to slightly more bias in $\bCI$,
although the CI coverage remained high.
The power also remained high, although in the third block with the mixed
signal  and noise variables
 the type I error rates appeared  slightly high.

\section{SEER-Medicare data example}\label{section:data analysis}

The SEER-Medicare linked database contains clinical
information  and claims codes  for $57011$
patients diagnosed between  $2004$ and $2009$.
The clinical and demographic information were collected at  diagnosis, and the insurance claim data were from the year prior to diagnosis. The clinical information
contained PSA, Gleason Score, AJCC stage and year of diagnosis.
Demographic information included age, race, and
marital status.  The same data set was considered in \cite{2017arXiv170407989H}, where the emphasis was on variable selection and  prediction error. Our focus is  on testing and construction of confidence intervals.


In the following, we consider 2000 patients diagnosed during the year of
2004.
The only cause for loss to follow-up was the administrative censoring
at the end of the study which was year 2011.
Consequently, the year of enrollment was the only factor affecting the censoring distribution.
In our sample, all the subjects share the same year of enrollment 2004,
so we may reasonably make the independent censoring assumption.
 Among them \(76\) died from the cancer and
\(337\)  had deaths unrelated to cancer.
The process of identifying of the causes
is detailed in \cite{RiviereEtal19}.
There were $9326$ binary claims
codes in the data. Here we would like to identify the risk factors for non-cancer
mortality using the Fine-Gray model. We kept only  the claims codes with at least $10$
and at most $1990$
occurrences.
The resulting  dataset had  \(1197\)
covariates.
 We center and
standardize all the covariates before performing the analysis. To
determine the penalty parameters \(\lambda\)  and \(\lambda_j\)  we used 10-fold cross-validation.

In Table \ref{table:SMdata}, we present the result for  21 coefficients. Here, we focused on potential risk factors for non-cancer mortality, such as heart disease and colon cancer (different than prostate cancer); the coefficients to be tested were chosen ahead of time following recommendations from the doctors. We also include the clinical markers associated with the prostate cancer in comparison.
A descriptions of the variables is given in Table \ref{table:code_description}.
For each coefficient,
we report the initial estimate $\bini$, one-step estimate $\bCI$,
corrected SE,
the $95\%$ CI constructed with the corrected SE
and  the Wald test p-value (2-sided) calculated using the uncorrected SE.

\vskip 30pt

\begin{longtable}{lccccccc}
  \caption{Inference for the SEER-Medicare linked data on non-cancer mortality among prostate cancer patients}\label{table:SMdata}
\endhead
\toprule
Variables & Initial estimate && \multicolumn{4}{c}{One-step  estimate and Inference} \\
 \cline{4-8}
& $\hat{\beta}$ & & $\hat{b}$ &   se($\hat{b}$) & $95\%$ CI & p-value\\
\midrule
 Age & 0.075 &  & 0.096 &   0.009 & [ 0.078,  0.114] & \ 2e-24$^*$  \\
   Marital & 0 &  & 0.218 &   0.147 & [-0.071,  0.507] & 0.042 \\
   Race.OvW & 0 &  & -0.213 &  0.224 & [-0.652,  0.225] & 0.317 \\
   Race.BvW & 0.244 &  & 0.528  & 0.122 & [ 0.288,  0.767] &  \ 1e-04$^*$ \\
    \hline
    PSA & 0 &  & 0.005 &   0.003 & [-0.000,  0.010] & 0.041 \\
    GleasonScore & 0 &  & 0.084 &  0.050 & [-0.014,  0.182] & 0.085 \\
     AJCC-T2 & 0 &  & -0.130 &  0.146  & [-0.418,  0.157] & 0.218 \\
     \hline
   ICD-9 51881 & 0.866 &  & 1.357 &   0.361 & [ 0.650,  2.064] &  \ 4e-07$^*$ \\
    ICD-9 4280 & 0.404 &  & 0.697 &   0.062 & [ 0.576,  0.818] &  \ 2e-06$^*$ \\
   CPT 93015 & -0.061 &  & -1.042 &   0.327 & [-1.683, -0.401] &  \ 4e-05$^*$ \\
  ICD-9 42731 & 0.135 &  & 0.459 &   0.191 & [ 0.086,  0.833] &  \ 0.001$^*$ \\
  CPT 72050 & 0 &  & 3.718 &  0.208 &   [ 3.310,  4.125] &  \ 4e-05$^*$ \\
  ICD-9 6001 & 0 &  & -2.454 & 0.577  & [-3.585, -1.322] &  \ 0.000$^*$  \\
  CPT 74170 & 0 &  & -1.689 &  0.288  &[-2.255, -1.124] &  \ 0.001$^*$  \\
     \hline
 ICD-9 2948 & 0.539 &  & 0.746 &   0.205 & [ 0.343,  1.148] & 0.009 \\
  ICD-9 49121 & 0.150 &  & 0.476 &  0.215 & [ 0.055,  0.896] & 0.015 \\
   ICD-9 2989 & 0.079 &  & 0.450 & 0.135  &[ 0.184,  0.715] & 0.062 \\
   ICD-9 79093 & -0.056 &  & -0.348 & 0.176 &[-0.693, -0.002] & 0.088 \\
   ICD-9 41189 & 0 &  & 1.332 & 0.434  & [ 0.480,  2.184] & \ \ \ 0.003$^{**}$  \\
  CPT 45380 & 0 &  & -2.250 &  0.544  &[-3.318, -1.182] &  \ \ \ 0.003$^{**}$  \\
 ICD-9 3320 & 0 &  & 0.378 &   0.373 & [-0.353,  1.110] & 0.327  \\
 \bottomrule
\end{longtable}
 \begin{tablenotes}
 \small
     \item[] $*$ denotes $5$\% significance after Bonferroni correction for these 21 variables,
    whereas $**$ denotes $10$\% significance after Bonferroni correction for these 21 variables
   \end{tablenotes}

\vskip 20pt


In Table \ref{table:SMdata}, we see that the claims codes ICD-9 4280,   CPT 93015,    ICD-9 42731 are all related to the heart disease, and are all  significant at 5\% level Bonferonni correction for the 21 variables included in the table.
However, a heart attack indicator variable,   ICD-9 41189, shows up significant at 10\% level although the naive regularized estimator was not able to select this variable as important; this indicates  that our inference procedure is much more delicate (stable) at discovering significant variables.
In support of that, an indicator of a possible cancer in the abdomen,
CPT 74170,  is reported as significant at 5\% although the initial Lasso regularized method failed to include such variable.
Similar result is seen for the indicator of a fall (CPT 72050) which for an elderly person can be fatal.
 An indicator of a colon cancer (CPT 45380) turns out to be significant at 10\% although the  Lasso method set it to zero initially.
 Therefore,
 our one-step method is able to recover important risk factors that would have been missed by the initial regularized estimator.

\begin{longtable}{l p{340pt}}
\caption{ Description of the variables in Table \ref{table:SMdata}}\label{table:code_description}
  \endfirsthead
\endhead
\toprule
Code & Description \\
\midrule
 Age & Age at diagnosis \\
 Marital & marSt1: married vs other \\
  Race.OvW & Race: Other vs White \\
  Race.BvW & Race: Black with White \\
   PSA & PSA \\
   GleasonScore & Gleason Score \\
     AJCC-T2 & AJCC stage-T: T2 vs T1 \\
     ICD-9 51881 & Acute respiratry failure (Acute respiratory failure) \\
       ICD-9 4280 & Congestive heart failure; nonhypertensive [108.] \\
       CPT 93015 & Global Cardiovascular Stress Test \\
       ICD-9 42731 & Cardiac dysrhythmias [106.] \\
     CPT 72050 & Diagnostic Radiology (Diagnostic Imaging) Procedures of the Spine and Pelvis \\
         ICD-9 6001 & Nodular prostate \\
     CPT 74170 & Diagnostic Radiology (Diagnostic Imaging) Procedures of the Abdomen \\
          ICD-9 2948 & Delirium dementia and amnestic and other cognitive disorders [653] \\
           ICD-9 49121 & Obstructive chronic bronchitis \\
             ICD-9 2989 & Unspecified psychosis \\
              ICD-9 41189 & acute and subacute forms of ischemic heart disease, other \\
               CPT 45380 & Under Endoscopy Procedures on the Rectum \\
            ICD-9 3320 & Parkinsons disease [79.] \\
              \bottomrule
\end{longtable}

In contrast,
 non-life-threatening diseases,  were not selected as significant predictors for the non-cancer mortality. These include
Parkinson's (ICD-9 3320),  Psychosis (ICD-9 2989), Bronchitis (ICD-9 49121) and Dementia (ICD-9 2948) in the table. It is worth noting that some of these were selected by the initial estimate but were then corrected by our test.
We also note that the prostate cancer related variables, PSA, Gleason Score ahd AJCC  all have large $p$-values for non-cancer mortality. This is consistent with the results in \cite{2017arXiv170407989H}, where under the competing risk models the predictors for a second cause only has secondary importance in predicting the events due to the first cause.

\section{Discussion}

 This article focuses on estimation and inference under the Fine-Gray model with many more covariates than the number of events, which is well-known to be the effective sample size for survival data.  The article studies the rate of convergence of a Lasso estimator and develops a new one-step estimator that can be utilized for asymptotically optimal inference: confidence intervals and testing. These results can be generalized to any sparsity-inducing and convex penalty functions including but not limited to one-step SCAD, adaptive LASSO, elastic net, to name a few.   Moreover, it is worth noting that the variance estimation is novel in that it regresses a re-weighted score vector onto the score vector; in this way, the usual difficulty with asymptotic Hessian is avoided; it is worth pointing that the sandwich estimator or bootstrap carry biases in high-dimensions.
 
An often overlooked restriction on the time-dependent covariates $Z_i(t)$, $i=1,\dots,n$, under  the Fine-Gray model is that  $Z_i(t)$ must be observable even after the $i$-th subject experiences a type 2 event.
In practice, $Z_i(t)$ should be either time independent or external \cite[]{KalbfleischPrentice02}.
In our case the continuity conditions \ref{Assum:CR}
and \ref{Assum:CRI} are  easily satisfied if the majority of the elements in $Z_i(t)$  are time independent, which is most likely to be the case in practice. Our theory does not apply in studies involving longitudinal variables that are supposed to be truly measured continuously over time.

We have illustrated that the method based on regularization only (without bias correction) might have severe disadvantages in many complex data situations -- for example, it may potentially fail to identify relevant variables that are associated with the response.  From the analysis of the SEER-medicare data, we see that variables like CPT 72050  (related to fall) or, CPT 74170 (related to diagnostic imaging of the abdomen, often for suspected malignancies) would not have been discovered as important risk factors for non-cancer mortality by regularization alone. In reality, both can be life-threatening events for an elderly patient.
 The one-step estimate, on the other hand, was able to detect these,  therefore providing a valuable tool for practical applications. The one-step estimator is applicable as long as the model is sparse, and no minimum signal strength is required; this is another important aspect which makes the estimator more desirable for practical use than the LASSO type estimators.


\begin{center}
\textbf{Acknowledgement}
\end{center}
We would like to acknowledge our collaboration with Dr.~James Murphy of the UC San Diego Department of Radiation Medicine and Applied Sciences on the linked Medicare-SEER data analysis project that motivated this work. We would also like to thank his group for help in preparing the data set.

\appendix

\vskip 20pt 

 \centerline{\bf\Large Appendix}
 
 \vskip 10pt 

\renewcommand\thefigure{\thesection.\arabic{figure}}
\renewcommand\theequation{\thesection.\arabic{equation}}
\renewcommand\thetable{\thesection.\arabic{table}}
\renewcommand\thelemma{\thesection.\arabic{lemma}}

\setcounter{equation}{0}
\setcounter{lemma}{0}

In the appendix, we denote global quantities as $Q$ and event sets as $\Omega$ with subscripts
labelled by their order of appearance.
Other quantities are all local, i.e. only defined for the current Lemma.
We denote the ordered observed type-1 event times as
$\T{1}, \dots, \T{\KN}$.
\section{Concentration Inequalities}
Here we give the statements of the inequalities frequently used in our proofs.
The notations in this section are all generic.

\subsection{Classical Concentration Inequalities}\label{section:classical}
\begin{lemma}\label{lemma:Hoeffding}
\textbf{Hoeffding's Inequality} (Theorem 2 of \cite{Hoeffding63} p.4)
If $X_1,\dots,X_n$ are independent and $a_i \le X_i \le b_i$ $(i=1,2,\dots,n)$,
then for $t>0$
$$
\Pr(\bar{X}-\mu \ge t) \le \exp\left(-\frac{2n^2t^2}{\sum_{i=1}^n (b_i-a_i)^2}\right).
$$
\end{lemma}

\begin{lemma}\label{lemma:Azuma}
\textbf{A version of Azuma's Inequality} (Theorem 1 and Remark 7 of \cite{Sason13} p.3 and p.5)
Let $\{X_k,\mathcal{F}_k\}_k=0^\infty$ be a discrete-parameter real-valued
martingale sequence such that for every $k$, the condition $|X_k-X_{k-1}|\le a_k$ holds almost surely
for some non-negative constants $\{a_k\}_{k=1}^\infty$. Then
$$
\Pr\left(\max_{k\in 1,\dots, n} |X_k - X_0| \ge t\right) \le 2 \exp\left(-\frac{t^2}{2\sum_{k=1}^n a_k^2}\right)
$$

\end{lemma}

\subsection{Concentration Inequalities for Time-dependent Processes}\label{section:concentration}

\begin{lemma}\label{lemma:supijt}
Let $\{(\bfS_i(t), N_i(t)) \in \R^q \times \N: i=1,\dots,n, t \in [0,t^*]\}$ be i.i.d. pairs of random processes.
Each $N_i(t)$ is a counting process bounded by $\KNi$.
Denote its jumps as $0 \le t_{i1} < \dots < t_{iK_i} \le t^*$.
Let $\bar{\bfS}(t) = n^{-1} \sumin \bfS_i(t)$ and $\E\{\bfS_i(t)\} = \bfs(t)$.
Suppose $\supij\supt\|\bfS_i(t)-\bfS_j(t)\|_{\max} \le K_S$ almost surely.
Then,
\begin{enumerate}[label = (\roman*), ref = \ref{lemma:supijt}(\roman*)]
  \item \label{lemma:supij}
  $\Pr\left(\supi\sup_{j=1,\dots, K_i} \left\|\bar{\bfS}(t_{ij}) - \bfs(t_{ij})\right\|_{\max} > K_S x+(K_S)/n\right) < 2 n \KNi q e^{-nx^2/2}.$
  \item \label{lemma:supt}
  Assume in addition that each $\bfS_i(t)$ is c\`{a}gl\`{a}d generated by
$$
\bfS_i(t) = \bfS_i(0) + \intt \bfd_s(u)du + \intt \bfJ_s(u)dN_i(u)
$$
for some $\bfd_s(t)$ and $\bfJ_s(t)$ satisfying $\|\bfd_s(t)\|_{\max} < \LS$ and
$\|\bfJ_s(u)\|_{\max}<K_S$,
and $\E\{N_i(t)\} = \intt \lamN_i(u) du$ for some $\lamN_i(t) \le \LN$.
We have $$\supi\supt \left\|\bar{\bfS}(t) - \bfs(t)\right\|_{\max} =  O_p(\sqrt{\log(n\KNi q)/n}).$$
\end{enumerate}
\end{lemma}

\begin{lemma}\label{lemma:supopM}
Let $\{M_i(t): t \in [0,t^*], \; i = 1, \dots, n\}$ be a $\Ft$-adapted counting process martingales $M_i(t) = N_i(t)-\intt Y_i(t)h_i(u)du$ satisfying $\supi\supt h_i(t) \le \KhA$.
Let  $\{\bfH_i(t): t \in [0,t^*], \; i = 1, \dots, n\}$ be the $q$ dimensional $\mathcal{F}_{t-}$-measurable processes
such that $$\supi\supt \|\bfH_i(t)\|_{\max} \le \KH.$$
For $\MH(t) = n^{-1}\sumin \intt \bfH_i(u) dM_i(u)$ ,
we have
\begin{enumerate}[label = (\roman*), ref = \ref{lemma:supopM}(\roman*)]
\item \label{lemma:supM}
$
\Pr\left( \supt \|\MH(t)\|_{\max} \ge \KH(1+\KhA t^*)x + \KH\KhA t^*/n \right) \le 2q e^{-nx^2/4}.
$
\item \label{lemma:opM} Assume in addition $\supi\supt\|\bfH_i(t)\|_{\max} = O_p(a_n)$ and
$\KhA t^* \asymp O(1)$.
Then,
$\supt \|\MH(t)\|_{\max} = O_p(a_n\sqrt{\log(q)/n}).$
\end{enumerate}
\end{lemma}

\setcounter{equation}{0}
\setcounter{lemma}{0}
\section{Proofs of Main Results}
We shall present our proofs in the following order.
First, we give the proofs to our theorems using the main Lemmas stated in Section \ref{section:theory}.
Second, we present the auxiliary lemmas necessary for the proofs of
main Lemmas.
Third, we present the proofs to the main Lemmas.
Lastly, we present the proofs to the our concentration inequalities and auxiliary lemmas.

\subsection{Proofs of Theorems}
\begin{proof}[Proof of Theorem \ref{thm:initial oracle}]
	
	Observe that the same techniques as those of \cite{HuangEtal13} apply (see for example Lemmas 3.1 and 3.2 therein). The structure of the partial likelihood is the same as that of the Cox model modular the IPW weight functions $w_j(t)$.
	 Following the same line of proof we can easily obtain on the event $\{\|\dm(\betao)\|_\infty < \lambda (\xi-1)/(\xi+1) \}$,
the estimation error of LASSO estimator $\bini$ defined in \eqref{def:bini} has the bound
\begin{equation}\label{eq:cone oracle}
  \|\bini-\betao\|_1 \le \frac{e^\varsigma (\xi+1)\sbeta\lambda}{2\compb^2},
\end{equation}
where $\varsigma$ is the smaller solution to
$$\varsigma e^{-\varsigma} = \KZ(\xi+1)\sbeta \lambda/\{2\compb^2\}.$$

	\begin{equation}\label{eq:HuangTh1}
	\|\bini-\betao\|_1 \le \frac{e^\varsigma(\xi+1)\sbeta\lambda}{2\compb^2}
	\end{equation}
	with $\varsigma_\bfb = \supt \supij |\bfb^\top \{\bZ_i(t)-\bZ_j(t)\}|$
	in the event $\|\dm(\betao)\|_1 \le \lambda(\xi-1)/(\xi+1)$. The proof is then completed by applying the
conclusion of  Lemma  \ref{lemma:score oracle}.
\end{proof}

\begin{proof}[Proof of Theorem \ref{thm:normality}]
Be Lemmas \ref{lemma:approx cond} and \ref{lemma:weak conv},
we have
$$
\sqrt{n}\frac{\bfc^\top (\bCI-\betao)}{\cTh \Scov \bTh^\top \bfc} = \sqrt{n}\frac{\bTh \dm(\betao)}{\cTh \Scov \bTh^\top \bfc} +o_p(1) \stackrel{d}{\to} N(0,1).
$$

In Lemma \ref{lemma:score cov}, we have shown that $\|\Scov\|_{\max}$ is bounded by $K^2(1+\LlamI e^{\Kb} t^*)^2\{1+2(1+\LcI)e^{\Kb}/\rstar\}^2$ with probability tending to one.
In Lemma \ref{lemma:bTh}, we have shown that $\|\bTh\|_1$ is bounded by $\Kgr /\rhosigI$.
Then, we can apply Lemmas \ref{lemma:gamma-l1} and \ref{lemma:score cov} to get
\begin{align*}
& |\cTh \Scov \bTh^\top \bfc - \bfc^\top \hTh \hScov \hTh^\top \bfc|
\le  \|\bfc\|_1 \|\bTh-\hTh\|_1 \|\Scov\|_{\max} \|\bTh\|_1\|\bfc\|_1 \\
& \qquad + \|\bfc\|_1 \{ \|\bTh\|_1 + \|\hTh-\bTh\|_1 \} \|\Scov-\hScov\|_{\max} \|\bTh\|_1\|\bfc\|_1 \\
&\qquad  + \|\bfc\|_1 \{ \|\bTh\|_1 + \|\hTh-\bTh\|_1 \}\{ \|\hScov-\Scov\|_{\max} +\|\Scov\|_{\max} \}\|\bTh-\hTh\|_1\|\bfc\|_1 \\
& \qquad =  2  O_p(\|\bTh-\hTh\|_1) + O_p(\|\Scov-\hScov\|_{\max}) = o_p(1).
\end{align*}
Note that we use the following fact
$$
\|\bfc^\top\bTh\|_1 = \sumjp |\sumip c_i \Theta_{i,j}|
\le \sumip |c_i| \sumjp |\Theta_{i,j}| \le \|\bfc\|_1\|\bTh\|_1.
$$
\end{proof}

\begin{proof}[Proof of Theorem \ref{thm*:initial rate}]

Since we assume \ref{Assum:betaI} now,
the relative risks are bounded almost surely from above and below by constants
$
0 < e^{-\Kb} \le \RR{i}{t}{o} \le e^{\Kb} <\infty.
$
We may set $\M = e^{\Kb}$ to directly obtain \ref{Assum:Hess} from \ref{Assum:HessI}.
We can also improve the rate of estimation error in Theorem \ref{thm:initial oracle} by $\log(n)$ because we need not let $\Ke$ in Lemma \ref{lemma:Sk} to  grow with $n$.
\end{proof}

\subsection{Auxiliary Lemmas}

\begin{lemma}\label{lemma:KZbar}
Let $\{a_i(t): t \in [0,t^*], i = 1,\dots, n\}$ be a set of nonnegative processes.
Under \eqref{aseq:Zij}, where $\KZ$ is defined,
$$
\left\|\frac{\sumin a_i(t)\bZ_i(t)^{\otimes l}}{\sumin a_i(t)}\right\|_{\max} \le (\KZ/2)^l,
\text{ and }
\left\|\frac{\E\{ a_i(t)\bZ_i(t)^{\otimes l}\}}{\E\{a_i(t)\}}\right\|_{\max} \le (\KZ/2)^l.
$$
As a result, the maximal norms of $\bS{l}(t,\bbeta)/\bSo(t,\bbeta)$ and $\btS{l}(t,\bbeta)/\btSo(t,\bbeta)$ and
$\bs{l}(t,\bbeta)/\bso(t,\bbeta)$, defined in \eqref{def:Sk} and \eqref{def:Stk},  are all uniformly bounded by $(\KZ/2)^l$.
\end{lemma}

\begin{lemma}\label{lemma:Ytau}
Let $\M$ and $\rstar$ be defined as in \eqref{aseq:denom}.
Define
\begin{equation}\label{def:btsoM}
\btSo(t;\M) = n^{-1}\sumin I(C_i \ge t^*) Y_i(t^*)  \min\{\M,\RR{i}{t}{o}\}.
\end{equation}
Let $\T{1}, \dots, \T{\KN}$ be the observed type-1 events.
Under \ref{Assum:design}, the event
\begin{equation}
\eventrM = \left\{n^{-1}\sum_{i=1}I(X_i \ge t^*) \ge \rstar/(2\M), \supk \btSo\left(\T{k};\M\right) \ge \rstar/2 \right\}
\label{def:eventrM}
\end{equation}
occurs with probability at least
$1 - e^{-n \rstar^2/(2\M^2)} - n e^{-n(\rstar-2/n)^2/(8\M^2)}$.

On $\eventrM$, we have $\supk \btSo(\T{k}) \ge \rstar/2$.
\end{lemma}

\begin{lemma}\label{lemma:RR}
Let
$
  \Ke =e^{\rtz \Lz \|\betao\|_\infty \dt} \log(n/\varepsilon)/(\dt\lamlow)
$
be defined as in \eqref{def:Ke}.
Under \ref{Assum:CR},
the event
\begin{equation}\label{def:eventKe}
\eventKe = \left\{\supi\supt I(\delta_i\epsilon_i > 1)\RR{i}{t}{o} < \Ke \right\}
\end{equation}
occurs with probability at least $1-
\varepsilon$.
\end{lemma}

\begin{lemma}\label{lemma:KM}
Define the IPW weights with true $G(t)$,
$
\tilde{\omega}_i(t) = r_i(t)G(t)/G(X_i\wedge t)
$, as in \eqref{def:wtil}
and
\begin{equation}\label{def:CKM}
\CKM = 4(\M/\rstar)^2\left\{(1+\Lc t^*)\sqrt{4\log(2/\varepsilon)/n}+ \Lc t^*/n\right\}.
\end{equation}
Under \ref{Assum:design},
\begin{equation}\label{def:eventKM}
\eventKM = \left\{ \supt\supt|\omega_i(t) - \tilde{\omega}_i(t)| \le \CKM \right\}
\end{equation}
occurs on event $\eventrM$ with probability at least $\Pr(\eventrM) - \varepsilon$.
\end{lemma}

\begin{lemma}\label{lemma:Sk}
 Define
 $$\bDelta{l}(t) = \bS{l}(t,\betao)-\btS{l}(t,\betao),$$
  with $\bS{l}$ and $\btS{l}$ defined in \eqref{def:Sk} and \eqref{def:Stk}.
 Let $\T{1}, \dots, \T{\KN}$ be the observed type-1 events for some $\KN \le n$.
 Denote
 $
  \Ke =e^{\rtz \Lz \|\betao\|_\infty \dt} \log(n/\varepsilon)/(\dt\lamlow)
$ and
 $$
 \CSk{l} =\frac{\Ke \KZ^l}{2^l}\left\{\frac{4\M^2(1+\Lc t^*)}{\rstar^2}\sqrt{\frac{4\log(2/\varepsilon)}{n}}+ \frac{4\M^2\Lc t^*}{\rstar^2n}
 + \sqrt{\frac{2\log(2np^l/\varepsilon)}{n}}+ \frac{1}{n}\right\}
 $$
 as in \eqref{def:Ke} and \eqref{def:Csk}.
Under \ref{Assum:design} and \ref{Assum:CR},
 \begin{align}
 \eventx = \left\{ \max_{l=0,1,2}\supk \left\|\bDelta{l}\left(\T{k}\right)\right\|_{\max}\le \CSk {l} \right\} \cap \eventrM \cap \eventKe \cap \eventKM \label{def:eventx},
 \end{align}
 with $\eventrM$, $\eventKe$ and $\eventKM$ defined in Lemmas \ref{lemma:Ytau}, \ref{lemma:RR} and \ref{lemma:KM},
  occurs with probability at least $1 - e^{-n \rstar^2/(2\M^2)} - n e^{-n(\rstar-2/n)^2/(8\M^2)}- 5\varepsilon$.

On $\eventx$, we have for $l=1,2$,
$$
\supk\left\|\frac{\bS{l}\left(\T{k},\betao\right)}{\bSo\left(\T{k},\betao\right)}
            -\frac{\btS{l}\left(\T{k},\betao\right)}{\btSo\left(\T{k},\betao\right)}\right\|_{\max} \le
           2\{\CSk{l}+(\KZ/2)^l\CSk{0}\}/\rstar.
$$
\end{lemma}

\begin{lemma}\label{lemma:opNorm}
Denote
$
  \bDelta{l}(t) = \bS{l}(t,\betao)-\btS{l}(t,\betao)
$ as in Lemma \ref{lemma:Sk},
with $\bS{l}(t,\betao)$ and $\btS{l}(t,\betao)$ defined in \eqref{def:Sk} and \eqref{def:Stk}, respectively.
Under \ref{Assum:design}, \ref{Assum:betaI} - \ref{Assum:CRI} and \ref{Assum:norm rate},
\begin{enumerate}[label = (\roman*), ref = \ref{lemma:opNorm}(\roman*)]
  \item \label{lemma:opbmu}
  $\supt \|\Delta^{(0)}(t)\|_{\max} = O_p\left(\sqrt{\log(n)/n}\right)$;

  $\sup_{l=1,2}\supt \|\bDelta{l}(t)\|_{\max}$,
  $\supt \|\bZbar(t,\betao)-\bZtil(t,\betao)\|_\infty$,

  $\supt \|\bZtil(t,\betao) - \bmu(t)\|_\infty$
  and $\supt \|\bZbar(t,\betao) - \bmu(t)\|_\infty$ are all $O_p\left(\sqrt{\log(p)/n}\right)$;
  \item \label{lemma:opSk} Define
  \begin{equation}\label{def:dS}
  \dS_i(t) = \{\omega_i(t)-I(C_i > t)\}Y_i(t).
  \end{equation}
Let $\bfh(\bZ)$ be a differentiable operator $\R^p\mapsto \R^q$ uniformly bounded by $\Kh \asymp 1$
with $\|\nabla \bfh(\bZ)\|_1 < \Lh \asymp 1$,
and $\bfg(t)$ be a $\ccFtH$ adapted process in $\R^{q'}$ with bound $\supt\|\bfg(t)\|_{\max}\le\Kg \asymp 1$.
Whenever $q q' = p$, we have
\begin{equation}\label{eq:opSk}
\left\|n^{-1/2}\sumin \inttao n^{-1} \sumjn \dS_j(t)\bfh(\bZ_j(t)) \bfg(t)^\top I(C_i \ge t)d\MI_i(t)\right\|_{\max} = o_p(1);
\end{equation}
 \item \label{lemma:opbeta}
   for any $\tbeta \in \R^p$, $\supt \|\bZbar(t,\betao) - \bZbar(t,\tbeta)\|_\infty = O_p(\|\tbeta - \betao\|_1)$; if $\|\tbeta - \betao\|_1 = o_p(1)$,
  $$\supi\supt \left|\frac{\RR{i}{t}{o}}{\bSo(t,\betao)} - \frac{\RRt{i}{t}}{\bSo(t,\tbeta)}\right| = O_p(\|\tbeta - \betao\|_1).$$
\end{enumerate}
\end{lemma}

\begin{lemma}\label{lemma:YtauI}
Let $\bSo$ and $\btSo$ be defined as in \eqref{def:Sk} and \eqref{def:Stk}, respectively.
Under \ref{Assum:design} and \ref{Assum:betaI}, $\supt|n/\{\sumin I(X_i \ge t^*)\}|$,
 $\supt|\bSo(t,\betao)^{-1}|$ and $\supt|\btSo(t,\betao)^{-1}|$
are all $O_p(1)$.
\end{lemma}

\begin{lemma}\label{lemma:gamma-oracle}
Let $\Gr_j$, $\bini$ and $\gro_j$ be defined as in \eqref{def:Gr}, \eqref{def:bini} and \eqref{def:grotau},
respectively.
On the event
\begin{equation}\label{def:eventg}
\eventg :=\left \{ \left \|\dGr_j \bigl(\gro_j, \bini \bigl) \right \|_{\infty} \leq (\xi_j-1)\lambda_j/(\xi_j+1), \forall j = 1,\dots, p \right \},
\end{equation}
 we have under \ref{Assum:HessI}
\begin{enumerate}[label=(\roman*)]
\item \label{part:gamma-cone} the estimation error $\tilde{\bfgr}_j := \hgr_j - \gro_j$ belongs to the cone
\begin{equation}\label{def:conej}
\conegrj:= \{\bfv\in \mathbb{R}^{p-1}: \|\bfv_{\mathcal{O}^c_j}\|_1 \leq \xi_j \|\bfv_{\mathcal{O}_j}\|_1\}
\end{equation}
\item and
$\|\hgr_j - \gro_j\|_1 \le \{s_j\lambda_j(\xi_j+1)\}/\{2\compgr^2\}$, with compatibility factor
\begin{equation}\label{def:kappaj}
\compgr = \sup_{0\neq \bfg \in \conegrj}\frac{\sqrt{\sgr_j \bfg^\top \ddGr(\gro,\bini)\bfg}}{\|\bfg_{\Ocal_j}\|_1}
\end{equation}
\end{enumerate}
 for all $j = 1, \ldots, p$.
\end{lemma}

\begin{lemma}\label{lemma:gamma-score}
Let $\Gr_j$, $\bini$ and $\gro_j$ be defined as in \eqref{def:Gr}, \eqref{def:bini} and \eqref{def:grotau},
respectively.
Under  \ref{Assum:design} and \ref{Assum:betaI}-\ref{Assum:norm rate},
$\maxj \left \|\dGr_j \bigl(\gro_j, \bini \bigl) \right \|_\infty =  O_p\left(\|\bini-\betao\|_1+\sqrt{\log(p)/n}\right). $
\end{lemma}

\begin{lemma}\label{lemma:Hess}
Let $\hHess$, $\Hess$, $\ddm$ be defined as in \eqref{def:hHess}, \eqref{def:bSig} and \eqref{def:Hess},
respectively.
Under  \ref{Assum:design} and \ref{Assum:betaI}-\ref{Assum:norm rate},
\begin{enumerate}[label=(\roman*),ref = \ref{lemma:Hess}(\roman*)]
\item \label{lemma:hHess} $\left\|\hHess - \Hess\right\|_{\max} =O_p\left(\sbeta\sqrt{\log(p)/n}\right)$;
\item \label{lemma:ddm} for any $\tbeta$ such that $\|\tbeta-\betao\|_1=o_p(1)$,
$$\left\|-\ddm(\tbeta) - \Hess\right\|_{\max} =O_p\left(\|\tbeta-\betao\|_1 + \sqrt{\log(p)/n}\right).$$
\end{enumerate}
\end{lemma}

\begin{lemma}\label{lemma:gamma-compatibility}
Let $\compgr$ be define as in Lemma \ref{lemma:gamma-oracle} \eqref{def:conej}.
Under \ref{Assum:design} and \ref{Assum:betaI}-\ref{Assum:norm rate},
 setting $\xi_{\max} = \max_{j=1,\dots, p} \xi_j \asymp 1$, we have
$$\Pr\left(\inf_j \compgr^2 \ge \rhosigI/2\right) \to 1.$$
\end{lemma}

\subsection{Proof of Main Lemmas}
\begin{proof}[Proof of Lemma \ref{lemma:score oracle}]
Let $\T{1}, \dots, \T{\KN}$ be the observed type-1 events.
  We may decompose the score $\dm(\betao)$ as its martingale proxy plus an approximation error,
    $$
     \dm(\betao) = \dtm(\betao)
     + n^{-1}\sum_{k=1,\dots,\KN} \left\{\bZtil\left(\T{k},\betao\right)- \bZbar\left(\T{k},\betao\right)\right\} ,
    $$
    with $\bZtil$ and $\bZbar$ defined in \eqref{def:Sk} and \eqref{def:Stk}, respectively.

    Recall that the counting process for observed type-1 event can be written as
    $\Nobs_i(t) = \intt I(C_i \ge u) d \NI_i(t)$. Moreover,
    $\dtm(\betao)$ takes the form of the Cox model score with counting process $\{\Nobs_i(t)\}$
    and at-risk process $\{I(C_i \ge t)Y_i(t)\}$.
    The ``censoring complete" filtration $\ccFt$ can also be equivalently generated by $\{\Nobs_i(t), I(C_i \ge t)Y_i(t), \bZ_i(t)\}$.
    Thus, we may apply Lemma 3.3 in \cite{HuangEtal13} under \eqref{aseq:Zij} from \ref{Assum:design},
    $$
    \Pr(\|\dtm(\betao)\|_\infty > \KZ x) \le 2pe^{-nx^2/2}.
    $$
    Notice that the inequality is sharper than that in Lemma \ref{lemma:supM}
    because the compensator part of $\dtm(\betao)$ is zero.

    The concentration result for approximation error
    $$
    \bZtil\left(\T{k},\betao\right)- \bZbar\left(\T{k},\betao\right)
    = \frac{\bS{1}\left(\T{k},\betao\right)}{\bSo\left(\T{k},\betao\right)}
            -\frac{\btS{1}\left(\T{k},\betao\right)}{\btSo\left(\T{k},\betao\right)}
    $$
    is established in Lemma \ref{lemma:Sk} on $\eventx$.
    We obtain the concentration inequality for $\dm(\betao)$ by adding the bounds and tail probabilities together.
\end{proof}

\begin{proof}[Proof of Lemma \ref{lemma:kappa}]
    Our strategy here is the same as that for Lemma \ref{lemma:score oracle}.
    We first show that $\compb$ is lower bounded by $\kappa(\xi, \Ocal;-\ddtm(\betao))$
    plus a diminishing error.
    Since $\ddtm(\betao)$ takes the form of a Cox model Hessian, we then may apply the results from \cite{HuangEtal13}.

    By Lemma 4.1 in \cite{HuangEtal13} (for a similar result, see \cite{vdGeerBuhlmann09} Corollary 10.1),
    $$
    \kappa^2(\xi, \Ocal;-\ddm(\betao))
     \ge \kappa^2(\xi, \Ocal;-\ddtm(\betao)) - \sbeta (\xi+1)^2 \|\ddm(\betao)-\ddtm(\betao)\|_{\max}.
    $$
    Let $\T{1}, \dots, \T{\KN}$ be the observed type-1 events.
    We can write $\ddm(\betao)-\ddtm(\betao)$ as
    $$
    -n^{-1}\sum_{k=1}^{\KN}\left[ \frac{\bS{2}\left(\T{k},\betao\right)}{\bSo\left(\T{k},\betao\right)}
    -\frac{\btS{2}\left(\T{k},\betao\right)}{\btSo\left(\T{k},\betao\right)}
            -\bZbar\left(\T{k},\betao\right)^{\otimes 2}
            +\bZtil\left(\T{k},\betao\right)^{\otimes 2} \right],
    $$
    with $\bS{l}$, $\btS{l}$, $\bZtil$ and $\bZbar$ defined in \eqref{def:Sk} and \eqref{def:Stk}.
    By Lemma \ref{lemma:KZbar}, $\supt \|\bZbar(t,\betao)\|_\infty$ and $\supt \|\bZtil(t,\betao)\|_\infty$ are both bounded by $\KZ/2$.
    On the $\eventx$ as defined in Lemma \ref{lemma:Sk},
    we apply Lemma \ref{lemma:Sk} once with $l=2$ and twice with $l=1$ to get
    $$\|\ddm(\betao)-\ddtm(\betao)\|_{\max} \le \left\{2\CSk{2}+4\KZ\CSk{1} + (5/2)\KZ^2\CSk{0}\right\}/\rstar,$$
    with $\CSk{l}$ defined in \eqref{def:Csk}.

    Our \ref{Assum:design} and \ref{Assum:Hess} contains all the condition for Theorem 4.1 in \cite{HuangEtal13}.
    Hence, we may apply their result
    \begin{align*}
    \kappa^2(\xi, \Ocal;-\ddtm(\betao))  &\ge \kappa^2(\xi, \Ocal;\bSig(\M))
    - \sbeta(\xi+1)^2\KZ^2
    \\
    & \times \left\{(1+t^*\Llam)\sqrt{2\log\big(p(p+1)/\varepsilon\big)/n}+
    (2/\rstar)t^*\Llam \tnp^2\right\}
    \end{align*}
    with probability at least $\Pr(\eventx) - 3\varepsilon$.
    We have bounded $\btSo(t;\M)$ away from zero at all observed type-1 events in $\eventx$,
    so the $e^{-n\rstar^2/(8\M^2)}$ term is absorbed into $\Pr(\eventx)$.
\end{proof}

\begin{proof}[Proof of Lemma \ref{lemma:bTh}]
The notations in the proof are defined in Section \ref{section:method nodewise}.
  Denote
    $$
    \bfXi = \inttao \{\bZ(t)-\bmu(t)\}d\Nobs(t).
    $$
    Without loss of generality, we set $j=1$.
    Since we define
    $\gro_1 = \argmin_{\bfgr} \bGr(\bfgr)$ as the minimizer of a convex function,
    it must satisfy the first order condition
    $$
    \dbGr(\gro_1) = \E\left\{(\Xi_1-\bfXi_{-1}^\top \gro_1)\bfXi_{-1}\right\} = \mathbf{0}_{p-1}.
    $$
    Recall that $\tau^2_1 = \bGr(\gro_1)$.
    Applying the first order condition, we get
    $$\tau^2_1 = \E\{\Xi_1-\bfXi_{-1}^\top \gro_1\}^2 =  \E\{(\Xi_1-\bfXi_{-1}^\top \gro_1)\Xi_1\}. $$
    We construct a vector $\bth_1 = (1, -\grot_1)^\top/\tau^2_1 \in \R^p$. Then,
    $\bth_1$ satisfies
    $$
    \bth_1^\top \Hess = (1, -\grot_1) \E\{\bfXi\bfXi^\top\}/\tau^2_1
     = (1, \bzero_{p-1}^\top).
    $$
    Hence, we have
    $$
    ( \bth_1, \dots , \bth_p)^\top
    = \Hess^{-1} = \bTh.
    $$
    We can directly bound
    $$
    \|\gro_j\|_1 = \|\bth_j/\Theta_{j,j}\|_1-1 \le \Kgr-1 < \Kgr.
    $$

    By \ref{Assum:HessI}, the minimal eigenvalue of $\Hess$ is at least $\rhosigI$.
    We obtain through a spectral decomposition that the maximal eigenvalue of $\bTh = \Hess^{-1}$ is at most $\rhosigI^{-1}$.
    Hence, we have
    $$
    \tau_j^2 = \left(\mathbf{e}_j^\top \bTh \mathbf{e}_j\right)^{-1} \ge \rhosigI
    $$
    and
    $$
    \|\bTh\|_1 \le \maxj \|\bth_j/\Theta{j,j}\| \maxj |\Theta_{j,j}| \le \Kgr/\rhosigI.
    $$
\end{proof}

\begin{proof}[Proof of Lemma \ref{lemma:gamma-l1}]
By Lemma \ref{lemma:gamma-score},
we may choose $\xi_1=\dots=\xi_p=2$ and $\lambda_1=\dots =\lambda_p = \lambda_\varepsilon \asymp O_p(\sbeta\sqrt{\log(p)/n})$
such that
$\eventg$ defined in Lemma \ref{lemma:gamma-oracle} occurs with probability $1-\varepsilon$.
Then, we establish the oracle inequality by Lemma \ref{lemma:gamma-oracle},
$$
\Pr\left(\maxj\|\hgr_j-\gro_j\|_1/\sgr_j \le \frac{2\lambda_\varepsilon}{\rhosigI}\right)
\ge \Pr\left(\minj \compgr^2 \ge \rhosigI/2 \right) - \varepsilon.
$$
We have shown that $\Pr\left(\minj \compgr^2 \ge \rhosigI/2 \right)$ tends to one in Lemma \ref{lemma:gamma-compatibility}.
Hence, $\maxj\|\hgr_j-\gro_j\|_1 = O_p\left(\sbeta\sgr_{\max}\sqrt{\log(p)/n}\right)$.

Define according to \eqref{def:asymp nodewise}
$\bfXi_i = \inttao \{\bZ_i(t)-\bmu(t)\}d\Nobs_i(t)$.
By Lemma \ref{lemma:KZbar}, $\supi\|\bfXi_i\|_\infty \le K$.
We introduce
$$
  \tGr_j(\bfgr) = n^{-1}\sumin \{\Xi_j - \bfXi_{i,-j}^\top \gro_j\} =
  n^{-1}\sumin \inttao \{Z_{ij}(t)-\mu_j(t)-\bfgr^\top \bZ_{i,-j}(t) + \bfgr^\top \bmu_{-j}(t)\}^2 d\Nobs_i(t)
$$
and decompose
$$
\hat{\tau}^2_j - \tau^2_j =  \Gr_j(\hgr_j,\bini) - \tGr_j(\gro_j) + \tGr_j(\gro_j) - \bGr_j(\gro_j).
$$
$\Gr_j(\hgr_j,\bini) - \tGr_j(\gro_j) = O_p\left(\sbeta\sgr_j\sqrt{\log(p)/n}\right)$ by the results from Theorem \ref{thm*:initial rate}, Lemma \ref{lemma:opNorm} and first part of this Lemma.
Apparently, $\tGr_j(\gro_j)$ is the average of i.i.d. terms.
The expectation of the summands in $\tGr_j(\gro_j)$
is defined as $\bGr_j(\gro_j)$ in \eqref{def:asymp nodewise}.
Hence,
we finish the proof by applying Lemma \ref{lemma:Hoeffding}.

Along with Lemma \ref{lemma:bTh}, we can prove with the previous results in this Lemma, $\|\hTh - \bTh\|_1 = O_p\left(\sbeta\sgr_{\max}\sqrt{\log(p)/n}\right)$.
\end{proof}

\begin{proof}[Proof of Lemma \ref{lemma:approx cond}]
We decompose
\begin{align}\label{eq:approx decomp}
    & \sqrt{n} \bfc^\top \left\{\bTh\dm(\betao) + \betao - \bCI\right\}
  \\
  & =  \sqrt{n} \bfc^\top \{\bTh-\hTh\}\dm(\bini)+ \sqrt{n} \cTh\{\dm(\betao) - \dm(\bini)\}
  + \sqrt{n} \bfc^\top(\betao-\bini).
\end{align}

By Lemma \ref{lemma:gamma-l1},
$\|\bTh-\hat{\bTh}\|_1 = O_p(\sbeta\sgr_{\max}\sqrt{\log(p)/n})$.
Each summand in $\dm(\bini)$ is the integral of $\bZ_i(t)$ minus a weighted average $\bZbar(t,\bini)$ over a counting measure $d\Nobs_i(t)$.
By the KKT condition and Theorem \ref{thm*:initial rate}, $\|\dm(\bini)\|_\infty \asymp \lambda \asymp O(\sqrt{\log(p)/n})$.
Putting these together, we obtain
\begin{align}
\sqrt{n} |\bfc^\top \{\bTh-\hTh\}\dm(\bini)| &\le \sqrt{n}\|\bfc\|_1 \|\bTh-\hTh\|_1 \|\dm(\bini)\|_\infty
\\
&= O_p\left(\sbeta\sgr_{\max}\log(p)/\sqrt{n}\right) = o_p(1).
\end{align}
By the KKT condition and Theorem \ref{thm:initial oracle},
$\|\dot{\bfm}(\hat{\bbeta})\| \le \lambda \asymp n^{-(1/2-d)}$.
Hence, the first term in \eqref{eq:approx decomp} is $o_p(1)$.
Like in the proof of Lemma \ref{lemma:weak conv},
we have $\|\cTh\|_1 \le \|\bfc\|_1\|\bTh\|_1 \le \Kgr e^{\Kb}/ \rstar$ from Lemma \ref{lemma:bTh}.

Define $\bbeta_r = \betao + r(\bini-\betao)$.
Applying mean value theorem to $h(r) = \cTh \dm(\beta_r)$,
we get
$$
\cTh \dm(\betao) - \cTh \dm(\bini) = - h'(\tr)
= -\cTh \ddm(\bbeta_{\tr}) (\bini - \betao)
$$
for some $\tr \in [0,1]$.
By Theorem \ref{thm*:initial rate}, we have
$$\|\bbeta_{\tr}-\betao\|_1 = \tr \|\bini-\betao\|_1 = O_p\left(\sbeta\sqrt{\log(p)/n}\right).$$
By Lemma \ref{lemma:ddm},
$\|-\ddm(\bbeta_{\tr})- \Hess\|_{\max} = O_p\left(\sbeta\sqrt{\log(p)/n}\right)$.
Along with Theorem \ref{thm*:initial rate} and Lemma \ref{lemma:bTh},
we have
\begin{align*}
\sqrt{n} |\cTh\{\dm(\betao) - \dm(\bini)\}
  +  \bfc^\top(\betao-\bini)|
= & \sqrt{n} | \cTh\{\Hess + \ddm(\bbeta_{\tr})\}(\betao-\bini)| \\
\le & \sqrt{n} \|\bfc\|_1\|\bTh\|_1 \|-\ddm(\bbeta_{\tr})- \Hess\|_{\max} \|\bini - \betao\|_1 \\
= & O_p\left(\sbeta^2\log(p)/\sqrt{n}\right).
\end{align*}
\end{proof}

\begin{proof}[Proof of Lemma \ref{lemma:weak conv}]
Since $\omega_i(t)Y_i(t) \neq I(C_i \ge t)Y_i(t)$ implies $\epsilon_i >1$ thus $\NI_i(t^*)=0$,
we have the equivalence $d\Nobs_i(t)=\omega_i(t)d\NI_i(t)=I(C_i \ge t)d\NI_i(t)$.
Recall for the following calculation that
\begin{align*}
&\bS{l}(t,\betao)=n^{-1}\sumin \omega_i(t)Y_i(t)\RR{i}{t}{o}\bZ_i(t)^{\otimes l}, \\
& \btS{l}(t,\betao)=n^{-1}\sumin I(C_i \ge t)Y_i(t)\RR{i}{t}{o}\bZ_i(t)^{\otimes l}, \\
&\bDelta{l}(t) = \bS{l}(t,\betao)-\btS{l}(t,\betao), \\
& \E\{\bS{l}(t,\betao)\} = \E\{\btS{l}(t,\betao)\} = \bs{l}(t,\betao) \\
&\bZbar(t,\betao) = \bS{1}(t,\betao)/\bSo(t,\betao), \quad
\bZtil(t,\betao) = \btS{1}(t,\betao)/\btSo(t,\betao), \\
&\bmu(t) = \bs{1}(t,\betao)/\bso(t,\betao), \quad
Y_i(t) = 1-\NI_i(t-)\\
\text{and }& \MI_i(t) = \NI_i(t)-\intt Y_i(t) \RR{i}{u}{o} \lamT(u)du.
\end{align*}

We decompose
\begin{align*}
 \sqrt{n}\dm(\betao) =
  & n^{-1/2} \sumin \inttao \left\{\bZ_i(t)-\bZbar(t,\betao)\right\}d\Nobs_i(t) \\
= & n^{-1/2} \sumin \inttao \left\{\bZ_i(t)-\bZbar(t,\betao)\right\}\omega_i(t)d\MI_i(t) \\
=
  &   n^{-1/2} \sumin \inttao \left\{\bmu(t)-\bZtil(t,\betao)\right\}I(C_i \ge t)d\MI_i(t) \\
  &+ n^{-1/2} \sumin \inttao \left\{\bZtil(t,\betao)-\bZbar(t,\betao)\right\} I(C_i \ge t)d\MI_i(t) \\
  &+ n^{-1/2} \sumin \inttao \left\{\bZbar(t,\betao) -\bmu(t)\right\}\Delta^{(0)}(t) \lamT(t) dt \\
  & + n^{-1/2} \sumin \inttao \left\{\bZ_i(t)-\bmu(t)\right\}\omega_i(t)d\MI_i(t) \\
\triangleq & I_1 + I_2 + I_3 + I_4.
\end{align*}

Notice that $I_1$ is a $\ccFt$ martingale.
We have $\|\bmu(t)-\bZtil(t,\betao)\|_\infty = O_p(\sqrt{\log(p)/n})$ from Lemma \ref{lemma:opbmu}.
Hence, we can apply Lemma \ref{lemma:opM} to get $\|I_1\|_\infty= \sqrt{n}O_p(\sqrt{\log(p)/n}^2)
= o_p(1)$.

We further decompose $I_2$ into 3 terms
\begin{align*}
  &  -n^{-1/2} \sumin \inttao \frac{\bDelta{1}(t)}{\btSo(t,\betao)} I(C_i \ge t)d\MI_i(t)
  -  n^{-1/2} \sumin \inttao\frac{\Delta^{(0)}(t)}{\btSo(t,\betao)}\bmu(t) I(C_i \ge t)d\MI_i(t) \\
 & +  n^{-1/2} \sumin \inttao\frac{\Delta^{(0)}(t)}{\btSo(t,\betao)}\{\bmu(t)-\bZbar(t,\betao)\} I(C_i \ge t)d\MI_i(t) \\
\triangleq & I_2'+I_2^{\prime\prime} + I_2^{\prime\prime\prime}.
\end{align*}
By \ref{Assum:betaI} and \ref{Assum:CRI}, each $\MI_i(t)$ has one jump at observed event time and $e^{\Kb}\LlamI-$Lipschitz elsewhere.
Since the $\{C_i, T^1_i: i=1,\dots, n\}$ is a set of independent continuous random variables,
there is no tie among them with probability one.
Hence, we may modify the integrand in $I_2'$ and $I_2^{\prime\prime}$ at observed censoring times
without changing the integral.
Replacing $\bDelta{l}(t)$ with $n^{-1}\sum_{j=1}^n \dS_i(t)\RR{i}{t}{o}\bZ_i(t)^{\otimes l}$,
we can apply Lemma \ref{lemma:opSk} to get
that $\|I_2'\|_\infty$ and $\|I_2^{\prime\prime}\|_\infty$ are both $o_p(1)$.

The total variation of $\MI_i(t)$ is at most $\max\{1, e^{\Kb}\LlamI t^*\} \asymp 1$.
By Lemma \ref{lemma:opbmu},
$\|\bDelta{0}(t)\{\bmu(t)-\bZbar(t,\betao)\}\|_\infty = O_p(\sqrt{\log(n)\log(p)}/n)$.
Hence, we obtain $\|I_2^{\prime\prime\prime}\|_\infty = O_p(\sqrt{\log(n)\log(p)}/n)=o_p(1)$.
Similarly, we obtain $\|I_3\|_\infty= O_p(\sqrt{\log(n)\log(p)}/n)=o_p(1)$.

Besides the one in Lemma \ref{lemma:KM},
$\omega_i(t) - \tilde{\omega}_i(t)$ has another martingale representation.
Denote the Nelson-Aalen estimator
$$
\hLamc(t) =\sumin \intt \frac{I(X_i \ge u)}{\sum_{j=1}^n I(X_j \ge u)}d \Nc_i(u).
$$
We have a $\cFt$ martingale
$$
\MNA(t) = \hLamc(t) - \intt \lamC(u) du
 = \sumin \intt \frac{I(X_i \ge u)}{\sum_{j=1}^n I(X_j \ge u)}d \Mc_i(u).
$$
By Lemma \ref{lemma:supM}, $\supt|\MNA(t)| = O_p\left(n^{-1/2}\right)$
For $t>X_i$ and $\delta_i\epsilon_i>1$,
$$
\omega_i(t) - \tilde{\omega}_i(t) = -\tilde{\omega}_i(t) \intt I(u>X_i)d\MNA(u) + R_i(t)
$$
with an error
$$
R_i(t) = \frac{\hG(t)}{\hG(X_i)}-\exp\left\{\hLamc(X_i)-\hLamc(t)\right\}
+ \frac{G(t)}{G(X_i)} \left[e^{-\intt I(u>X_i)d\MNA(u)} + \intt I(u>X_i)d\MNA(u) \right].
$$
It is the discrepancy between the Kaplan-Meier and the Nelson-Aalen plus a second order Tailer expansion remainder.
We shall show that it is $O_p(1/n)$.
Since
$$
\left|\intt I(u>X_i)d\MNA(u)\right| \le 2\supt |\MNA(t)| = O_p\left(n^{-1/2}\right),
$$
the second order remainder
$$
\left|e^{-\intt I(u>X_i)d\MNA(u)} + \intt I(u>X_i)d\MNA(u) \right| = O_p(1/n).
$$
Under \ref{Assum:design}, $\{\sumin I(X_i \ge t)\}^{-1} \le \{\sumin I(X_i \ge t^*)\}^{-1} = O_p(1/n)$.
Let $c_k$ be an observed censoring time.
The increment in $-\log(\hG(t)) - \hLamc(t)$ at $c_k$ is a second order remainder
$$
\log\left(1 - \frac{1}{\sumin I(X_i \ge c_k)}\right)
 - \frac{1}{\sumin I(X_i \ge c_k)} = O_p\left(n^{-2}\right).
$$
Hence, $\supt |-\log(\hG(t)) - \hLamc(t)| = O_p(1/n)$.
Applying the Mean Value Theorem, we obtain $\supt |\hG(t) - \exp\{-\hLamc(t)\}| = O_p(1/n)$.
Under \ref{Assum:design}, $G(t) \ge  G(t^*)$ is bounded away from zero, and $-\log(G(t)) \le -\log(G(t^*))$ is bounded from above. We have shown that both $\hG(t)$ and $\hLamc(t)$ are uniformly $\sqrt{n}$ consistent.
We obtain that $\hG(X_i)$ is bounded away from zero and $\hLamc(t)$ is bounded with probability tending to one.
Putting these together, we obtain
$$\supi\supt |R_i(t)| = O_p(1/n).$$

Define
$$
\tilde{\bfq}(t) = n^{-1} \sumin I(t\ge X_i) \int_t^{t^*}\left\{\bZ_i(u)-\bmu(u)\right\}\tilde{\omega}_i(u)d\MI_i(u),
$$
$\hat{\pi}(t) = n^{-1}\sumin I(X_i \ge t)$ and
$\bfq(t) = \E\{\tilde{\bfq}(t)\}$, $\pi(t)=\E\{\hat{\pi}(t)\}$.
We write $I_4$ as i.i.d. sum plus error through integration by parts,
\begin{align*}
& n^{-1/2} \sumin \inttao \left\{\bZ_i(t)-\bmu(t)\right\}\tilde{\omega}_i(t)d\MI_i(t)
\\
& \qquad + n^{-1/2} \sumin \inttao \left\{\bZ_i(t)-\bmu(t)\right\}\{\omega_i(t)-\tilde{\omega}_i(t)\}d\MI_i(t) \\
= & n^{-1/2} \sumin \inttao \left\{\bZ_i(t)-\bmu(t)\right\}\tilde{\omega}_i(t)d\MI_i(t)
+ n^{-1/2} \sumin \inttao \left\{\bZ_i(t)-\bmu(t)\right\}R_i(t)d\MI_i(t) \\
& - n^{-1/2} \sum_{k=1}^n \inttao  \frac{\bfq(t)
}{\pi(t)} I(X_k \ge u)d\Mc_k(t)
 \\
 & \qquad + n^{-1/2} \sum_{k=1}^n \inttao  \frac{\bfq(t)
}{\hat{\pi}(t)\pi(t)}\{\hat{\pi}(t) -  \pi(t)\} I(X_k \ge u)d\Mc_k(t) \\
& + n^{-1/2} \{\bfq(0)-\tilde{\bfq}(0)\}\sum_{k=1}^n \inttao  \frac{1
}{\hat{\pi}(t)} I(X_k \ge u)d\Mc_k(t) \\
& -  n^{-1/2}\sum_{k=1}^n \inttao  \frac{\{\bfq(0)-\bfq(t)-\tilde{\bfq}(0)+\tilde{\bfq}(t)\}
}{\hat{\pi}(t)} I(X_k \ge u)d\Mc_k(t) \\
\triangleq & \Iiv{1}+\Iiv{2}+\Iiv{3}+\Iiv{4}+\Iiv{5}+\Iiv{6}.
\end{align*}
$\Iiv{1} + \Iiv{3}$ is already a sum of i.i.d..
We have shown that $\supt |R_i(t)| = O_p(1/n)$.
Hence, we have $\|\Iiv{2}\|_\infty = O_p\left(n^{-1/2}\right)=o_p(1)$.
$I(t\ge X_i) \int_t^{t^*}\left\{\bZ_i(u)-\bmu(u)\right\}\tilde{\omega}_i(u)d\MI_i(u)$
is uniformly bounded by $K(\LlamI t^*+1)$. It has at most one jump and is $K\LlamI-$Lipschitz elsewhere.
Hence, we can apply Lemma \ref{lemma:supt} to get $\supt\|\bfq(t)-\tilde{\bfq}(t)\|_\infty=O_p(\sqrt{\log(p)/n})$
and $\supt|\pi(t)-\hat{\pi}(t)| = O_p(\sqrt{\log(n)/n})$.
Notice that $\Iiv{4}$, $\Iiv{6}$ and $n^{-1}\sum_{k=1}^n \inttao  \hat{\pi}(t)^{-1} I(X_k \ge u)d\Mc_k(t)$ in $\Iiv{5}$ are all $\cFt$ martingales.
We may apply Lemmas \ref{lemma:supM} and \ref{lemma:opM} to obtain
$\Iiv{4} = O_p(\sqrt{\log(n)\log{p}/n})=o_p(1)$, $\Iiv{5} = O_p(\sqrt{\log{p}/n})=o_p(1)$
and $\Iiv{6} = O_p(\log{p}/\sqrt{n})=o_p(1)$.

By Lemma \ref{lemma:bTh}, we can bound the $l_1$ norm of $\cTh$ by
$$
\|\cTh\|_1 = \sum_{i=1}^p \sum_{j=1}^p |c_i| |\Theta_{ij}|
\le \sum_{i=1}^p  |c_i| \Kgr/\rhosigI = \Kgr/\rhosigI.
$$
Finally, we write $\cTh \dm(\betao)$ as i.i.d. sum
\begin{align*}
&n^{-1/2} \sumin \cTh \left[\inttao \left\{\bZ_i(t)-\bmu(t)\right\}\tilde{\omega}_i(t)d\MI_i(t)
- \inttao  \frac{\bfq(t)
}{\pi(t)} I(X_i \ge u)d\Mc_i(t)\right] + o_p(1) \\
\triangleq &   n^{-1/2} \sumin \cTh \{\bfeta_i -  \bfpsi_i\} + o_p(1).
\end{align*}
We have $\E\{\cTh\bfeta_i\} = 0 $ because of its martingale structure.
We show $\E\{\cTh\bfpsi_i\} = 0$ again by introducing its martingale proxy
\begin{align*}
\E\{\cTh\bfpsi_i\} = & \E \left[\inttao \cTh \left\{\bZ_i(t)-\bmu(t)\right\}I(C_i \ge t)d\MI_i(t)\right] \\
 & +
\E \left[\inttao \cTh \left\{\bZ_i(t)-\bmu(t)\right\}\E\{\tilde{\omega}_i(t)-I(C_i \ge t)| T_i,\bZ_i(\cdot)\}d\MI_i(t)\right].
\end{align*}
The first term above is zero because of the martingale structure.
The second term is zero because the IPW weights satisfy $\E\{\tilde{\omega}_i(t)-I(C_i \ge t)| T_i,\bZ_i(\cdot)\} = 0$.
Each $\cTh \{\bfpsi_i -  \bfeta_i\}$ is mean zero and bounded by
$\Kgr/\rhosigI K(1+\LlamI t^*)+\Kgr/\rhosigI K(1+\LlamI t^*)(1+\LcI t^*)2e^{\Kb}/\rstar$ with probability equaling one.
The variance $\cTh \Scov \bTh \bfc$ has a bounded and non-degenerating limit $\nu^2$.
Hence, $\{ \cTh(\bfpsi_i -  \bfeta_i): i =1, \dots, n\}$ satisfies the Lindeberg condition.

By Lindeberg-Feller CLT,
$$
\sqrt{n} \frac{\bfc^\top \bTh \dot{\bfm}(\bbeta^o)}{\sqrt{\bfc^\top \bTh \Scov \bTh \bfc}} = \frac{\cTh\sumin \{\bfeta_i -  \bfpsi_i\}}{\sqrt{n\cTh \Scov \bTh \bfc}} + o_p(1) \stackrel{d}{\to} N(0,1).
$$
We conclude the proof of the Lemma.
 \end{proof}

 \begin{proof}[Proof of Lemma \ref{lemma:score cov}]
We define
$$
\teta_i = \inttao \{\bZ_i(u)-\bmu(u)\}\tilde{\omega}_i(u)d\MIt_i(u), $$
with
$$
\MIt_i(t) = \Nobs_i(t) - n^{-1}\sumjn \intt \frac{Y_i(u)\RR{i}{u}{o}}{\btSo(u,\betao)}d\Nobs_j(u).
$$
Under \ref{Assum:betaI} and \ref{Assum:design},
the total variation of $\MIt_i(t)$ is at most $1+2e^{2\Kb}/\rstar$ with probability tending to one by Lemma \ref{lemma:YtauI}.
The difference between $\teta_i$ and $\heta_i$ is
\begin{align*}
\heta_i - \teta_i = & n^{-1}\sumjn \inttao \{\bZ_i(u)-\bZbar(u,\bini)\}\omega_i(u)Y_i(u)
\left\{\frac{\RR{i}{u}{o}}{\btSo(u,\betao)}-\frac{\RRh{i}{u}}{\bSo(u,\bini)}\right\}d\Nobs_j(u) \\
& + \inttao  \{\bmu(u)\tilde{\omega}_i(u)-\bZbar(u,\bini)\omega_i(u)\}d\MIt_i(u).
\end{align*}
By Lemmas \ref{lemma:KM}, \ref{lemma:opbmu} and \ref{lemma:opbeta},
$\supi\|\heta_i - \teta_i\|_\infty = O_p\left(\|\bini-\betao\|_1 + \sqrt{\log(p)/n}\right)$.

Then, we study
$$
\bfeta_i-\teta_i = n^{-1}\sumjn \inttao \{\bZ_i(u)-\bmu(u)\}\tilde{\omega}_i(u)I(C_j \ge u)d\MI_j(u).
$$
We have the bound $\|\bZ_i(u)-\bmu(u)\|_\infty \le K$
from Lemma \ref{lemma:KZbar}.
$\tilde{\omega}_i(u)$ is not $\ccFt$ measurable, but we can define a new filtration
$\ccFit = \sigma\{X_i, \delta_i, \epsilon_i, \bZ_i(\cdot), I(C_j \ge u), \NI_j(u), \bZ_j(\cdot):
u\le t, j\neq i\}$ for each $i$, such that
$$
n^{-1}\sum_{j\neq i} \inttao \{\bZ_i(u)-\bmu(u)\}\tilde{\omega}_i(u)I(C_j \ge u)d\MI_j(u)
= \bfeta_i-\teta_i + O_p(1/n)
$$
is a $\ccFit$ martingale.
Hence, we can apply Lemma \ref{lemma:supM} to get
$$
\Pr\left(\|\bfeta_i-\teta_i\|_\infty \ge K(1+e^{\Kb}\LlamI t^*)\sqrt{4\log(2np/\varepsilon)/n}+K(1+2 e^{\Kb}\LlamI t^*)/n\right) \le \varepsilon/n.
$$
Taking union bound, we get $\|\bfeta_i-\teta_i\|_\infty = O_p(\sqrt{\log(p)/n})$.
Hence, $\supi\|\heta_i - \bfeta_i\|_\infty = O_p\left(\|\bini-\betao\|_1 + \sqrt{\log(p)/n}\right)$.

Recall that $\hbfq(t)$ and $\bfq(t)$ also take a similar form.
We can likewise define
$$
\tbfq(t)=n^{-1}\sumin I(t>X_i) \int_t^{t^*} \{\bZ_i(u)-\bmu(u)\}\tilde{\omega}_i(u) d\MIt_i(u)
$$
and
$$
\tbfq^*(t) = n^{-1}\sumin I(t>X_i) \int_t^{t^*} \{\bZ_i(u)-\bmu(u)\}\tilde{\omega}_i(u) d\MI_i(u).
$$
By Lemmas \ref{lemma:KM}, \ref{lemma:opbmu} and \ref{lemma:opbeta},
we have
$$\supi\supt\|\tbfq(t)-\hbfq(t)\|_\infty = O_p\left(\|\bini-\betao\|_1 + \sqrt{\log(p)/n}\right).$$
By Lemma \ref{lemma:supt}, $\supt\|\tbfq^*(t)- \bfq(t)\| = O_p\left(\sqrt{\log(p)/n}\right)$.
We only need to find the rate for
$$
\tbfq^*(t) - \tbfq(t) = n^{-1} \sumin I(t>X_i) n^{-1}\sumjn \int_t^{t^*} n^{-1} \sumin \{\bZ_i(u)-\bmu(u)\}\tilde{\omega}_i(u) I(C_j \ge u)d\MI_j(u).
$$
We repeat the trick for $\bfeta_i-\teta_i$.
Applying Lemma \ref{lemma:opM} to the $\ccFit$ martingale
$$
\bfM^q_i(t) = n^{-1}\sum_{j\neq i} \intt n^{-1} \sumin \{\bZ_i(u)-\bmu(u)\}\tilde{\omega}_i(u) I(C_j \ge u)d\MI_j(u)
$$
and obtain $\supi\supt\|\bfM^q_i(t)\|_\infty = O_p(\sqrt{\log(p)/n})$.
Hence,
$$
\supt\|\tbfq^*(t) - \tbfq(t)\|_\infty  \le 2 \supi\supt\|\bfM^q_i(t)\|_\infty + O_p(1/n) = O_p(\sqrt{\log(p)/n}).
$$
Putting the rates together, we have $\supt\|\hbfq(t) - \bfq(t)\|_\infty = O_p\left(\|\bini-\betao\|_1 + \sqrt{\log(p)/n}\right)$.

We can directly obtain $\supt|\hpi(t)-\pi(t)| = O_p\left(\sqrt{\log(n)/n}\right)$ from Lemma \ref{lemma:supt}.
Define
$$
\tpsi_i = \inttao \frac{\bfq(t)}{\pi(t)} d \Mch_i(t)
$$
The total variation of $\Mch_i(t)$ is at most $1+2e^{\Kb}/\rstar$ with probability tending to one by Lemma \ref{lemma:YtauI}.
Using the results so far,
we have
$$\supi\|\hpsi_i - \tpsi_i\|_\infty = O_p\left(\|\bini-\betao\|_1 + \sqrt{\log(p)/n}\right).$$
The remainder
$$
\bfpsi_i-\tpsi_i = n^{-1} \sumjn \inttao \frac{\bfq(t)}{\pi(t)} I(C_i \ge t)I(X_j \ge t)d\Mc_j(t)
$$
is a $\cFt$ martingale.
We can put the $n$ martingales in $\R^p$ into a $\R^{np}$ vector and apply Lemma \ref{lemma:supM},
$$
\supi\|\bfpsi_i-\tpsi_i\|_\infty = O_p\left(\sqrt{\log(np)/n}\right)=O_p\left(\sqrt{\log(p)/n}\right).
$$
Therefore, we get $\supi\|\bfpsi_i-\hpsi_i\|_\infty = O_p\left(\|\bini-\betao\|_1 + \sqrt{\log(p)/n}\right)$.

Finally, we decompose
\begin{align*}
\|\hScov - \Scov\|_{\max} \le &
n^{-1}\sumin \|\heta_i+\hpsi_i\|_\infty \|\heta_i+\hpsi_i - \bfeta_i-\bfpsi_i\|_\infty
\\
&+n^{-1}\sumin \|\heta_i+\hpsi_i - \bfeta_i-\bfpsi_i\|_\infty\|\bfeta_i+\bfpsi_i\|_\infty\\
& + \left\|n^{-1}\sumin (\bfeta_i+\bfpsi_i)(\bfeta_i+\bfpsi_i)^\top - \Scov\right\|_{\max}.
\end{align*}
We have shown that $\supi\|\heta_i+\hpsi_i - \bfeta_i-\bfpsi_i\|_\infty = o_p(1)$.
Moreover, $\supi\|\heta_i+\hpsi_i\|_\infty$ is $O_p(1)$ by Lemmas \ref{lemma:KZbar} and \ref{lemma:YtauI}. In addition, we observe  that
$n^{-1}\sumin (\bfeta_i+\bfpsi_i)(\bfeta_i+\bfpsi_i)^\top$ is an average of i.i.d. terms whose expectation is defined as $\Scov$.
By Lemmas \ref{lemma:KZbar} and \ref{lemma:YtauI}, we have the uniform maximal bound
$$
\supi\|(\bfeta_i+\bfpsi_i)(\bfeta_i+\bfpsi_i)^\top\|_{\max} = \supi \|(\bfeta_i+\bfpsi_i)\|_\infty^2
$$
is also $O_p(1)$.
We finish the proof by applying Lemma \ref{lemma:Hoeffding}  to  the last term in the decomposition above, $\left\|n^{-1}\sumin (\bfeta_i+\bfpsi_i)(\bfeta_i+\bfpsi_i)^\top - \Scov\right\|_{\max}$.

\end{proof}

\subsection{Proofs of Auxiliary Lemmas}\label{section:aux proof}

\begin{proof}[Proof of Lemma \ref{lemma:supijt}]
\begin{enumerate}[label = (\roman*)]
\item
Without loss of generality, let $t_{11}$ be the first jump time of $N_1(t)$.
By the i.i.d. assumption, $t_{11}$ is independent of all $\bfS_i(t)$ with $i\ge 2$.
Thus, the sequence
$$
\bfL_l = n^{-1}\sum_{i=2}^l  \left\{ \bfS_i(t_{11})-\bfs(t_{11})\right\}
$$
is a martingale with respect to filtration $\left\{\sigma\big(\bfS_i(t), i\le l\big), l = 2, \dots, n\right\}$.
The increment is bounded as
$$
 n^{-1}\left\{ \bfS_i(t_{11})-\bfs(t_{11})\right\}
 = n^{-1} {\E}_{\bfS_j}\left\{ \bfS_i(t_{11})-\bfS_j(t_{11})\right\}
 \le n^{-1} K_S.
$$
Applying Lemma \ref{lemma:Azuma}  to $\bfL_n$, we get
$
\Pr\left(\left\|\bfL_n\right\|_{\max} > K_Sx\right) < 2q e^{-nx^2/2}
$.
Since the dropped first term is also bounded by $K_S/n$, we get
$$
\Pr\left(\left\|\bar{\bfS}(t_{11})-\bfs(t_{11})\right\|_{\max} > K_Sx+K_S/n\right) < 2q e^{-nx^2/2}.
$$
We use simple union bound to extend the result to all $t_{ij}$'s whose number is at most $n\KNi$.

\item Define a deterministic set
$\Tcal_n = \{kt^*/n: k=1,\dots, n\} \cup \Tz$.
  By the union bound of Hoeffding's inequality \cite{Hoeffding63} ,
  we have
  $$
  \Pr\left( \sup_{t\in \Tcal_n} \left\|\bar{\bfS}(t)-\bfs(t)\right\|_{\max}> K_S x\right) < 2(n+|\Tz|)qe^{-nx^2/2}.
  $$  Combining the result from Lemma \ref{lemma:supij},
  we obtain
  $$
  \left\|\bar{\bfS}(t) - \bfs(t)\right\|_{\max} = O_p(\sqrt{\log(npq)/n})
  $$
  over a grid containing $\Tcal_n$ and jumps of $N_i(t)$.
  We only need to show that the variation of $\bar{\bfS}(t) - \bfs(t)$ is sufficiently small inside each bin created by the grid.

Let $t'$ and $t^{\prime\prime}$ be consecutive elements by order in $\Tcal_n$.
By our construction, there is no jump of any of the counting processes $N_i(t)$ in the interval  $(t',t^{\prime\prime})$.
Otherwise, the jump time is another element in $\Tcal_n$ between $t'$ and $t^{\prime\prime}$ so that $t'$ and $t^{\prime\prime}$ are not consecutive elements by order.
Under the assumption of the lemma, elements of all
$\bfS_i(t)$'s are $\LS-$Lipschitz in $(t',t^{\prime\prime})$.
Moreover, $|t^{\prime\prime}-t'|\le t^*/n$ because of the deterministic $\{kt^*/n: k=1,\dots, n\} \subset \Tcal_n$.
  Along with the c\`{a}gl\`{a}d property, we obtain a bound of variation of $\bar{\bfS}(t)$ in $(t',t^{\prime\prime})$
  $$
  \sup_{t\in(t',t^{\prime\prime})} \|\bar{\bfS}(t) - \bar{\bfS}(t^{\prime\prime})\|_{\max}
  \le \supi \sup_{t\in(t',t^{\prime\prime})}  \|\bfS_i(t) -\bfS_i(t^{\prime\prime})\|_{\max}
  \le \LS |t^{\prime\prime}-t'| \le \LS t^*/n.
  $$
  For any $t \in (t',t^{\prime\prime})$, we bound the variation of $\bfs(t)$ by
  $$
  \|\bfs(t)-\bfs(t^{\prime\prime})\|_{\max} \le \int_t^{t^{\prime\prime}} \E \|\bfd_s(u)\|_{\max}du
  +  \int_t^{t^{\prime\prime}} \E\{ \|\bfJ_s(u)\|_{\max}\lamN_i(u)\} du
  \le (\LS+K_S \LN) t^*/n.
  $$
  For arbitrary $t \in [0,t^*]$, we find the the corresponding bin $(t',t^{\prime\prime}]$ contains $t$.
  Putting the results together, we have
  \begin{align*}
  \|\bar{\bfS}(t) - \bfs(t)\|_{\max} &\le \|\bar{\bfS}(t) -\bar{\bfS}(t^{\prime\prime})\|_{\max} +
  \| \bfs(t) - \bfs(t^{\prime\prime})\|_{\max} +\|\bar{\bfS}(t^{\prime\prime}) - \bfs(t^{\prime\prime})\|_{\max}
  \\
  &\le O_p(\sqrt{\log(npq)/n}) + O(1/n).
 \end{align*}
\end{enumerate}
\end{proof}

\begin{proof}[Proof of Lemma \ref{lemma:supopM}]
\begin{enumerate}[label = (\roman*)]
\item
  The summands in $\MH(t)$ are the integrals of $\mathcal{F}_{t-}$-measurable processes over $\Ft$-adapted martingales,
  so $\MH(t)$ is a $\Ft$-adapted martingale \cite[p.165]{KalbfleischPrentice02}.

  Suppose $\{T_i: i= 1, \dots, n\}$ are the jump times of $\{N_i(t)\}$.
  We artificially set $T_i = t^*$ if $N_i(t)$ has no jump in $[0,t^*]$.
  Define $0 \le R_1 \le \dots \le R_{2n}$ be the order statistics of
  $\{T_i: i= 1, \dots, n\} \cup \{kt^*/n: k = 1, \dots, n\}$.
  Hence, $\{R_k: k = 1, \dots, 2n\}$ is a set of ordered $\Ft$ stopping times.
  Applying optional stopping theorem,
  we get a discrete time martingale $\MH(R_k)$
  adapted to $\mathcal{F}_{R_k}$.

  The increment of $\MH(R_k)$ comes from either the counting part or the compensator part,
  which we can bound separately.
  By our construction of $R_k$'s,
  each left-open right-closed bin $(R_k, R_{k+1}]$ satisfies two conditions.
  There is at most one jump from $\sumin N_i(t)$ in the bin at $R_{k+1}$.
  The length of the bin is at most $t^*/n$.
  The increment of the martingale $\MH(t)$ over $(R_k, R_{k+1}]$ is decomposed into two coordinate-wise integrals, a jump minus a compensator,
  $$
  \MH(t) = n^{-1}\sumin \int_{R_k}^{R_{k+1}}  \bfH_i(u) dN_i(u) - n^{-1}\sumin \int_{R_k}^{R_{k+1}} \bfH_i(u) h_i(u) du.
  $$
  With the assumed a.s. upper bound for $\supt\|\bfH_i(t)\|_{\max} \le \KH$,
  we have almost surely the jump of $\MH(t)$ in the bin be bounded by
  $$
  \left\|n^{-1}\sumin \int_{R_k}^{R_{k+1}}  \bfH_i(u) dN_i(u)\right\|_{\max} \le \KH/n.
  $$
  Additionally with the assumed upper bound for $\supi\supt h_i(t) \le \KhA$,
  we have the compensator of $\MH(t)$ increases over the bin by at most
  $$
  \left\| \int_{R_k}^{R_{k+1}} n^{-1}\sumin \bfH_i(u) h_i(u) du \right\|_{\max} \le \KH \KhA (R_{k+1}-R_k) \le \KH \KhA t^*/n.
  $$
  We obtain a uniform concentration inequality for $\MH(R_k)$ by Lemma \ref{lemma:Azuma}
  $$
  \Pr\left(\sup_{k=1,\dots,2n} \|\MH(R_k)\|_{\max} \ge \KH(1+\KhA t^*)x\right)
  \le 2q e^{-nx^2/4}.
  $$
  Remark that the uniform version of Lemma \ref{lemma:Azuma} is the application of Doob's maximal inequality \cite[Theorem 5.4.2, page 213]{Durrett13}.
  For $t \in (R_k, R_{k+1})$, we use the bounded increment derived above
  $$
  \|\MH(t)-\MH(R_k)\|_{\max} \le \left\| \int_{R_k+}^{R_{k+1}+} n^{-1}\sumin \bfH_i(u) h_i(u) du \right\|_{\max} \le \KH \KhA t^*/n.
  $$

  \item Under the additional assumption $\supi\supt\|\bfH_i(t)\|_{\max} = O_p(a_n)$,
we can find $K_{\Phi,\varepsilon}$ for every $\varepsilon>0$ such that
$$
\Pr\left(\supi\supt\|\bfH_i(t)\|_{\max} \le K_{\Phi,\varepsilon}a_n \right)
\ge 1 - \varepsilon/2
$$
for any $n$.
We apply Lemma \ref{lemma:supM} to obtain that event
\begin{align*}
& \left\{  \supt \|\MH(t)\|_{\max} \le K_{\Phi,\varepsilon}a_n\{(1+\KhA t^*)\sqrt{2\log(4q)/n} + \KhA t^*/n\}, \right.
\\
& \qquad \qquad  \left. \supi\supt\|\bfH_i(t)\|_{\max} \le K_{\Phi,\varepsilon}a_n \right\}
\end{align*}
occurs with probability no less than $1-\varepsilon$.
  \end{enumerate}
\end{proof}

\begin{proof}[Proof of Lemma \ref{lemma:KZbar}]
  Notice all $a_i(t)$'s are nonnegative.
  Hence, $\sumin |a_i(t)| = \sumin a_i(t)$.
  We apply H\"{o}lder's inequality for each coordinate
 $$
  \left|\left\{\frac{\sumin a_i(t)\bZ_i(t)^{\otimes l}}{\sumin a_i(t)}\right\}_j\right|
 =
  \left| \sumin \frac{a_i(t)}{\sumin a_i(t)}\left\{\bZ_i(t)^{\otimes l}\right\}_j\right|
   \le \frac{\sumin |a_i(t)|}{\left|\sumin a_i(t)\right|} \supi \left|\left\{\bZ_i(t)^{\otimes l}\right\}_j\right|.
  $$
  Hence, the maximal norm of $\sumin a_i(t)\bZ_i(t)^{\otimes l}$ is bounded by $(\KZ/2)^l$ under \eqref{aseq:Zij}.
  Similar result can be achieved with the sum replaced by the expectation.

  To apply the result above to $\bS{l}(t,\bbeta)/\bSo(t,\bbeta)$, $\btS{l}(t,\bbeta)/\btSo(t,\bbeta)$ and
$\bs{l}(t,\bbeta)/\bso(t,\bbeta)$,
we set $a_i(t)$ as $\omega_i(t)Y_i(t)\RR{i}{t}{}$ and $I(C_i \ge t)Y_i(t)\RR{i}{t}{}$.
\end{proof}

\begin{proof}[Proof of Lemma \ref{lemma:Ytau}]
Since $\{I(X_i \ge t^*), i =1,\dots, n\}$ are i.i.d. Bernoulli random variable,
we may apply Lemma \ref{lemma:Hoeffding}  for lower tail,
$$
\Pr\left(n^{-1}\sumin I(X_i \ge t^*) < \Pr(X_i \ge t^*) - x\right)
\le \exp(-2nx^2).
$$
By \eqref{aseq:denom}, we can find lower bounds for the probability
$$
\Pr(X_i \ge t^*) \ge \Pr(C_i \ge t^*, \infty>T^1_i \ge t^*)
 = G(t^*) \E\{F_1(\infty;\bZ_i)-F_1(t^*;\bZ_i)\} \ge \rstar/\M.
$$
We may relax the inequality at $x = \rstar/(2\M)$ to
$$
\Pr\left(n^{-1}\sumin I(X_i \ge t^*) < \rstar/(2\M)\right)
\le e^{-n\rstar^2/(2\M^2)}.
$$

Because $I(C_i \ge t) \ge I(C_i \ge t^*)$ and $Y_i(t) \ge Y_i(t^*)$,
$\btSo(t;\M)$
is a lower bound for $\btSo(t)$.
The summands in $\btSo(t;\M)$ are i.i.d. uniformly bounded by $\M$.
Thus, we may apply Lemma \ref{lemma:supij} with one-sided version,
$$
\Pr\left(\supk \btSo(t;\M) < \E\{\btSo(t;\M)\} - \M x - \M /n \right) < n e^{-nx^2/2}.
$$
By \ref{Assum:CR}, the expectation has a lower bound
$$
\E\{\btSo(t;\M )\} = G(t^*)\E\left[\{1-F_1(t;\bZ_i)\}\min\{\M ,\RR{i}{t}{o}\}\right]
> \rstar.
$$
We relax the inequality at $x = (\rstar/2 - 1/n)/\M $,
$$
\Pr\left(\supk \btSo(t;\M ) < \rstar \right) < n e^{-n(\rstar-2/n)^2/(8\M ^2)}.
$$
\end{proof}

\begin{proof}[Proof of Lemma \ref{lemma:RR}]
Since $\epsilon_i>1$ implies $T^1_i = \infty$,
the probability of observing a type-2 event conditioning on $\bZ_i(\cdot)$ has an upper bound
\begin{align*}
\Pr\big(\epsilon_i > 1| \bZ_i(\cdot)\big) = &
\exp\left\{-\int_0^\infty \RR{i}{u}{o} \lamT(u) du\right\} \\
\le  & \exp\left\{-\Ketmp x\int_0^\infty I\left(\RR{i}{u}{o} \ge \Ketmp x \right) \lamT(u) du\right\}.
\end{align*}
Hence, we may derive a bound for
$$
\Pr\left(\delta_i\epsilon_i > 1, \supt \RR{i}{t}{o} > \Ketmp \right)
\le \Pr\left(\epsilon_i > 1\Big| \supt \RR{i}{t}{o} > \Ketmp\right)
$$
if we can bound $\int_0^\infty I\left(\RR{i}{u}{o} \ge \Ketmp x \right) \lamT(u) du$ away from zero with a certain $x$ whenever $\RR{i}{t'}{o} > \Ketmp$ for some $t'\in [0,t^*]$.

Under \ref{Assum:CR}, there is an interval $I'$ containing $t'$ of length $\dt$
in which $\bZ_i(\cdot)$ has no jumps.
The variation of linear predictor is bounded
$$
\sup_{t\in I'}\left|\betaot\bZ_i(t) - \betao\bZ_i(t')\right|
\le \rtz \Lz \|\betao\|_\infty \dt.
$$
So, the relative risk $\RR{i}{t}{o}$ is greater than $\Ketmp \exp\{-\rtz \Lz \|\betao\|_\infty \dt\}$ over $I'$.
Hence, we get a lower bound for
$$
\int_0^\infty I\left(\RR{i}{u}{o} \ge \Ketmp\exp\{-\rtz \Lz \|\betao\|_\infty \dt\} \right) \lamT(u) du
\ge \dt \lamlow.
$$
We finish the proof by taking a union bound over $i = 1, \dots, n$.
\end{proof}

\begin{proof}[Proof of Lemma \ref{lemma:KM}]
  Recall that $\Mc_i(t) = I(C_i \le t) - \intt I(C_i \ge u)\lamC(u)du$ is
a counting process martingale adapted to complete data filtration $\cFt$.
The Kaplan-Meier estimator $\hG(t)$ has the martingale representation \cite[p.170 (5.45)]{KalbfleischPrentice02},
$$
M^G(t) = \frac{\hG(t)}{G(t)}-1=n^{-1}\sumin \intt\frac{\hG(u-)I(X_i \ge u)}{G(u)n^{-1}\sumjn I(X_j \ge u)}d\Mc_i(u).
$$
For $\delta_i\epsilon_i>1$ and $t>X_i$,
$$
\omega_i(t) - \tilde{\omega}_i(t) = -\frac{\hG(t)}{\hG(X_i)} M^G(X_i)
+ \frac{G(t)}{G(X_i)} M^G(t),
$$
so we will be able to establish a concentration result for the error from Kaplan-Meier
$$
\left\|n^{-1} \sumin \{\omega_i(t) - \tilde{\omega}_i(t)\}Y_i(t)\RR{i}{t}{o}\bZ_i(t)^{\otimes l}\right\|_{\max}
\le 2\Ke (\KZ/2)^l \supt |M^G(t)|
$$
if we first obtain a concentration result for $\supt |M^G(t)|$.
On event $n^{-1}\sumjn I(X_j \ge u) \ge \rstar/(2\M )$, the integrated functions are $\cFtH$-adapted with uniform bound $2(\M ^/\rstar)^2$.
The hazard $\lamC(t) \le \Lc$ by \ref{Assum:design}.
Hence, we may apply Lemma \ref{lemma:supM} with $x=\sqrt{4\log(2/\varepsilon)/n}$ to obtain
\begin{align}
&\Pr\left(\supt |M^G(t)| < 2(\M /\rstar)^2\left\{(1+\Lc t^*)\sqrt{4\log(2/\varepsilon)/n}+ \Lc t^*/n\right\} \right)
\\
&\qquad \qquad \le \Pr(\eventrM \cap \eventKe) - \varepsilon.
\end{align}
\end{proof}

\begin{proof}[Proof of Lemma \ref{lemma:Sk}]
A sharper inequality is available if $\bZ_i$'s are not time-dependent.
We may exploit the martingale structure of $\bDelta{l}(t)/G(t)$.
With general time-dependent covariates,
we would decompose the approximation error $\bDelta{l}(t)$ into two parts,
the error from Kaplan-Meier estimate $\hG(t)$ and
the error from missingness in $C_i$'s among the type-2 events.

Define the indicator $\ipw_i(t) = I(t > X_i)I(\delta_i\epsilon_i > 1)$.
Since $\{\omega_i(t) - I(C_i \ge t)\}Y_i(t)$ is non-zero only when $\ipw(t) = 1$,
we may alternatively write
$$
\bDelta{l}(t) = n^{-1}\sumin \{\omega_i(t) - I(C_i \ge t)\}\ipw_i(t) \RR{i}{t}{o} \bZ_i(t)^{\otimes l}.
$$
We may use the upper bound $\supi\supt\left|\ipw_i(t) \RR{i}{t}{o}\right| \le \Ke$ on $\eventKe$.
By Lemma \ref{lemma:KM},
$$
\left\|n^{-1} \sumin \{\omega_i(t) - \tilde{\omega}_i(t)\}\ipw_i(t)\RR{i}{t}{o}\bZ_i(t)^{\otimes l}\right\|_{\max}
\le \Ke (\KZ/2)^l \CKM
$$
on $\eventKe \cap \eventKM$.

Define the error from missingness in $C_i$'s among the type-2 events as
$$
\btDelta{l}(t) = n^{-1} \sumin \{\tilde{\omega}_i(t)- I(C_i \ge t)\}\ipw_i(t)\RR{i}{t}{o}\bZ_i(t)^{\otimes l}.
$$
Since $\E\{r_i(t)| T_i\} = G(t\wedge T_i)$, \cite{FineGray99} has shown that
$$
\E\{\tilde{\omega}_i(t)| T_i\} = \E\{I(C_i \ge t)| T_i\} = G(t).
$$
Applying tower property, we have $\E\left\{\btDelta{l}(t)\right\} = \bzero$.
Hence, we can apply Lemma \ref{lemma:supij} with $x=\sqrt{2\log(2np^l/\varepsilon)/n}$
$$
\Pr\left(\supk\left\|\btDelta{l}\left(\T{k}\right)\right\|_{\max} \le \Ke (\KZ/2)^l \left\{\sqrt{2\log(2np^l/\varepsilon)/n}+ 1/n\right\} \right)
\ge \Pr(\eventrM \cap \eventKe) - \varepsilon.
$$
This finishes the proof of the first result.

We prove the other result by decomposing the differences into terms with $\bDelta{l}(t)$,
$$
 \frac{\bS{l}(t,\betao)}{\bSo(t,\betao)}
            -\frac{\btS{l}(t,\betao)}{\btSo(t,\betao)} = \frac{1}{\btSo(t,\betao)}\bDelta{l}(t)-\frac{\bS{l}(t,\betao)}{\bSo(t,\betao)\btSo(t,\betao)}
\Delta^{(0)}(t).
$$
$\bS{l}(t,\betao)/\bSo(t,\betao)$ is the weighted average of $\bZ_i(t)^{\otimes l}$, so its maximal norm is bounded by $(\KZ/2)^l$.
    On the event $\eventrM$,
    $$
    \left\|\frac{\bS{l}(t,\betao)}{\bSo(t,\betao)}
            -\frac{\btS{l}(t,\betao)}{\btSo(t,\betao)} \right\|_\infty \le
            \frac{2}{\rstar }\|\bDelta{l}(t)\|_\infty
            + \frac{\KZ^l}{2^{l-1}\rstar }|\Delta^{(0)}(t)|.
    $$
    We can simply plug in the bounds and tail probabilities for $\Delta^{(0)}\left(\T{k}\right)$ and $\bDelta{1}\left(\T{k}\right)$ in \eqref{def:eventx}.
\end{proof}

\begin{proof}[Proof of Lemma \ref{lemma:opNorm}]
\begin{enumerate}[label = (\roman*)]
\item
By \ref{Assum:design} and \ref{Assum:betaI},
we have $\left\|\RR{i}{t}{o}\bZ_i(t)^{\otimes l}\right\|_{\max} \le (\KZ/2)^l e^{\Kb} \asymp 1$.
Thus, all terms involved are bounded.
Moreover, $\RR{i}{t}{o}\bZ_i(t)^{\otimes l}$ jumps only at the jumps of $\Nz_i(t)$ by \ref{Assum:CRI}.
Define the outer product of arrays $\bfu \in \R^{p_1\times\dots \times p_d}$ and
$\bfv \in \R^{q_1\times \dots \times q_{d'}}$ as
$$
\bfu \otimes \bfv \in \R^{p_1\times\dots \times p_d\times q_1 \times\dots \times q_{d'}},
\quad
(\bfu \otimes \bfv)_{i_1,\dots,i_{d+d'}} = \bfu_{i_1,\dots, i_d}\times \bfv_{i_{d+1},\dots,i_{d+d'}}.
$$
Between two consecutive jumps of $\Nz_i(t)$,
\begin{align*}
 & \left\|\frac{d}{dt} \RR{i}{t}{o}\bZ_i(t)^{\otimes l}\right\|_{\max}
\\
= &
\left\|\RR{i}{t}{o}\bZ_i(t)^{\otimes l} \betaot \dZ_i(t) + I(l > 0)\RR{i}{t}{o} l \bZ_i(t)^{\otimes l-1}\otimes \dZ_i(t) \right\|_{\max} \\
\le & e^{\Kb} \{ (\KZ/2)^l\LzI+I(l>1)(\KZ/2)^{l-1}\LzI \}
\asymp 1.
\end{align*}
Hence, $\RR{i}{t}{o}\bZ_i(t)^{\otimes l}$ satisfies the continuity condition for Lemma \ref{lemma:supt}.

Like in Lemma \ref{lemma:Sk},
we first replace $\omega_i(t)$ by $\tilde{\omega}_i(t) = r_i(t)G(t)/G(X_i\wedge t)$.
Denote $\btDelta{l}(t) = n^{-1}\sumin \{\tilde{\omega}_i(t)-I(C_i \ge t)\}Y_i(t) \RR{i}{t}{o}\bZ_i(t)^{\otimes l}$.
By Lemma \ref{lemma:KM},
$\supt\|\bDelta{l}(t) - \btDelta{l}(t)\|_{\max} = O_p\left(n^{-1/2}\right)$.
Then, we apply Lemma \ref{lemma:supt} to the i.i.d. mean zero process $\btDelta{l}(t)$,
$$
\supt\|\btDelta{l}(t)\|_{\max} = O_p\left(\sqrt{\log(np^l\Kz)/n}\right).
$$
Similarly,
$$
\supt\|\btS{l}(t,\betao) - \bs{l}(t,\betao)\|_{\max} = O_p\left(\sqrt{\log(np^l\Kz)/n}\right).
$$

Finally, we extend to results to the quotients by decomposition
$$
\frac{\bS{1}(t,\betao)}{\bSo(t,\betao)} - \frac{\btS{1}(t,\betao)}{\btSo(t,\betao)}
= \frac{1}{\btSo(t,\betao)}\bDelta{l}(t)-\frac{\bS{l}(t,\betao)}{\bSo(t,\betao)\btSo(t,\betao)}
\Delta^{(0)}(t).
$$
The denominators are bounded away from zero by Lemma \ref{lemma:YtauI} by choosing $\M  = e^{\Kb}$.

\item
  First, we show that $\dS_i(t)$ is related to the martingales $\Mc_i(t)=\Nc_i(t) - \intt \{1-\Nc_i(u-)\}\lamC(u) du$.
  $\dS_i(t)$ is non-zero only after an observed type-2 event.
  To simplify notation, we define the indicator for non-zero $\dS_i(t)$, $\ipw_i(t) = r_i(t)Y_i(t)I(t>X_i) = I(\delta_i\epsilon_i>1)I(t>X_i)$.

  Denote the Nelson-Aalen type estimator for censoring cumulative hazard as
  $$
  \NAc(t) = \sumin \intt \Big\{\sumjn I(X_j \ge u)\Big\}^{-1} I(X_i \ge u)d \Nc_i(u).
  $$
  Define $
R_i(t) = \hG(t)/\hG(X_i) - 1 + \int_{X_i}^t\hG(u-)d\NAc(u)/\hG(X_i)
$.
Let $c_k$ and $c_{k+1}$ be two consecutive observed censoring times greater than $X_i$.
The increment $R_i(c_{k+1}) - R_i(c_k)$ is in fact
$$
 \frac{\hG(c_k)}{\hG(X_i)}\left\{\frac{\sumjn I(X_j \ge c_{k+1})-1}{\sumjn I(X_j \ge c_{k+1})} - 1 + \frac{1}{\sumjn I(X_j \ge c_{k+1})}\right\}
 = 0.
$$
  For $t > X_i$, we have $R_i(t) = 0 $. Thus,
\begin{align}
\dS_i(t) = &  \{\hG(t)/\hG(X_i) - 1 + \Nc_i(t) - \Nc_i(X_i) - R_i(t)\}\ipw_i(t) \notag \\
 = & \int_{X_i}^t \ipw_i(u) d M_i^c(u) - \int_{X_i}^t \omega_i(u-)\ipw_i(u)\frac{\sumjn I(X_j \ge u) dM^c_j(u)}{\sumjn I(X_j \ge u)} + \notag \\
& + \int_{X_i}^t \{I(C_i \ge u) - \omega_i(u-)\}\ipw_i(u) \lamC(u) du.  \label{eq:ode}
\end{align}
Notice $\ipw_i(t)$ does not change beyond $X_i$ if $C_i > X_i$, i.e. an event is observed.
Since $\lamC(u)\le \LcI< \infty$, we may modify the integrand at countable many points without changing the integral
$$
\int_{X_i}^t \{I(C_i \ge u) - \omega_i(u-)\}\ipw_i(u) \lamC(u) du
= - \int_{X_i}^t \dS_i(u) \lamC(u) du.
$$
Hence, \eqref{eq:ode} gives an first order linear integral equation for $\dS_i(u)$.
The general solution to the related homogeneous problem
$$
\dS_i(t) = -  \int_{X_i}^t \dS_i(u)\lamC(u)du, \quad \dS_i(X_i) =0
$$
has only one unique solution $\dS_i(t)=0$.
Thus, we only need to find one specific solution to \eqref{eq:ode}.
Define an integral operator
$
I \circ f = \int_{X_i}^t f(u) \lamC(u)du.
$
Then, the solution to
$
f(t) = g(t) - I \circ f (t)
$
can be written as
$$
f(t) = (1-I+I^2-I^3 +\dots)\circ g (t) \triangleq e^{-I}\circ g (t).
$$
By inductively using integration by parts, we are able to calculate
$$
I^n \circ g (t) = \frac{1}{n!} \sum_{k = 1}^n \binom{n}{k} (-1)^k{\LamC(t)}^{n-k}\int_{X_i}^t{\LamC(u)}^k d g (t).
$$
Hence, the solution can be calculated as the series
\begin{align*}
f(t) = &
  \sum_{n=1}^\infty (-1)^n  I^n \circ g(t)
=  \sum_{n=1}^\infty \sum_{k = 1}^n  \frac{\{-\LamC(t)\}^{n-k}}{(n-k)!}\int_{X_i}^t\frac{{\LamC(u)}^k}{k!} d  g (u) \\
  = & \sum_{k = 1}^\infty \int_{X_i}^t\frac{{\LamC(u)}^k}{k!} d  g (u) \sum_{n=k}^\infty \frac{\{-\LamC(t)\}^{n-k}}{(n-k)!}
  =   G(t) \int_{X_i}^tG(u)^{-1} d  g(u).
\end{align*}
Applying to \eqref{eq:ode}, we get
$$
\dS_i(t) =  G(t) \int_{X_i}^t G(u)^{-1} d\Md_i(u),
$$
with a $\cFt-$ martingale
$$
\Md_i(t) = \intt I(C_i \ge u)\ipw_i(u) d\Mc_i(u) - \intt \omega_i(u-)\ipw_i(u)\frac{\sumjn I(X_j \ge u) dM^c_j(u)}{\sumjn I(X_j \ge u)}.
$$

Now, we use the martingale structure to prove the Lemma.
 Denote the $\ccFt$ martingale
  $$
  \Mg(t) =n^{-1}\sumin \intt G(u) \RR{j}{u}{o} \bfg(u) I(C_i \ge u) d\MI_i(u).
  $$
  $\Mg(t)$ satisfies the condition for Lemma \ref{lemma:supM}.
  Hence, we have $\supt \|\Mg(t)\|_{\max} = O_p\left(\sqrt{\log(q')/n}\right)$.
  Also define
  $$\dtS_i(t) = \{\tilde{\omega}_i(t)-I(C_i > t)\}Y_i(t). $$
  By Lemma \ref{lemma:KM}, $\supi\supt|\dS_i(t)-\dtS_i(t)| = O_p\left(n^{-1/2}\right)$.
The total variation of each $\dS_i(t)$ is at most $2$.
Hence, we can apply integration by parts
to \eqref{eq:opSk},
\begin{align*}
& G^{-1}(t^*) \Mg(t^*-)\otimes n^{-1/2}\sumjn \dS_j(t^*-)\bfh(\bZ_j(t^*-))
\\
& -  n^{-1/2}\sumjn \inttao \Mg(t)\otimes \bfh(\bZ_j(t)) d\Md_j(t)\\
& - n^{1/2} \inttao \Mg(t)\otimes G^{-1}(t)n^{-1}\sumjn\dS_j(t) d \bfh(\bZ_j(t)) \\
\triangleq & I_1-I_2-I_3.
\end{align*}

We have shown that $|\Mg(t^*-)| = O_p\left(\sqrt{\log(q')/n}\right) $ and $\sup_{j=1,\dots,n} |\dS_j(t^*-)-\dtS_j(t^*-)|= O_p\left(n^{-1/2}\right)$.
By assumption, $\|\bfh(\bZ_j(t^*-))\|_{\max} \le \Kh \asymp 1$.
As a result, we may replace the $\dS_i(t)$ in $I_1$ by $\dtS_i(t)$ with with an $O_p\left(\sqrt{\log(q')/n}\right)$ error.
Since $\dtS_j(t^*-)\bfh(\bZ_j(t^*-))$'s are i.i.d. mean zero random variables,
$$\|n^{-1}\sumjn \dtS_j(t^*-)\bfh(\bZ_j(t^*-))\|_{\max} = O_p\left(\sqrt{\log(q)/n}\right)$$
 by Lemma \ref{lemma:Hoeffding}.
Multiplying the rates together, we get $\|I_1\|_{\max} = O_p\left(\sqrt{\log(q)\log(q')/n}\right) =o_p(1)$.

$I_2$ can be expanded as
\begin{align*}
&n^{-1/2}\sumjn \inttao G(t)^{-1}\Mg(t) \bigg\{
   I(C_j\ge t) \ipw_j(t)h(\bZ_j(t))
    \biggl.
   \\
&   \biggl.  \qquad \qquad -  \frac{\sumkn \omega_k(t-)\ipw_k(t)h(\bZ_k(t))}{\sumkn I(X_k \ge t)}I(X_j\ge t)\bigg\} d\Mc_j(t)
\end{align*}
By Lemma \ref{lemma:YtauI}, $n\left\{\sumkn I(X_k \ge t)\right\}^{-1} = O_p(1)$.
The integrand in $I_2$ is the product of $\Mg(t)$ and a $O_p(1)$ term.
Hence, we can apply Lemma \ref{lemma:opM} to get
$\|I_2\|_{\max} = O_p\left(\sqrt{\log(q')\log(qq')/n}\right)
 = o_p(1)$.

By \ref{Assum:CRI},
we may further expand $I_3$ into
\begin{align*}
&n^{1/2} \inttao \Mg(t) \otimes G^{-1}(t)n^{-1}\sumjn\dS_j(t) \nabla \bfh(\bZ_j(t))^\top \dZ_j(t)dt \\
& + n^{1/2} \inttao \Mg(t) \otimes G^{-1}(t)n^{-1}\sumjn\dS_j(t) \triangle \bfh(\bZ_j(t))d\Nz_j(t) \\
\triangleq & I_3' + I_3^{\prime\prime},
\end{align*}
where $\triangle \bfh(\bZ_j(t))= \bfh(\bZ_j(t)) - \bfh(\bZ_j(t-))$.
By assumption on $h(\bZ)$ and \ref{Assum:CRI}, $|\nabla \bfh(\bZ_j(t))^\top \dZ_j(t)|$ and $\triangle \bfh(\bZ_j(t))$ are bounded by
$\Lh \LzI$ and $ \Lh K$, respectively.
With $\supt|\Mg(t)| = O_p\left(\sqrt{\log(q')/n}\right)$ and $\Nz_j(t^*)<\Kz=o(\sqrt{n/(\log(p)\log(n))})$,
we may replace the $\dS_j(t)$'s by $\dtS_j(t)$'s with an $o_p(1)$ error.
Each $\dtS_j(t) \nabla \bfh(\bZ_j(t))^\top \dZ_j(t)$ has mean zero and at most $\Kz+1$ jumps,
and it is $(\Lh\LzI + \Kh\LzII)$-Lipschitz between two consecutive jumps under \ref{Assum:CRI} and conditions on $\bfh(\bz)$.
By applying Lemma \ref{lemma:supt}, we get
$$
\supt \left\|n^{-1}\sumjn\dtS_j(t) \nabla \bfh(\bZ_j(t))^\top \dZ_j(t)\right\|_{\max}
= O_p(\sqrt{\log(nq)/n}).
$$
Hence, $\|I_3'\|_{\max} = O_p(\sqrt{\log(q')\log(nq)/n})+o_p(1) = o_p(1)$.
By applying Lemma \ref{lemma:supij} to
$$\{\dtS_j(t)\triangle h(\bZ_j(t)), \Nz_j(t): j= 1,\dots, n\},$$
we get at the jumps of $\Nz_i(t)$'s, at the $t_{ik}$, satisfy
\begin{align*}
&\supi\supk \left|n^{-1}\sumjn\dS_j(t_{ik}) \triangle \bfh(\bZ_j(t_{ik}))\right|
\\
& \qquad \qquad = O_p(\sqrt{\log(n\Kz q)/n}) = O_p(\sqrt{\log(nq)/n}).
\end{align*}
Hence, $\|I_3^{\prime\prime}\|_{\max} = O_p\left(\Kz \sqrt{\log(nq)\log(q')/n} \right) = o_p(1)$.
This completes the proof.

\item
    Define $\bbeta_r = \betao + r \{\tbeta - \betao\}$ and
  $h_j(r; t) =  \bZbar_j(t,\bbeta_r)$.
  The subscript $j$ means the j-th element of corespondent vector.
  By mean-value theorem, we have some $r \in (0,1)$ such that
  \begin{align*}
  h_j(1; t) - h_j(0;t) = & \left( \{\tbeta - \betao\}^\top \frac{\bS{2}(t,\bbeta_r)\bSo(t,\bbeta_r) - \bS{1}(t,\bbeta_r)^{\otimes 2}}{\bSo(t,\bbeta_r)^2} \right)_j \\
  = & \left( \{\tbeta - \betao\}^\top \sumin \frac{\omega_i(t)Y_i(t)e^{\bbeta_r^\top \bZ_i(t)}}{n\bSo(t,\bbeta_r)} \{\bZ_i(t)-\bZbar_j(t,\bbeta_r)\}^{\otimes 2}\right)_j
  \end{align*}
  Since each $\left\|\{\bZ_i(t)-\bZbar_j(t,\bbeta_r)\}^{\otimes 2}\right\|_{\max} \le \KZ^2$ under \ref{Assum:design},
  their weighted average
  $$
  \left\|\sumin \frac{\omega_i(t)Y_i(t)e^{\bbeta_r^\top \bZ_i(t)}}{n\bSo(t,\bbeta_r)} \{\bZ_i(t)-\bZbar_j(t,\bbeta_r)\}^{\otimes 2}\right\|_{\max} \le \KZ^2.
  $$
  Hence, we have shown that
  $$
  \supt \|\bZbar(t,\betao) - \bZbar(t,\tbeta)\|_\infty \le  \|\tbeta - \betao\|_1 \KZ^2 = O_p(\|\tbeta - \betao\|_1).
  $$

  By a similar argument, we can show for some $r \in (0,1)$
  $$
  \frac{\RR{i}{t}{o}}{\bSo(t,\betao)} - \frac{\RRt{i}{t}}{\bSo(t,\tbeta)}
  = \frac{e^{\bbeta_r^\top \bZ_i(t)}(\tbeta - \betao)^\top}{\bSo(t,\bbeta_r)} \sumjn \frac{\omega_j(t)Y_j(t)e^{\bbeta_r^\top \bZ_j(t)}}{n\bSo(t,\bbeta_r)} \{\bZ_i(t)-\bZ_j(t)\}.
  $$
  On event $\{ \|\tbeta - \betao\|_1 \le \Kb, \, n^{-1}\sumin I(X_i \ge t^*)\ge e^{-\Kb}\rstar/2\}$,
  we have
  $$
  \inf_{t\in[0,t^*]}\bSo(t,\tbeta) > r^*/2*e^{2 \Kb},
  \quad \inf_{t\in[0,t^*]}\bSo(t,\betao) > r^*/2*e^{\Kb}.
  $$
  Hence,
  $$
  |\RR{i}{t}{o}/\bSo(t,\betao)-\RRt{i}{t}/\bSo(t,\tbeta)| \le \|\tbeta - \betao\|_1  2\KZ e^{4\Kb }e^{\Kb}/\rstar = O_p\left(\|\tbeta - \betao\|_1\right).
  $$
  The event occurs with probability tending to one because we have$\|\tbeta - \betao\|_1 = o_p(1)$ from Theorem \ref{thm*:initial rate} and
  $\supt|\bSo(t,\betao)^{-1}| = O_p(1)$ from Lemma \ref{lemma:YtauI}.
\end{enumerate}
\end{proof}

\begin{proof}[Proof of Lemma \ref{lemma:YtauI}]
Consider the event
$$
\eventrM^* = \left\{ n^{-1} \sumin I(X_i \ge t^*)I(\epsilon_i=1) \ge e^{-\Kb}\rstar/2\right\}.
$$
Each $I(X_i \ge t^*)I(\epsilon_i=1)$ is i.i.d. with expectation $G(t^*) \E [\{ F_1(\infty;\bZ) - F_1(t^*;\bZ)\}]$.
Applying Lemma \ref{lemma:Hoeffding} under \eqref{aseq:denom} and \eqref{aseq:bZ}
from \ref{Assum:design} and \ref{Assum:betaI}, we get
that $\eventrM$ occurs with probability $1-e^{-n e^{-2\Kb}\rstar^2}$.

Apparently, we have $I(X_i \ge t^*) \ge I(X_i \ge t^*)I(\epsilon_i=1)$.
Morevoer, $\bSo(t,\betao)$ and $\btSo(t,\betao)$ are both lower bounded by
$n^{-1} \sumin I(X_i \ge t^*) e^{-\Kb}$.

On $\eventrM$,
$\supt|n/\{\sumin I(X_i \ge t^*)\}| \le 2e^{\Kb}/\rstar$ and
$$
\max\left\{\supt|\bSo(t,\betao)^{-1}|, \supt|\btSo(t,\betao)^{-1}|\right\} \le 2e^{\Kb}e^{\Kb}/\rstar.
$$
\end{proof}

\begin{proof}[Proof of Lemma \ref{lemma:gamma-oracle}]
To simplify notation, wherever possible we will use
$\hGr_j(\bfgr) = \Gr_j(\bfgr, \bini)$.
\begin{enumerate}[label = (\roman*)]
  \item We want to prove that for all $j = 1, \ldots, p$,  the differences $\tilde{\bfgr}_j := \hgr_j - \gro_j$ belong to a certain convex cone.

It follows from the KKT conditions that, for $l = 1, \ldots, p-1$,
\[
\begin{cases}
\frac{\partial \hGr_j(\hgr_j)}{\partial \bfgr_{j,l}} + \lambda_j\mathrm{sgn}(\hgr_{j,l}) = 0 & \mbox{if } \hgr_{j,l} \neq 0;\\
\biggl|\frac{\partial \hGr_j(\hgr_j)}{\partial \bfgr_{j,l}}\biggr| \leq \lambda_j & \mbox{if } \hgr_{j,l} = 0.
\end{cases}
\]
Denote $\mathcal{O}_j := \bigl\{l \in \{1, \ldots, p-1\}: \, \gro_{j, l} \neq 0\bigr\}$ and $\mathcal{O}^c_j := \{1,\ldots,p-1\}\setminus\mathcal{O}_j$.  For $\xi_j > 1$, it follows from the KKT conditions above that on the event \[
\Omega_0 := \{\|\dhGr_j(\gro_j)\|_{\infty} \leq (\xi_j-1)\lambda_j/(\xi_j+1)\},
\]
with $\bar \bfgr_j = \alpha \hgr_j + (1-\alpha) \gro_j$, $\alpha \in (0,1)$
\begin{align*}
\label{Eq:KKT}
0 &\leq
 2\tilde{\bfgr}_j^{\top}\ddhGr_j(\bar{\bfgr}_j)  \tilde{\bfgr}_j \nonumber \\
 &=\tilde{\bfgr}_j^{\top}\bigl\{\dhGr_j(\hgr_j) -\dhGr_j(\gro_j) \bigr\} \nonumber \\
& = \sum_{l \in \mathcal{O}_j^c}\tilde{\bfgr}_{j,l}\frac{\partial \hGr_j(\hgr_j)}{\partial \bfgr_{j,l}}
+ \sum_{l\in \mathcal{O}_j} \tilde{\bfgr}_{j,l}\frac{\partial \hGr_j(\hgr_j)}{\partial \bfgr_{j,l}}
- \tilde{\bfgr}_j^{\top} \dhGr_j(\gro_j) \nonumber \\
& \leq -\lambda_j\sum_{l\in \mathcal{O}^c_j}\hgr_{j,l}\mathrm{sgn}(\hgr_{j,l})
 + \lambda_j\sum_{l \in \mathcal{O}_j}|\tilde{\bfgr}_{j,l}|
 + \frac{(\xi_j-1)\lambda_j}{\xi_j+1}\bigl\|\tilde{\bfgr}_{j, \mathcal{O}_j}\bigr\|_1
 + \frac{(\xi_j-1)\lambda_j}{\xi_j+1}\bigl\|\tilde{\bfgr}_{j, \mathcal{O}_j^c}\bigr\|_1 \nonumber \\
& = -\frac{2\lambda_j}{\xi_j+1}\bigl\|\tilde{\bfgr}_{j, \mathcal{O}_j^c}\bigr\|_1 + \frac{2\xi_j\lambda_j}{\xi_j+1}\bigl\|\tilde{\bfgr}_{j, \mathcal{O}_j}\bigr\|_1.
\end{align*}

  \item
Let $\bfv = \tilde{\bfgr}/\|\tilde{\bfgr}\|_1$ be the $l_1$-standardized direction for $\tilde{\bfgr} = \hgr - \gro$.
By part \ref{part:gamma-cone} and convexity of $\Gr_j$ in $\bfgr_j$, any $x \in (0, \|\tilde{\bfgr}\|_1]$ satisfies
$$
\bfv^\top\left\{ \dhGr_j(\bfgr^*+x\bfv) -\dhGr_j(\bfgr^*) \right\} \le
-\frac{2\lambda_j}{\xi_j+1}\|\bfv_{\Ocal_j^c}\|_1 + \frac{2\xi_j\lambda_j}{\xi_j+1}\bigl\|\bfv_{\Ocal_j}\|_1.
$$
We relax the inequality about $x$ above to establish an upper bound for $\|\tilde{\bfgr}\|_1$.
By the definition of $\kappa_j$, the left hand side can be bounded by
$$
\bfv^\top\left\{ \dhGr_j(\bfgr^*+x\bfv) -\dhGr_j(\bfgr^*) \right\}
 =  x \bfv^\top \ddhGr_j(\bfgr^*) \bfv
 \ge \frac{x \|\bfv_{\Ocal_j}\|_1^2  \compgr}{\sgr_j}.
$$
The right hand side can be bounded using the complete square $\{\|\bfv_{\Ocal_j}\|_1 - 2/(\xi_j+1)\}^2$,
$$
-\frac{2\lambda_j}{\xi_j+1}\|\bfv_{\Ocal_j^c}\|_1 + \frac{2\xi_j\lambda_j}{\xi_j+1}\bigl\|\bfv_{\Ocal_j}\|_1
= 2\lambda_j\|\bfv_{\Ocal_j}\|_1 - \frac{2\lambda_j}{\xi_j+1}
\le \lambda_j(\xi_j+1) \|\bfv_{\Ocal_j}\|_1^2.
$$
Combining the bounds for both sides in the inequality,
we get an upper bound for
$\|\tilde{\bfgr}\|_1 $.
\end{enumerate}
\end{proof}

\begin{proof}[Proof of Lemma \ref{lemma:gamma-score}]
  We define
  $$
  \tGr_j(\bfgr) = n^{-1}\sumin \inttao \{Z_{ij}(t)-\mu_j(t)-\bfgr^\top \bZ_{i,-j}(t) + \bfgr^\top \bmu_{-j}(t)\}^2 d\Nobs_i(t).
  $$
  By Lemmas \ref{lemma:opNorm} and \ref{lemma:bTh},
  $$\maxj\|\dhGr_j(\gro_j,\bini) - \dtGr_j(\gro_j)\|_\infty = O_p\left(\|\bini-\betao\|_1+\sqrt{\log(p)/n}\right).$$
  $\dtGr_j(\gro_j)$ is the average of i.i.d. vectors with mean $\dbGr_j(\gro_j) = \bzero$
  and maximal bound $K^2(1+\Kgr)$.
  We can apply Lemma \ref{lemma:Hoeffding} to the matrix $(\dtGr_1(\gro_1),\dots, \dtGr_p(\gro_p))$
  to get
\begin{align*}
  \maxj \|\dtGr_j(\gro_j)\|_\infty & = \|(\dtGr_1(\gro_1),\dots, \dtGr_p(\gro_p))\|_{\max}
  \\
  &= O_p(\sqrt{\log(p^2)/n})= O_p(\sqrt{\log(p)/n}).
  \end{align*}
\end{proof}

\begin{proof}[Proof of Lemma \ref{lemma:Hess}]
\begin{enumerate}[label=(\roman*)]
  \item   We define
  $$
  \tHess = n^{-1} \sumin \inttao \{\bZ_i(t)-\bmu(t)\}^{\otimes 2} d \Nobs_i(t).
  $$
  The total variation of each $\Nobs_i(t)$ is at most $1$.
  By Lemma \ref{lemma:opNorm},
  we have
  $$\supt\|\bZbar(t,\bini)-\bmu\|_\infty = O_p\left(\|\bini-\betao\|_1+\sqrt{\log(p)/n}\right).$$
  Hence,
  $$
  \|\hHess - \tHess\|_{\max} \le 2K O_p\left(\|\bini-\betao\|_1+\sqrt{\log(p)/n}\right) = O_p\left(\|\bini-\betao\|_1+\sqrt{\log(p)/n}\right).
  $$
  Now, $\tHess$ is average of i.i.d. with mean $\Hess$ and bounded maximal norm $K^2$.
  We apply Lemma \ref{lemma:Hoeffding} with union bound,
  $$
  \Pr\left(\|\tHess - \Hess\|_{\max} \ge K^2 x\right) \le 2p^2e^{2nx^2}.
  $$
  Choosing $x = \sqrt{\log(2p^2/\varepsilon)/(2n)}$, we have $\|\tHess - \Hess\|_{\max} = O_p(\sqrt{\log(p)/n})$.

  \item We alternatively use the following form
  $$
  \ddm(\bbeta) = n^{-1}\sumin \inttao \left\{n^{-1}\sumjn \frac{\omega_j(t)Y_j(t)\RR{j}{t}{}}{\bSo(t,\bbeta)}\bZ_i(t)^{\otimes 2} - \bZbar(t,\bbeta)^{\otimes 2}\right\} d\Nobs_i(t).
  $$
  By Lemma \ref{lemma:opbeta}, we have
  $$
  \|\ddm(\tbeta) - \ddm(\betao)\|_{\max} = O_p(\|\tbeta-\betao\|_1).
  $$
  We also have a similar form for
  $$
    \ddtm(\bbeta) = n^{-1}\sumin \inttao \left\{n^{-1}\sumjn \frac{I(C_j \ge t)Y_j(t)\RR{j}{t}{}}{\btSo(t,\bbeta)}\bZ_i(t)^{\otimes 2} - \bZtil(t,\bbeta)^{\otimes 2}\right\} d\Nobs_i(t).
  $$
  By Lemma \ref{lemma:opbmu}, we have
  $$
  \|\ddm(\betao) - \ddtm(\betao)\|_{\max} = O_p\left(\sqrt{\log(p)/n}\right).
  $$
  Finally, we use the martingale property of
  \begin{align*}
  &  \ddtm(\betao)-\tHess =   n^{-1}\sumin \inttao \left\{ \frac{\btS{2}(t,\betao)}{\btSo(t,\betao)} - \bZtil(t,\betao)^{\otimes 2}\right\} I(C_i\ge t)d\MI_i(t) \\
    &
    - n^{-1}\sumin \inttao \{\bZ_i(t)-\bZtil(t,\betao)\}^{\otimes 2} I(C_i \ge t) d\MI_i(t)\\
    & + n^{-1}\sumin \inttao \left[\{\bZ_i(t)-\bZtil(t,\betao)\}^{\otimes 2}
    - \{\bZ_i(t)-\bmu(t)\}^{\otimes 2}\right] I(C_i \ge t) d\Nobs_i(t)
  \end{align*}
  under filtration $\ccFt$. The integrands in the first two martingale terms are bounded by $K^2$.
  Hence, we can apply Lemma \ref{lemma:opM} to obtain that their maximal norms are both $O_p\left(\sqrt{\log(p)/n}\right)$.
  We apply Lemma \ref{lemma:opbmu} to the integrand of the third term, equivalently expressed as
  $$
  \{\bmu(t)-\bZtil(t,\betao)\}\{\bZ_i(t)-\bZtil(t,\betao)\}^\top
    + \{\bZ_i(t)-\bmu(t)\}\{\bmu(t)-\bZtil(t,\betao)\}^\top.
  $$
  Therefore, we obtain $\|\ddtm(\betao)-\tHess\|_{\max} = O_p\left(\sqrt{\log(p)/n}\right)$.

  We put the rates together by the triangle inequality.
  \end{enumerate}
\end{proof}

\begin{proof}[Proof of Lemma \ref{lemma:gamma-compatibility}]
The proof is similar to that of Lemma \ref{lemma:kappa}.
  Define the compatibility factor for $\conegrj$ and symmetric matrix $\bfH$ as
  $$
  \kappa_j(\xi_j,\Ocal_j; \bfH) =  \sup_{0\neq \bfg \in \conegrj}\frac{\sqrt{\sgr_j \bfg^\top \bfH\bfg}}{\|\bfg_{\Ocal_j}\|_1}.
  $$
  Apparently, $\compgr=  \kappa_j\big(\xi_j,\Ocal_j; \ddGr(\gro,\bini)\big)$.
  Notice that
  $$
  \ddGr(\gro,\bini) = n^{-1}\sumin \inttao \{\bZ_{i,-j}(t) - \bZbar_{-j}(t,\bini)\}^{\otimes 2} d\Nobs_i(t) = \hHess_{-j,-j},
  $$
  where $\hHess_{-j,-j}$ is a $\hHess$ dropping its $j$th row and column.
  By Lemma 4.1 in \cite{HuangEtal13} (for a similar result, see \cite{vdGeerBuhlmann09} Corollary 10.1),
    $$
    \compgr^2 = \kappa_j^2\big(\xi_j, \Ocal_j;\hHess_{-j,-j}\big)
     \ge \kappa_j^2(\xi_j, \Ocal_j;\Hess_{-j,-j}) - \sgr_j (\xi_j+1)^2 \|\Hess_{-j,-j}-\hHess_{-j,-j}\|_{\max}.
    $$
    For any non-zero $\bfg \in \R^{p-1}$, let $\bfg^*$ be its embedding into $\R^p$ defined as
    $$
    g^*_k = \left\{ \begin{array}{cc}
    g_k & k<j\\
    0   & k = j\\
    g_{k-1} & k>j\end{array}\right.
    $$
    Then, we may establish a lower bound for the smallest eigenvalue of $\Hess_{-j,-j}$ by \ref{Assum:HessI}
    $$
    \inf_{0\neq \bfg \in \R^{p-1}} \bfg^\top \Hess_{-j,-j} \bfg =
    \inf_{0\neq \bfg \in \R^{p-1}} \bfg^{*\top} \Hess \bfg^* \ge \rhosigI \|\bfg\|_2^2.
    $$
    Hence, $\inf_{j=1,\dots,p} \kappa_j^2(\xi_j, \Ocal_j;\Hess_{-j,-j}) \ge \rhosigI$.
    Using the result in Lemma \ref{lemma:hHess} under \ref{Assum:norm rate},
    we have
    $$
    \inf_{j=1,\dots,p} \compgr^2 \ge \rhosigI - \|\Hess-\hHess\|_{\max} \sgr_{\max}\max_{j=1,\dots,p} (\xi_j+1)^2 = \rhosigI - o_p(1).
    $$
    Therefor, if $\xi_{\max} \asymp 1$, we must have that
    $\left\{\inf_j \compgr^2 \ge \rhosigI/2\right\}$ occurs with probability tending to one.
\end{proof}

\bibliographystyle{apalike}
\bibliography{HD_competing_risks}

\end{document}